\documentclass[12pt,letterpaper,english,onecolumn,draftcls]{IEEEtran}
\usepackage[T1]{fontenc}
\usepackage[latin9]{inputenc}
\usepackage{color}
\usepackage{babel}
\usepackage{verbatim}
\usepackage{amsthm}
\usepackage{amsmath}
\usepackage{amssymb}
\usepackage{graphicx}
\usepackage[unicode=true,pdfusetitle,
 bookmarks=true,bookmarksnumbered=false,bookmarksopen=false,
 breaklinks=false,pdfborder={0 0 0},backref=false,colorlinks=false]
 {hyperref}

\makeatletter

\pdfpageheight\paperheight
\pdfpagewidth\paperwidth

\providecommand{\tabularnewline}{\\}

 \let\oldforeign@language\foreign@language
 \DeclareRobustCommand{\foreign@language}[1]{%
   \lowercase{\oldforeign@language{#1}}}
  \theoremstyle{remark}
  \newtheorem{rem}{\protect\remarkname}
  \theoremstyle{definition}
  \newtheorem{defn}{\protect\definitionname}
  \theoremstyle{plain}
  \newtheorem{thm}{\protect\theoremname}
  \theoremstyle{plain}
  \newtheorem{fact}{\protect\factname}
  \theoremstyle{plain}
  \newtheorem{lem}{\protect\lemmaname}
  \theoremstyle{plain}
  \newtheorem{cor}{\protect\corollaryname}
  \theoremstyle{definition}
  \newtheorem{example}{\protect\examplename}

\usepackage{cite}
\usepackage{times}
\usepackage{hyperref}

\makeatother

\providecommand{\corollaryname}{Corollary}
\providecommand{\definitionname}{Definition}
\providecommand{\examplename}{Example}
\providecommand{\factname}{Fact}
\providecommand{\lemmaname}{Lemma}
\providecommand{\remarkname}{Remark}
\providecommand{\theoremname}{Theorem}

\begin{document}

\title{Duality, Polite Water-filling, and Optimization for MIMO B-MAC Interference
Networks and iTree Networks}

\author{{\normalsize{An Liu, Youjian (Eugene) Liu, Haige Xiang, Wu Luo}}%
\thanks{The work was supported in part by NSFC Grant No.60972008, and in part
by US-NSF Grant CCF-0728955 and ECCS-0725915. An Liu (Email: wendaol@pku.edu.cn),
Haige Xiang, and Wu Luo are with the State Key Laboratory of Advanced
Optical Communication Systems \& Networks, School of EECS, Peking
University, China. Youjian (Eugene) Liu is with the Department of
Electrical, Computer, and Energy Engineering, University of Colorado
at Boulder, USA. The corresponding author is Wu Luo.%
}}

\maketitle
\vspace{-0.5in}

\begin{abstract}
This paper gives the long sought network version of water-filling
named as polite water-filling. Unlike in single-user MIMO channels,
where no one uses general purpose optimization algorithms in place
of the simple and optimal water-filling for transmitter optimization,
the traditional water-filling is generally far from optimal in networks
as simple as MIMO multiaccess channels (MAC) and broadcast channels
(BC), where steepest ascent algorithms have been used except for the
sum-rate optimization. This is changed by the polite water-filling
that is optimal for all boundary points of the capacity regions of
MAC and BC and for all boundary points of a set of achievable regions
of a more general class of MIMO B-MAC interference networks, which
is a combination of multiple interfering broadcast channels, from
the transmitter point of view, and multiaccess channels, from the
receiver point of view, including MAC, BC, interference channels,
X networks, and most practical wireless networks as special case.
It is polite because it strikes an optimal balance between reducing
interference to others and maximizing a link's own rate. Employing
it, the related optimizations can be vastly simplified by taking advantage
of the structure of the problems. Deeply connected to the polite water-filling,
the rate duality is extended to the forward and reverse links of the
B-MAC networks. As a demonstration, weighted sum-rate maximization
algorithms based on polite water-filling and duality with superior
performance and low complexity are designed for B-MAC networks and
are analyzed for Interference Tree (iTree) Networks, a sub-class of
the B-MAC networks that possesses promising properties for further
information theoretic study.\end{abstract}
\begin{IEEEkeywords}
Water-filling, Duality, MIMO, Interference Channel, One-hop Network,
Transmitter Optimization, Network Information Theory
\end{IEEEkeywords}

\markboth{Submitted to IEEE Transactions on Information Theory, Apr. 2010,
revised Mar. 2011}{and this is for right pages}

\newpage{}

\section{Introduction}

The transmitter optimization of networks is an important but hard
problem due to non-convexity and the lack of understanding of the
optimal input structure that can be taken advantage of. This paper
provides a polite water-filling input structure that is optimal for
MIMO multiaccess channels (MAC), broadcast channels (BC), and a set
of achievable regions of a more general class of MIMO B-MAC interference
networks, paving the way for efficient optimization algorithms.

\subsection{System Setup}

The paper considers the optimization of general one-hop multiple-input
multiple-output (MIMO) interference networks, where each multi-antenna
transmitter may send independent data to multiple receivers and each
multi-antenna receiver may collect independent data from multiple
transmitters. Consequently, the network is a combination of multiple
interfering broadcast channels (BC), from the transmitter point of
view, and multiaccess channels (MAC), from the receiver point of view,
and thus is named the Broadcast-Multiaccess Channel (B-MAC), or the
\textit{B-MAC} (\textit{Interference}) \textit{Network}. It includes
BC, MAC, interference channels \cite{Carleial_IT75_StrongIFC,Han_IT81_hankobascheme,Sato_IT81_StrongIFC},
X channels \cite{Maddah-Al_IT_MMKforX,Jafar_IT08_DOFMIMOX}, X networks
\cite{Jafar_09IT_DOFXchannel}, and most practical communication networks,
such as cellular networks, wireless LAN, cognitive radio networks,
and digital subscriber line (DSL), as special cases. Therefore, optimization
of such networks has both theoretical and practical impact.

We consider a set of achievable rate regions with the following assumptions.
1) The input distribution is circularly symmetric complex Gaussian%
\footnote{The circularly symmetric assumption will be relaxed in a future work.%
}; 2) The interference among the links are specified by any binary
coupling matrices of zero and one that can be realized by some transmit
and receive schemes. That is, after cancellation, the interference
of a signal to a link is the product of an element of the coupling
matrix, a channel matrix, and the signal; 3) Each signal is decoded
at no more than one receiver. The reasons for considering the above
setting are as follows. 1) The setting is optimal in terms of capacity
region for MAC and BC and allows us to give the optimal input structure,
polite water-filling structure, for all boundary points, which had
been an open problem since the discovery of the water-filling structure
of the sum-rate optimal point \cite{Yu_IT04_MIMO_MAC_waterfilling_alg,Jindal_IT05_IFBC,Weiyu_IT06_DualIWF};
2) The setting includes a wide range of interference management techniques
as special cases, such as a) spatial interference reduction through
beamforming matrices without interference cancellation, which becomes
the spatial interference alignment at high SNR \cite{Cadambe_IT08_DOF,Cadambe_globecom08_DistributedIA,Liu_Allerton09_Duality_distributed,Maddah-Al_IT_MMKforX};
b) some combination of interference cancellation using dirty paper
coding (DPC) \cite{Costa_IT83_Dirty_paper} at transmitters and MMSE
plus successive interference cancellation (SIC) at receivers%
\footnote{Certain combinations of DPC and SIC may result in partial cancellation
and thus, are not included.%
}; c) transmitter cooperation, where a transmitter cancel another transmitter's
signal using DPC when another transmitter's signal is available through,
e.g., an optical link between them; 3) The scenario that each signal
is decoded at no more than one receiver is useful in practice where
low complexity is desired and a receiver does not know the transmission
schemes of interference from undesired transmitters; 4) Treating interference
as noise is optimal in the weak interference regime \cite{Kramer_IT07sub_bound_Interference_channel,Motahari_IT09_IFCbound,Annapureddy_IT09_IFCweek,Kramer_Allerton08_MIMO_interference,Bandemer_Asilomar08_MISOweekIFC,Annapureddy_Asilomar08_MIMOweekIFC}
or is asymptotically optimal at high SNR in terms of degree of freedom
for some cases, such as MIMO MAC, MIMO BC, two-user MIMO interference
channel \cite{Jafar_IT07_DOFMIMOIFC}, and some of the MIMO X channels
\cite{Jafar_IT08_DOFMIMOX}; 5) Limiting to the above setting enables
us to make progress towards the optimal solution. The results and
insight serve as a stepping stone to more sophisticated design. For
example, the extension to Han-Kobayashi scheme \cite{Han_IT81_hankobascheme,TSE_IT07_Infwith1bit,Telatar_ISIT07_MIMOIFC},
where a signal may be decoded and cancelled at more than one receiver
is discussed in Remark \ref{rem:ext-Han-Kobayashi} and in future
work at Section \ref{sec:Conclusion}.

\subsection{Single-user (Selfish) Water-filling}

The optimization of B-MAC networks has been hindered by that beyond
the full solution in the 1998 seminal work \cite{Tse:98} on single
antenna MAC, little is known about the Pareto optimal input structure
of the achievable region. Thus, general purpose convex programming
is employed to find the boundary points of a MAC/BC capacity region
\cite{Viswanathan_JSAC03_BCGD}. But it is desirable to exploit the
structure of the problem to design more efficient algorithms that
will work well for both convex and non-convex optimization of networks.
For sum-rate maximization in MIMO MAC, a water-filling structure is
discovered in \cite{Yu_IT04_MIMO_MAC_waterfilling_alg}. For individual
power constraints, it results in a simple iterative water-filling
algorithm \cite{Yu_IT04_MIMO_MAC_waterfilling_alg}. Using duality,
the approach is modified for sum-rate maximization for BC \cite{Jindal_IT05_IFBC,Weiyu_IT06_DualIWF}.
The above approach works because the sum-rate happens to look like
a single-user rate and thus, cannot be used for the weighted sum-rate
maximization. The approach in \cite{Jindal_IT05_IFBC} is generalized
for weighted sum-rate maximization with a single antenna at receivers
\cite{Kobayashi_JSAC06_ITWMISOBC}, where the generalized water-filling
structure no longer has a simple water-filling interpretation. 

Directly applying single-user water-filling to networks is referred
to as selfish water-filling here. It is well known to be far from
optimal \cite{Yu_JSAC02_Distributed_power_control_DSL,Popescu_Globecom03_Water_filling_not_good,Lai_IT08_water_filling_game_MAC}
because it does not control interference to others. Based on selfish
water-filling, game-theoretic, distributed, and iterative algorithms
have been well studied for DSL \cite{Yu_JSAC02_Distributed_power_control_DSL,Cioffi_03ISIT_GaussianIFCgame,Tse_07JSAC_specturmgame,Shum_07JSAC_GaussianIFCwaterfilling,Cendrillon_TSP07_DSL,Palomar_08IT_GaussianIFCwaterfiling},
for MIMO interference channels, e.g., \textcolor{black}{\cite{Blum_ITS03_OPTsigMIMOIFC,Arslan_TWC07_ImproGameIFC,Palomar_JSAC08_MIMOIFCgame,Palomar_09TSP_MIMOIWFgame}},
and for multiaccess channels, e.g., \cite{Lai_IT08_water_filling_game_MAC}.
The algorithms converge only under stringent conditions and the performance
is not near optimal.

The importance of controlling interference to others has been recognized
in literature, e.g., \cite{Huang_06ISIT_SpectrumIFC,Wei_07ITW_MultiuserWF,Cadambe_globecom08_DistributedIA,Yu_06Globecom_DistributedSensorNet}.
But a systematic, general, and optimal method has not been found.
In\emph{ }\textit{\emph{interference pricing}} method, each user maximizes
its own utility minus the interference cost determined by the interference
prices. With a proper choice of the interference price which can be
reverse engineered from the KKT conditions, the interference pricing
based method can find a stationary point of the sum utility maximization
problem. Several monotonically convergent interference pricing algorithms
have been proposed in \cite{Berry_JSAC06_IfPriceSISO,Berry_MonoIFCpricing_ISIT09,Berry_MILCOM09_MIMOIFprice}
for the SISO/MISO interference channels, and the MIMO interference
channel with single data stream transmission. Except for the SISO
case, all these algorithms update each user's beam sequentially and
exchange interference prices after each update%
\footnote{In contrast, the polite water-filling based algorithms presented in
this paper take advantage of the structure of the optimal input, relative
to the achievable regions, take care of interference to others using
duality, and consider multiple data streams and various interference
cancellation.%
}.

Consequently, the following problems have been open. 
\begin{itemize}
\item In single-user MIMO channels, no one uses general purpose optimization
algorithms for transmitter optimization because the simple water-filling
is optimal. What is the optimal input structure for all boundary points
of the MAC/BC capacity region? Does it resemble water-filling? Can
it be used to design algorithms with much lower complexity and better
performance than general purpose optimization algorithms.
\item More ambitiously, one can ask the above questions for the B-MAC networks
with respect to the achievable region defined above. What is the optimal
method to control interference to others? Can we decompose a B-MAC
network to multiple equivalent single-user channels so that the (distributed)
optimization can be made easy? 
\end{itemize}
This paper gives an optimal network version of water-filling for the
above problems. The optimality is relative to the capacity regions
of MAB/BC and the achievable regions of the general B-MAC networks.

\subsection{SINR and Rate Duality}

In this paper, we extend the MAC-BC duality to the forward and reverse
links of the B-MAC networks. The duality is deeply connected to the
new water-filling structure and can be used to design efficient iterative
algorithms. The extension is a simple consequence of the signal-to-interference-plus-noise
power ratio (SINR) duality, which states that if a set of SINRs is
achievable in the forward links, then the same SINRs%
{} can be achieved in the reverse links when the set of transmit and
receive beamforming vectors are fixed and the roles of transmit and
receive beams are exchanged.%
{} Thus, the optimization of the transmit vectors is equivalent to the
optimization of the receive vectors in the reverse links, which has
lower complexity. The SINR duality between MAC and BC was found in
\cite{Rashid:98}. Alternating optimization based on it has been
employed in \cite{Visotski__VTC99_SIMODual,Madhow_VTC99_LOSsinrdual,Chang_TWC02_BFalgduality,Martin_ITV_04_BFdual,Gan_06_GLOBECOM_SPGP,Boche_CISS07_Weighted_sum_rate,Martin_ITS_07_MMSEduality}.
In \cite{Rao_TOC07_netduality}, the SINR duality is extended to any
one-hop MIMO networks with linear beamformers. 

The MAC-BC rate duality has been established in \cite{Goldsmith_IT03_SISO_broadcast_min_rate_power_control,Goldsmith_IT03_MIMO_broadcast_sum_cap,Tse_IT03_MIMO_broadcast,Shamai_ISIT04_Broadcast_capacity_region}.
The forward link BC and reverse link MAC have the same capacity region
and the dual input covariance matrices can be calculated by different
transformations. In \cite{Goldsmith_IT03_MIMO_broadcast_sum_cap},
the covariance transformation is derived based on the flipped channel
and achieves the same rates as the forward links with equal or less
power. The transformation cannot be calculated sequentially in general
B-MAC networks, unless there is no loops in the interference graphs
as discussed later. The MAC-BC duality can also be derived from the
SINR duality as shown in \cite{Tse_IT03_MIMO_broadcast}, where a
different covariance transformation is calculated from the MMSE receive
beams and the SINR duality. Such a transformation achieves equal or
higher rates than the forward links under the same sum power and can
be easily generalized to B-MAC networks as followed in this paper.
The above MAC-BC duality assumes sum power constraint. It can be generalized
to a single linear constraint using minimax duality and SINR duality
\cite{Zhang_IT08_MACBC_LC,Yu_IT06_Minimax_duality}. Efficient interior
point methods have been applied to solve optimizations of MAC or BC
with multiple linear constraints in \cite{Caire_09ISIT_BC_linear_constraints,Caire_09sTSP_BC_intercell_interference}.
Exploiting the polite water-filling structure is expected to produce
more efficient algorithms that also work for the more general B-MAC
networks as discussed later.

\subsection{Contributions}

The following is an overview of the contributions of this paper.
\begin{itemize}
\item \emph{Duality:} As a simple consequence of the SINR duality, we show
that the forward and reverse links of B-MAC networks have the same
achievable rate regions in Section \ref{sub:Main-Results}. The dual
input covariance matrices are obtained from a transformation based
on the MMSE filtering and the SINR duality. Unlike the covariance
transformation in \cite{Goldsmith_IT03_MIMO_broadcast_sum_cap}, which
achieves the same rate as the forward links but with equal or less
power and cannot be calculated sequentially for general B-MAC networks,
the transformation in this paper achieves equal or larger rates with
the same power and can be calculated easily for general B-MAC networks.
We show that the two transformation coincide at the Pareto rate points.
\item \emph{Polite Water-filling:} The long sought network version of water-filling
for \emph{all} the Pareto optimal input of the achievable rate region
of the B-MAC networks is found in Section \ref{sub:polite water-filling}.
Different from the traditional selfish water-filling, the polite water-filling
strikes an \emph{optimal} balance between reducing interference to
others and maximizing a link's own rate in a beautifully symmetric
form. Conceptually, it tells us that the optimal method to control
the interference to others is through optimal pre-whitening of the
channel. It offers an elegant method to decompose a network into multiple
equivalent single-user channels and thus, paves the way for designing/improving
low-complexity centralized or distributed/game-theoretic algorithms
for related optimization problems, e.g., those in \cite{Yu_JSAC02_Distributed_power_control_DSL,Popescu_Globecom03_Water_filling_not_good,Blum_ITS03_OPTsigMIMOIFC,Martin_ITV_04_BFdual,Lai_IT08_water_filling_game_MAC,Maddah-Al_IT_MMKforX,Cadambe_globecom08_DistributedIA,Xiaohu_IT09_MISOsingledet}.
The polite water-filling has the following structure. Consider link
$l$ with channel matrix $\mathbf{H}_{l,l}$ and Pareto optimal input
covariance matrix $\mathbf{\Sigma}_{l}$. The equivalent input covariance
matrix $\mathbf{Q}_{l}\triangleq\hat{\mathbf{\Omega}}_{l}^{1/2}\mathbf{\Sigma}_{l}\hat{\mathbf{\Omega}}_{l}^{1/2}$
is found to be the water-filling of the pre- and post-whitened equivalent
channel $\bar{\mathbf{H}}_{l}=\mathbf{\Omega}_{l}^{-1/2}\mathbf{H}_{l,l}\hat{\mathbf{\Omega}}_{l}^{-1/2}$,
where $\mathbf{\Omega}_{l}$ is the interference-plus-noise covariance
of the forward link, used to avoid interference from others. The physical
meaning of $\hat{\mathbf{\Omega}}_{l}$, used to control interference
to others, is two folded. In terms of the reverse link, it is the
interference-plus-noise covariance resulted from the optimal dual
input. In terms of the forward link, it is the Lagrangian penalty
for causing interference to other links. The deep connection to duality
is that the optimal Lagrange multipliers work out beautifully to be
the optimal dual reverse link powers.
\item \emph{Extension to Single Linear Constraint:} We show that all results
in the paper, including duality, polite water-filling structure, and
algorithms, can be generalized from the case of a sum power constraint
and white noise to the case of a single linear constraint and colored
noise in Section \ref{sub:extension-linear-constraints}. As discussed
in Section \ref{sec:Conclusion}, the single linear constraint result
can be extended to handle multiple linear constraints, which arise
in individual power constraints, per-antenna power constraints, and
interference reduction to primary users in cognitive radios \cite{Zhang_IT08_MACBC_LC,Caire_09ISIT_BC_linear_constraints,Caire_09sTSP_BC_intercell_interference}.
\item \emph{Weighted Sum-Rate Maximization:} In Section \ref{sec:Near optimal scheme},
highly efficient weighted sum-rate maximization algorithms are designed
to illustrate the power of the polite water-filling. 
\item \emph{iTree Networks:} The optimization for B-MAC networks is not
convex in general. To analyze the algorithms and to provide a monotonically
converging algorithm, we introduce the Interference Tree (iTree) Networks,
in Section \ref{sub:itree}. \textcolor{black}{For a fixed interference
cancellation scheme, iTree networks are B-MAC networks that do not
have any directional loops in the interference graph}s. It appears
to be a logical extension of MAC and BC. An approach to making progress
in network information theory is to study special cases such as deterministic
channels \cite{Tse_ITW07_Deterministic_model_for_relay,TSE_IT07_Infwith1bit}
and degree of freedom \cite{TSE_IT07_Infwith1bit,Jafar_IT08_DOF,Jafar_IT08_DOFMIMOX}
in order to gain insight. iTree networks looks promising in this sense.
\end{itemize}

The rest of the paper is organized as follows. Section \ref{sec:System Model}
defines the achievable rate region and summarizes the preliminaries.
Section \ref{sec:Main-Results} presents the theoretical results on
the duality and polite water-filling. As an application, polite water-filling
is applied to weighted sum-rate maximization in Section \ref{sec:Near optimal scheme},
where iTree networks is introduced for the optimality analysis. The
performance of the algorithms is verified by simulation in Section
\ref{sec:Simulation-Results}. The results of the paper provides insight
into more general problems. They are discussed in Section \ref{sec:Conclusion}
along with the conclusions.

\section{System Model and Preliminaries\label{sec:System Model}}

We first define the achievable rate region, followed by a summary
of the SINR duality in \cite{Rao_TOC07_netduality}. We use $\left\Vert \cdot\right\Vert $
for $L_{2}$ norm and $\left\Vert \cdot\right\Vert _{1}$ for $L_{1}$
norm. The complex gradient of a real function will be used extensively
and is defined as follows. Define $f(\mathbf{Z}):\mathbb{C}^{M\times N}\rightarrow\mathbb{R}$.
The extension of the results in \cite{Hjorungnes_TSP07_ComplexDiff}
gives the gradient of $f(\mathbf{Z})$ over $\mathbf{Z}$ as 
\[
\nabla_{\mathbf{Z}}f\triangleq\left(\frac{df(\mathbf{Z})}{d\mathbf{Z}}\right)^{*},
\]
where $\frac{df(\mathbf{Z})}{d\mathbf{Z}}\in\mathbb{C}^{M\times N}$
is defined as
\[
\frac{df(\mathbf{Z})}{d\mathbf{Z}}=\left[\begin{array}{ccc}
\frac{\partial f}{\partial z_{1,1}} & \cdots & \frac{\partial f}{\partial z_{1,N}}\\
\vdots &  & \vdots\\
\frac{\partial f}{\partial z_{M,1}} & \cdots & \frac{\partial f}{\partial z_{M,N}}
\end{array}\right],
\]
and $z_{i,j}=x_{i,j}+jy_{i,j}$, $\frac{\partial f}{\partial z_{i,j}}=\frac{1}{2}\frac{\partial f}{\partial x_{i,j}}-\frac{j}{2}\frac{\partial f}{\partial y_{i,j}},\:\forall i,j$.
If $\mathbf{Z}$ is Hermitian, it can be proved that the above formula
can be used without change by treating the entries in $\mathbf{Z}$
as independent variables.

\subsection{Definition of the Achievable Rate Region}

We consider the general one-hop MIMO interference networks named B-MAC
networks, where each transmitter may have independent data for different
receivers and each receiver may want independent data from different
transmitters. There are $L$ data links. Assume the set of physical
transmitter labels is $\mathcal{T}=\{\text{TX}_{1},\text{TX}_{2},\text{TX}_{3},...\}$
and the set of physical receiver labels is $\mathcal{R}=\{\text{RX}_{1},\text{RX}_{2},\text{RX}_{3},...\}$.
Define transmitter $T_{l}$ of link $l$ as a mapping from $l$ to
link $l$'s physical transmitter label in $\mathcal{T}$. Define receiver
$R_{l}$ as a mapping from $l$ to link $l$'s physical receiver label
in $\mathcal{R}$. For example, in a 2-user MAC with two links, the
sets are $\mathcal{T}=\{\text{TX}_{1},\text{TX}_{2}\}$, $\mathcal{R}=\{\text{RX}_{1}\}$.
And the mappings could be $T_{1}=\text{TX}_{1}$, $T_{2}=\text{TX}_{2}$,
$R_{1}=\text{RX}_{1}$, $R_{2}=\text{RX}_{1}$. The numbers of antennas
at $T_{l}$ and $R_{l}$ are $L_{T_{l}}$ and $L_{R_{l}}$ respectively.
The received signal at $R_{l}$ is

\begin{eqnarray}
\mathbf{y}_{l} & = & \sum_{k=1}^{L}\mathbf{H}_{l,k}\mathbf{x}_{k}+\mathbf{w}_{l},\label{eq:recvsignal}
\end{eqnarray}
where $\mathbf{x}_{k}\in\mathbb{C}^{L_{T_{k}}\times1}$ is the transmit
signal of link $k$ and is assumed to be circularly symmetric complex
Gaussian (CSCG); $\mathbf{H}_{l,k}\in\mathbb{C}^{L_{R_{l}}\times L_{T_{k}}}$
is the channel matrix between $T_{k}$ and $R_{l}$; and $\mathbf{w}_{l}\in\mathbb{C}^{L_{R_{l}}\times1}$
is a CSCG noise vector with zero mean and identity covariance matrix. 

To handle a wide range of interference cancellation, we define a coupling
matrix $\mathbf{\Phi}\in\{0,1\}^{L\times L}$ as a result of some
interference cancellation scheme. It specifies whether interference
is completely cancelled or treated as noise: if $\mathbf{x}_{k}$,
after interference cancellation, still causes interference to $\mathbf{x}_{l}$,
$\mathbf{\Phi}_{l,k}=1$ and otherwise, $\mathbf{\Phi}_{l,k}=0$. 
\begin{rem}
The coupling matrices valid for the results of this paper are those
for which there exists a transmission and receiving scheme such that
each signal is decoded (and possibly cancelled) at no more than one
receiver, because in the achievable rate region defined later, a rate
is determined by one equivalent channel. Future extension to the Han-Kobayashi
scheme where a common message may be decoded by multiple receivers
is discussed in Remark \ref{rem:ext-Han-Kobayashi}.
\end{rem}

We give examples of valid coupling matrices. For a BC (MAC) employing
DPC (SIC) where the $l^{\textrm{th}}$ link is the $l^{\textrm{th}}$
one to be encoded (decoded), the coupling matrix is given by $\mathbf{\Phi}_{l,k}=0,\forall k\leq l$
and $\mathbf{\Phi}_{l,k}=1,\forall k>l$. %
Fig. \ref{fig:sysFig1} illustrates a B-MAC network employing DPC
and SIC. The first receiver labeled by $R_{1}/R_{2}$ is the intended
receiver for $\mathbf{x}_{1}$ and $\mathbf{x}_{2}$ and the second
receiver labeled by $R_{3}$ is the intended receiver for $\mathbf{x}_{3}$.
Therefore, $\mathbf{x}_{1}$ and $\mathbf{x}_{2}$ is only decoded
at $R_{1}/R_{2}$ and $\mathbf{x}_{3}$ is only decoded at $R_{3}$.
The following $\mathbf{\Phi}^{a},\mathbf{\Phi}^{b},\mathbf{\Phi}^{c},\mathbf{\Phi}^{d}$
are valid coupling matrices for link $1,2,3$ under the corresponding
encoding and decoding orders: \emph{a}. %
$\mathbf{x}_{3}$ is encoded after $\mathbf{x}_{2}$ and $\mathbf{x}_{2}$
is decoded after $\mathbf{x}_{1}$; \emph{b}. $\mathbf{x}_{2}$ is
encoded after $\mathbf{x}_{3}$ and $\mathbf{x}_{2}$ is decoded after
$\mathbf{x}_{1}$; \emph{c}. $\mathbf{x}_{3}$ is encoded after $\mathbf{x}_{2}$
and $\mathbf{x}_{1}$ is decoded after $\mathbf{x}_{2}$; \emph{d}.
There is no interference cancellation.
\begin{align*}
\mathbf{\Phi}^{a}=\left[\begin{array}{ccc}
0 & 1 & 1\\
0 & 0 & 1\\
1 & 0 & 0
\end{array}\right], & \:\mathbf{\Phi}^{b}=\left[\begin{array}{ccc}
0 & 1 & 1\\
0 & 0 & 0\\
1 & 1 & 0
\end{array}\right],\\
\mathbf{\Phi}^{c}=\left[\begin{array}{ccc}
0 & 0 & 1\\
1 & 0 & 1\\
1 & 0 & 0
\end{array}\right], & \:\mathbf{\Phi}^{d}=\left[\begin{array}{ccc}
0 & 1 & 1\\
1 & 0 & 1\\
1 & 1 & 0
\end{array}\right].
\end{align*}
We give the details on how to obtain the coupling matrix $\mathbf{\Phi}^{c}$.
At $R_{1}/R_{2}$, because $\mathbf{x}_{1}$ is decoded after $\mathbf{x}_{2}$,
we have $\mathbf{\Phi}_{1,2}^{c}=0$ and $\mathbf{\Phi}_{2,1}^{c}=1$.
The interference between $\mathbf{x}_{1}$ and $\mathbf{x}_{3}$ can
not be cancelled by DPC or SIC and thus $\mathbf{\Phi}_{1,3}^{c}=1$
and $\mathbf{\Phi}_{3,1}^{c}=1$. Note that $R_{3}$ is not allowed
to decode $\mathbf{x}_{2}$. The interference from $\mathbf{x}_{2}$
to $\mathbf{x}_{3}$ is cancelled by DPC at the transmitter $T_{2}/T_{3}$
rather than by SIC at the receiver $R_{3}$. Therefore, we have $\mathbf{\Phi}_{2,3}^{c}=1$
and $\mathbf{\Phi}_{3,2}^{c}=0$.

\begin{figure}
\begin{centering}
\textsf{\includegraphics[clip,scale=0.3]{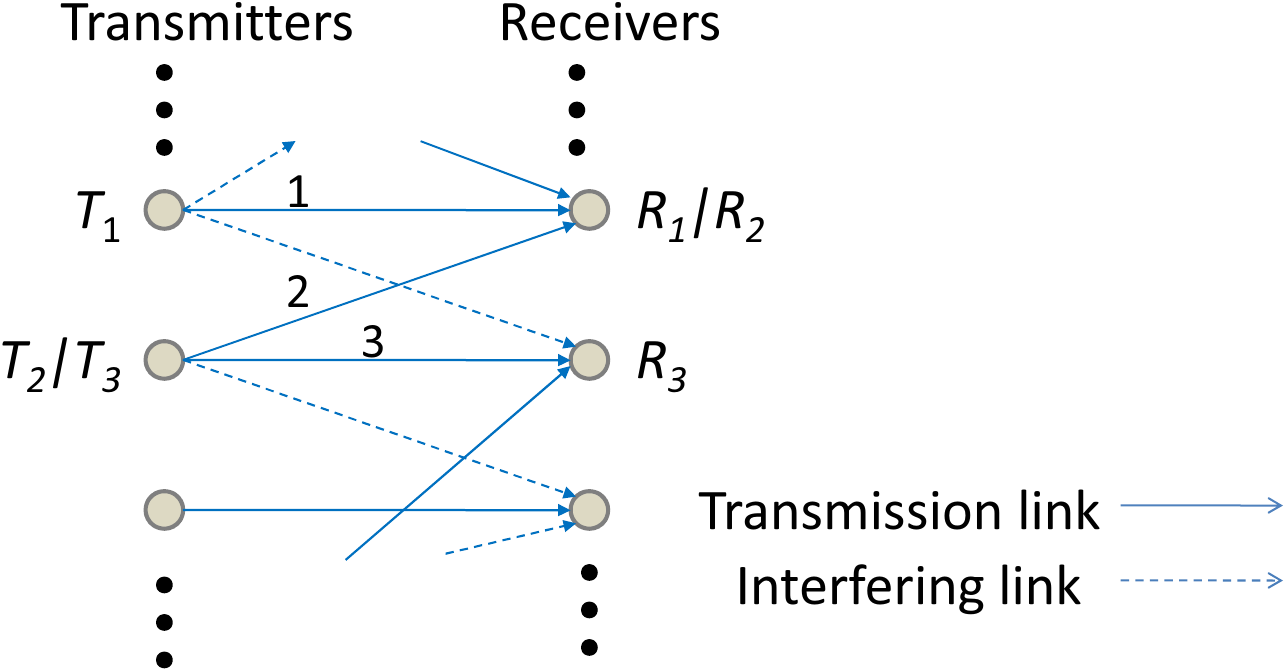}}
\par\end{centering}

\caption{\label{fig:sysFig1}Example of a B-MAC network}
\end{figure}

\begin{rem}
\label{rem:PatialSIC}Some combination of encoding and decoding orders
for DPC and SIC may result in partial interference cancellation that
can not be described by the coupling matrix of 0's and 1's. For example,
in Fig. 1, if $\mathbf{x}_{2}$ is encoded after $\mathbf{x}_{3}$
and $\mathbf{x}_{1}$ is decoded after $\mathbf{x}_{2}$, the receiver
$R_{1}/R_{2}$ can decode the data embedded in $\mathbf{x}_{2}$ but
can only reconstruct a signal that is a function of both $\mathbf{x}_{2}$
and $\mathbf{x}_{3}$ due to the subtlety in DPC, resulting in partial
cancellation \cite{Maddah-Ali_ISIT06_MMKforX,Maddah-Al_IT_MMKforX}.
One may employ the results in this paper to obtain the inner/outer
bounds of the achieved rates using pessimistic/optimistic coupling
matrices. For example, a pessimistic coupling matrix can assume none
of $\mathbf{x}_{2}$ and $\mathbf{x}_{3}$ is cancelled, and an optimistic
coupling matrix can assume both $\mathbf{x}_{2}$ and $\mathbf{x}_{3}$
are cancelled. It is an interesting future research to design coding
techniques to fully cancel the interference from both $\mathbf{x}_{2}$
and $\mathbf{x}_{3}$ to $\mathbf{x}_{1}$.
\end{rem}

The achievable rate region in this paper refers to the following.
Note that $\mathbf{\Phi}_{l,l}=0$ by definition. The interference-plus-noise
of the $l^{\text{th}}$ link is $\sum_{k=1}^{L}\mathbf{\Phi}_{l,k}\mathbf{H}_{l,k}\mathbf{x}_{k}+\mathbf{w}_{l}$,
whose covariance matrix is 
\begin{eqnarray}
\mathbf{\Omega}_{l} & = & \mathbf{I}+\sum_{k=1}^{L}\mathbf{\Phi}_{l,k}\mathbf{H}_{l,k}\mathbf{\Sigma}_{k}\mathbf{H}_{l,k}^{\dagger},\label{eq:whiteMG}
\end{eqnarray}
where $\mathbf{\Sigma}_{k}$ is the covariance matrix of $\mathbf{x}_{k}$.
We denote all the covariance matrices as $\mathbf{\Sigma}_{1:L}=\left(\mathbf{\Sigma}_{1},\mathbf{\Sigma}_{2},...,\mathbf{\Sigma}_{L}\right)$.
Then the mutual information (rate) of link $l$ is given by \cite{Telatar_EuroTrans_1999_MIMOCapacity}

\begin{eqnarray}
\mathcal{I}_{l}\left(\mathbf{\Sigma}_{1:L},\mathbf{\Phi}\right) & = & \textrm{log}\left|\mathbf{I}+\mathbf{H}_{l,l}\mathbf{\Sigma}_{l}\mathbf{H}_{l,l}^{\dagger}\mathbf{\Omega}_{l}^{-1}\right|.\label{eq:linkkMIG}
\end{eqnarray}

\begin{defn}
The\emph{ Achievable Rate Region }with \textit{\emph{a fixed coupling
matrix $\mathbf{\Phi}$}} and sum power constraint $P_{T}$ is defined
as 
\end{defn}
\begin{eqnarray}
\mathcal{R}_{\mathbf{\Phi}}\left(P_{T}\right) & \triangleq & \underset{\mathbf{\Sigma}_{1:L}:\sum_{l=1}^{L}\textrm{Tr}\left(\mathbf{\mathbf{\Sigma}}_{l}\right)\leq P_{T}}{\bigcup}\left\{ \mathbf{r}\in\mathbb{R}_{+}^{L}:\right.\label{eq:Ratereg1G}\\
 &  & \left.r_{l}\le\mathcal{I}_{l}\left(\mathbf{\Sigma}_{1:L},\mathbf{\Phi}\right),1\leq l\leq L\right\} .\nonumber 
\end{eqnarray}

Assuming the sum power constraint is the necessary first step for
more complicated cases. It has its own theoretical value and the result
can be used for individual power constraints and the more general
multiple linear constraints as discussed in Section \ref{sub:extension-linear-constraints}
and \ref{sec:Conclusion}.
\begin{defn}
If $\left[\mathcal{I}_{l}\left(\mathbf{\Sigma}_{1:L},\mathbf{\Phi}\right)\right]_{l=1:L}$
is a Pareto rate point of $\mathcal{R}_{\mathbf{\Phi}}\left(P_{T}\right)$,
the input covariance matrices $\mathbf{\Sigma}_{1:L}$ are said to
be \emph{Pareto optimal}. 
\end{defn}

A bigger achievable rate region can be defined by the convex closure
of $\bigcup_{\mathbf{\Phi}\in\Xi}\mathcal{R}_{\mathbf{\Phi}}\left(P_{T}\right)$,
where $\Xi$ is a set of valid coupling matrices. For example, if
DPC and/or SIC are employed, $\Xi$ can be a set of valid coupling
matrices corresponding to various encoding and/or decoding orders.
In this paper, we focus on a fixed \textit{\emph{coupling matrix $\mathbf{\Phi}$}},
which is the basis to study the bigger regions. The optimal \textit{\emph{coupling
matrix}} $\mathbf{\Phi}$, or equivalently, the optimal encoding and/or
decoding order of the weighted sum-rate maximization problem is partially
characterized in Section \ref{sub:WSRMP}. 

The coupling matrix setup implies successive decoding and cancellation.
If there is no interference cancellation at the transmitters, employing
joint decoding at each receiver does not enlarge the achievable regions
that we consider, where one message is decoded at no more than one
receiver. The reason is as follows. For B-MAC networks with any given
set of input covariance matrices, after isolating a MAC at a receiver
from the network, the time sharing of different decoding orders achieves
the same region as that of joint decoding \cite{Cover_book91_information_theory}.
And this is achieved without changing the interference to others because
the transmission covariance matrices are not changed. The performance
gain of joint decoding over that of SIC in general cases will be studied
in future works.

The above definition ensures the following.
\begin{thm}
\label{thm:StongPareto}For each boundary point $\mathbf{r}$ of $\mathcal{R}_{\mathbf{\Phi}}\left(P_{T}\right)$,
there is no $\mathbf{r}'\in\mathcal{R}_{\mathbf{\Phi}}\left(P_{T}\right)$
satisfying $\mathbf{r}'\ge\mathbf{r}$ and $\exists k,\ \text{s.t. }r_{k}^{'}>r_{k}$,
i.e., all boundary points of $\mathcal{R}_{\mathbf{\Phi}}\left(P_{T}\right)$
are strong Pareto optimums.
\end{thm}

The proof is given in Appendix \ref{sub:Proof-StrongPareto}.

In Section \ref{sec:Main-Results}, we establish a rate duality between
the achievable rate regions of the forward and reverse links. The
reverse links are obtained by reversing the directions of the forward
links and replacing the channel matrices by their conjugate transposes.
\textcolor{black}{The coupling matrix for the reverse links is the
transpose of that for the forward links}. We use the notation $\hat{}$
to denote the terms in the reverse links. The interference-plus-noise
covariance matrix of reverse link $l$ is 
\begin{eqnarray}
\hat{\mathbf{\Omega}}_{l} & = & \mathbf{I}+\sum_{k=1}^{L}\mathbf{\Phi}_{k,l}\mathbf{H}_{k,l}^{\dagger}\mathbf{\hat{\mathbf{\Sigma}}}_{k}\mathbf{H}_{k,l}.\label{eq:WhiteMRV}
\end{eqnarray}
And the rate of reverse link $l$ is given by $\mathcal{\hat{I}}_{l}\left(\hat{\mathbf{\Sigma}}_{1:L},\mathbf{\Phi}^{T}\right)=\textrm{log}\left|\mathbf{I}+\mathbf{H}_{l,l}^{\dagger}\hat{\mathbf{\Sigma}}_{l}\mathbf{H}_{l,l}\hat{\mathbf{\Omega}}_{l}^{-1}\right|$.
For a fixed $\mathbf{\Phi}^{T}$, the achievable rate region for the
reverse links is defined as:

\begin{eqnarray}
\mathcal{\hat{R}}_{\mathbf{\Phi}^{T}}\left(P_{T}\right) & \triangleq & \underset{\hat{\mathbf{\Sigma}}_{1:L}:\sum_{l=1}^{L}\textrm{Tr}\left(\hat{\mathbf{\Sigma}}_{l}\right)\leq P_{T}}{\bigcup}\left\{ \mathbf{\hat{r}}\in\mathbb{R}_{+}^{L}:\right.\label{eq:Ratereg2G}\\
 &  & \left.\hat{r}_{l}\le\mathcal{\hat{I}}_{l}\left(\hat{\mathbf{\Sigma}}_{1:L},\mathbf{\Phi}^{T}\right),1\leq l\leq L\right\} .\nonumber 
\end{eqnarray}

\subsection{SINR Duality for MIMO B-MAC Networks\label{sub:precoding}}

The above achievable rate region can be achieved by \textcolor{black}{the
following} spatial multiplexing scheme.
\begin{defn}
\label{def:Decomposition}The \emph{Decomposition of a MIMO Link to
Multiple SISO Data Streams} is defined as, for link $l$ and any $M_{l}\ge\text{Rank}(\mathbf{\Sigma}_{l})$\textcolor{black}{,
where the physical meaning of $M_{l}$ is the} number of SISO data
streams for link $l$, finding a precoding matrix $\mathbf{T}_{l}=\left[\sqrt{p_{l,1}}\mathbf{t}_{l,1},...,\sqrt{p_{l,M_{l}}}\mathbf{t}_{l,M_{l}}\right]$
satisfying 
\begin{equation}
\mathbf{\Sigma}_{l}=\mathbf{T}_{l}\mathbf{T}_{l}^{\dagger}=\sum_{m=1}^{M_{l}}p_{l,m}\mathbf{t}_{l,m}\mathbf{t}_{l,m}^{\dagger},\label{eq:DecomSig}
\end{equation}
where $\mathbf{t}_{l,m}\in\mathbb{C}^{L_{T_{l}}\times1}$ is a transmit
vector with $\left\Vert \mathbf{t}_{l,m}\right\Vert =1$; \textcolor{black}{and
$\mathbf{p}=\left[p_{1,1},...,p_{1,M_{1}},...,p_{L,1},...,p_{L,M_{L}}\right]^{T}$
are the transmit powers}. 
\end{defn}

Note that the precoding matrix always exists but is not unique because
$\mathbf{T}_{l}^{'}=\mathbf{T}_{l}\mathbf{V}$ with unitary $\mathbf{V}\in\mathbb{C}^{M_{l}\times M_{l}}$
also gives the same covariance matrix in (\ref{eq:DecomSig}). Without
loss of generality, we assume the intra-signal decoding order is that%
{} the $m^{\text{th}}$ stream is the $m^{\text{th}}$ to be decoded
and cancelled. The receive vector $\mathbf{r}_{l,m}\in\mathbb{C}^{L_{R_{l}}\times1}$
for the $m^{th}$ stream of link $l$ is obtained by the MMSE filtering
as 
\begin{equation}
\mathbf{r}_{l,m}=\alpha_{l,m}\left(\sum_{i=m+1}^{M_{l}}\mathbf{H}_{l,l}p_{l,i}\mathbf{t}_{l,i}\mathbf{t}_{l,i}^{\dagger}\mathbf{H}_{l,l}^{\dagger}+\mathbf{\Omega}_{l}\right)^{-1}\mathbf{H}_{l,l}\mathbf{t}_{l,m},\label{eq:MMSErev1G}
\end{equation}
where $\alpha_{l,m}$ is chosen such that $\left\Vert \mathbf{r}_{l,m}\right\Vert =1$.
This is referred to as MMSE-SIC receiver. 

For each stream, one can calculate its SINR and rate. %
For convenience, define the collections of transmit and receive vectors
as 
\begin{eqnarray}
\mathbf{T} & = & \left[\mathbf{t}_{l,m}\right]_{m=1,...,M_{l},l=1,...,L},\label{eq:udef}\\
\mathbf{R} & = & \left[\mathbf{r}_{l,m}\right]_{m=1,...,M_{l},l=1,...,L}.\label{eq:TRkdef}
\end{eqnarray}
\textcolor{black}{Define the $\left(\sum_{i=1}^{l-1}M_{i}+m\right)^{\textrm{th}}$
row and $\left(\sum_{i=1}^{k-1}M_{i}+n\right)^{\textrm{th}}$ column
of a cross-talk matrix $\mathbf{\Psi}\left(\mathbf{T},\mathbf{R}\right)\in\mathbb{R}_{+}^{\sum_{l}M_{l}\times\sum_{l}M_{l}}$
\cite{Martin_ITV_04_BFdual} as} 
\begin{align}
\mathbf{\Psi}_{l,m}^{k,n}= & \begin{cases}
0 & k=l\:\textrm{and}\: m\geq n,\\
\left|\mathbf{r}_{l,m}^{\dagger}\mathbf{H}_{l,l}\mathbf{t}_{l,n}\right|^{2} & k=l,\:\textrm{and}\: m<n,\\
\mathbf{\Phi}_{l,k}\left|\mathbf{r}_{l,m}^{\dagger}\mathbf{H}_{l,k}\mathbf{t}_{k,n}\right|^{2} & \textrm{otherwise}.
\end{cases}\label{eq:faiG}
\end{align}
 Then the SINR and the rate for the $m^{th}$ stream of link $l$
can be expressed as
\begin{eqnarray}
\gamma_{l,m}\left(\mathbf{T},\mathbf{R},\mathbf{p}\right) & = & \frac{p_{l,m}\left|\mathbf{r}_{l,m}^{\dagger}\mathbf{H}_{l,l}\mathbf{t}_{l,m}\right|^{2}}{1+{\displaystyle \sum_{k=1}^{L}}{\displaystyle \sum_{n=1}^{M_{k}}}p_{k,n}\mathbf{\Psi}_{l,m}^{k,n}},\label{eq:SINR1G}\\
r_{l,m}\left(\mathbf{T},\mathbf{R},\mathbf{p}\right) & = & \textrm{log}\left(1+\gamma_{l,m}\left(\mathbf{T},\mathbf{R},\mathbf{p}\right)\right).\label{eq:streamrateG}
\end{eqnarray}

Such decomposition of data to streams with MMSE-SIC receiver is information
lossless. 
\begin{fact}
\label{fac:MMSE-SC}The mutual information in (\ref{eq:linkkMIG})
is achieved by the MMSE-SIC receiver \cite{Varanasi_Asilomar97_MMSE_is_optimal},
i.e., it is equal to the sum-rate of all streams of link $l$:
\end{fact}
\begin{eqnarray}
r_{l}^{\textrm{s}} & \triangleq & \sum_{m=1}^{M_{l}}r_{l,m}\left(\mathbf{T},\mathbf{R},\mathbf{p}\right)\label{eq:srsu}\\
 & = & \mathcal{I}_{l}\left(\mathbf{\Sigma}_{1:L},\mathbf{\Phi}\right).\nonumber 
\end{eqnarray}

In the reverse links, we can obtain SINRs using $\mathbf{R}$ as transmit
vectors and $\mathbf{T}$ as receive vectors. The transmit powers
is denoted as $\mathbf{q}=\left[q_{1,1},...,q_{1,M_{1}},...,q_{L,1},...,q_{L,M_{L}}\right]^{T}$.
The decoding order of the streams within a link is the opposite to
that of the forward link, i.e., the $m^{\text{th}}$ stream is the
$m^{\text{th}}$ last to be decoded and cancelled. Then the SINR for
the $m^{th}$ stream of reverse link $l$ is
\begin{eqnarray}
\hat{\gamma}_{l,m}\left(\mathbf{R},\mathbf{T},\mathbf{q}\right) & = & \frac{q_{l,m}\left|\mathbf{t}_{l,m}^{\dagger}\mathbf{H}_{l,l}^{\dagger}\mathbf{r}_{l,m}\right|^{2}}{1+{\displaystyle \sum_{k=1}^{L}}{\displaystyle \sum_{n=1}^{M_{k}}}q_{k,n}\mathbf{\Psi}_{k,n}^{l,m}}.\label{eq:SINR2G}
\end{eqnarray}
For simplicity, we will use $\left\{ \mathbf{T},\mathbf{R},\mathbf{p}\right\} $
($\left\{ \mathbf{R},\mathbf{T},\mathbf{q}\right\} $) to denote the
transmission and reception strategy described above in the forward
(reverse) links. 

The achievable SINR regions of the forward and reverse links are the
same. %
{} Define the achievable SINR regions $\mathcal{T}_{\mathbf{\Phi}}\left(P_{T}\right)$
and $\mathcal{\hat{T}}_{\mathbf{\Phi}^{T}}\left(P_{T}\right)$ as
the set of all SINRs that can be achieved under the sum power constraint
$P_{T}$ in the forward and reverse links respectively. For given
set of SINR values $\mathbf{\gamma}^{0}=\left[\gamma_{l,m}^{0}\right]_{m=1,...,M_{l},l=1,...,L}$,
define a diagonal matrix $\mathbf{D}\left(\mathbf{T},\mathbf{R},\mathbf{\gamma}^{0}\right)\in\mathbb{R}_{+}^{\sum_{l}M_{l}\times\sum_{l}M_{l}}$
being a function of $\mathbf{T},\mathbf{R}$ and $\mathbf{\gamma}^{0}$,
where the $\left(\sum_{i=1}^{l-1}M_{i}+m\right)^{\text{th}}$ diagonal
element is given by 

\begin{eqnarray}
\mathbf{D}_{\sum_{i=1}^{l-1}M_{i}+m,\sum_{i=1}^{l-1}M_{i}+m} & = & \gamma_{l,m}^{0}/\left|\mathbf{r}_{l,m}^{\dagger}\mathbf{H}_{l,l}\mathbf{t}_{l,m}\right|^{2}.\label{eq:DG}
\end{eqnarray}
\textcolor{black}{We restate the SINR duality, e.g. \cite{Rao_TOC07_netduality},
as follows.}
\begin{lem}
\label{lem:lem1G}If a set of SINRs $\mathbf{\gamma}^{0}$ is achieved
by the transmission and reception strategy $\left\{ \mathbf{T},\mathbf{R},\mathbf{p}\right\} $
with $\left\Vert \mathbf{p}\right\Vert _{1}=P_{T}$ in the forward
links, then $\mathbf{\gamma}^{0}$ is also achievable in the reverse
links with $\left\{ \mathbf{R},\mathbf{T},\mathbf{q}\right\} $, where
the reverse transmit power vector $\mathbf{q}$ satisfies $\left\Vert \mathbf{q}\right\Vert _{1}=P_{T}$
and is given by
\begin{eqnarray}
\mathbf{q} & = & \left(\mathbf{D}^{-1}\left(\mathbf{T},\mathbf{R},\mathbf{\gamma}^{0}\right)-\mathbf{\Psi}^{T}\left(\mathbf{T},\mathbf{R}\right)\right)^{-1}\mathbf{1}.\label{eq:qpower}
\end{eqnarray}
And thus, one has $\mathcal{T}_{\mathbf{\Phi}}\left(P_{T}\right)=\mathcal{\hat{T}}_{\mathbf{\Phi}^{T}}\left(P_{T}\right)$.
\end{lem}

\section{Theory\label{sec:Main-Results}}

In this section, we first establish a rate duality as a simple consequence
of the SINR duality in \cite{Rao_TOC07_netduality}. Then the polite
water-filling structure and properties of the Pareto optimal input
are characterized. Finally, we discuss the extension from the sum
power constraint and white noise to a single linear constraint and
colored noise.

\subsection{A Rate Duality for B-MAC Networks\label{sub:Main-Results}}

The rate duality is a simple consequence of the SINR duality.
\begin{thm}
\label{thm:mainthem}The achievable rate regions of the forward and
reverse links of a B-MAC network defined in (\ref{eq:Ratereg1G})
and (\ref{eq:Ratereg2G}) respectively are the same, i.e., $\mathcal{R}_{\mathbf{\Phi}}\left(P_{T}\right)=\hat{\mathcal{R}}_{\mathbf{\Phi}^{T}}\left(P_{T}\right)$.\end{thm}
\begin{IEEEproof}
For any rate point $\mathbf{r}$ in the region $\mathcal{R}_{\mathbf{\Phi}}\left(P_{T}\right)$
achieved by the input covariance matrices $\mathbf{\Sigma}_{1:L}$,
the covariance transformation defined below can be used to obtain
$\hat{\mathbf{\Sigma}}_{1:L}$ such that a reverse link rate point
$\mathbf{\hat{r}}\geq\mathbf{r}$ under the same sum power constraint
$P_{T}$ can be achieved, according to Lemma \ref{lem:mainlem}. The
same is true for the reverse links. Therefore, we have $\mathcal{R}_{\mathbf{\Phi}}\left(P_{T}\right)=\hat{\mathcal{R}}_{\mathbf{\Phi}^{T}}\left(P_{T}\right)$. 
\end{IEEEproof}

\begin{defn}
\label{def:CovT}Let $\mathbf{\Sigma}_{l}=\sum_{m=1}^{M_{l}}p_{l,m}\mathbf{t}_{l,m}\mathbf{t}_{l,m}^{\dagger},l=1,...,L$
be a decomposition of $\mathbf{\Sigma}_{1:L}$. Compute the MMSE-SIC
receive vectors $\mathbf{R}$ from (\ref{eq:MMSErev1G}) and the reverse
transmit powers $\mathbf{q}$ from (\ref{eq:qpower}). The \emph{Covariance
Transformation} from $\mathbf{\Sigma}_{1:L}$ to $\hat{\mathbf{\Sigma}}_{1:L}$
is defined as 

\begin{eqnarray}
\hat{\mathbf{\Sigma}}_{l} & = & \sum_{m=1}^{M_{l}}q_{l,m}\mathbf{r}_{l,m}\mathbf{r}_{l,m}^{\dagger},l=1,...,L.\label{eq:CovTrans}
\end{eqnarray}

\end{defn}

\begin{lem}
\label{lem:mainlem}%
For any input covariance matrices $\mathbf{\Sigma}_{1:L}$ satisfying
the sum power constraint $P_{T}$ and achieving a rate point $\mathbf{r}\in\mathcal{R}_{\mathbf{\Phi}}\left(P_{T}\right)$,
its covariance transformation $\hat{\mathbf{\Sigma}}_{1:L}$ achieves
a rate point $\mathbf{\hat{r}}\geq\mathbf{r}$ in the reverse links
under the same sum power constraint. \end{lem}
\begin{IEEEproof}
According to fact \ref{fac:MMSE-SC}, $\mathbf{r}$ is achieved by
the transmission and reception strategy $\left\{ \mathbf{T},\mathbf{R},\mathbf{p}\right\} $
with $\left\Vert \mathbf{p}\right\Vert _{1}=\sum_{l=1}^{L}\textrm{Tr}\left(\mathbf{\Sigma}_{l}\right)$.
It follows from Lemma%
{} \ref{lem:lem1G} that $\left\{ \mathbf{R},\mathbf{T},\mathbf{q}\right\} $
with $\left\Vert \mathbf{q}\right\Vert _{1}=\sum_{l=1}^{L}\textrm{Tr}\left(\mathbf{\Sigma}_{l}\right)$
can also achieve the same rate point $\mathbf{r}$ in the reverse
links. Because $\mathbf{T}$ may not be the optimal MMSE-SIC receive
vectors, the reverse link rates may be improved with a better receiver
to obtain $\mathbf{\hat{r}}\geq\mathbf{r}$.
\end{IEEEproof}

The following corollary follows immediately from Lemma \ref{lem:mainlem}
and Theorem \ref{thm:mainthem}.
\begin{cor}
\label{cor:simplematrix}For any input covariance matrices $\mathbf{\Sigma}_{1:L}$
achieving a Pareto rate point, its covariance transformation $\hat{\mathbf{\Sigma}}_{1:L}$
achieves the same Pareto rate point in the reverse links.
\end{cor}

The following makes connection of the covariance transformation (\ref{eq:CovTrans})
to the existing ones. First, it is essentially the same as the MAC-BC
transformations in \cite{Zhang_IT08_MACBC_LC} and \cite{Joham_Globe08_Gduality}.
The difference is that we do not specify the decomposition of $\mathbf{\Sigma}_{l}$
in (\ref{eq:DecomSig}), while two particular decompositions are used
in \cite{Zhang_IT08_MACBC_LC} and \cite{Joham_Globe08_Gduality}.
Second, the MAC-BC transformation in \cite{Goldsmith_IT03_MIMO_broadcast_sum_cap}
is equivalent to the covariance transformation in (\ref{eq:CovTrans})
at the Pareto rate point. 
\begin{thm}
\label{thm:contMACBC}For any input covariance matrices $\mathbf{\Sigma}_{1:L}$
achieving a Pareto rate point, its covariance transformation $\hat{\mathbf{\Sigma}}_{1:L}$
produced by (\ref{eq:CovTrans}) is the same as that produced by the
MAC-BC transformation in \cite{Goldsmith_IT03_MIMO_broadcast_sum_cap},
i.e., 
\begin{equation}
\hat{\mathbf{\Sigma}}_{l}=\mathbf{\Omega}_{l}^{-1/2}\mathbf{F}_{l}\mathbf{G}_{l}^{\dagger}\hat{\mathbf{\Omega}}_{l}^{1/2}\mathbf{\Sigma}_{l}\hat{\mathbf{\Omega}}_{l}^{1/2}\mathbf{G}_{l}\mathbf{F}_{l}^{\dagger}\mathbf{\Omega}_{l}^{-1/2},\: l=1,...,L,\label{eq:smithequ}
\end{equation}
where $\mathbf{F}_{l}\in\mathbb{C}^{L_{R_{l}}\times N_{l}},\ \mathbf{G}_{l}\in\mathbb{C}^{L_{T_{l}}\times N_{l}}$
are obtained by the thin singular value decomposition (SVD): $\mathbf{\Omega}_{l}^{-1/2}\mathbf{H}_{l,l}\hat{\mathbf{\Omega}}_{l}^{-1/2}=\mathbf{F}_{l}\mathbf{\Delta}_{l}\mathbf{G}_{l}^{\dagger}$;
and $N_{l}$ is the rank of $\mathbf{H}_{l,l}$. \end{thm}
\begin{IEEEproof}
The proof relies on Theorem \ref{thm:WFST} in Section \ref{sub:polite water-filling},
which states that at the Pareto rate point, both $\mathbf{\Sigma}_{1:L}$
and its covariance transformation $\hat{\mathbf{\Sigma}}_{1:L}$ satisfy
the polite water-filling structure as in (\ref{eq:WFFar}) and (\ref{eq:WFrev})
respectively. Then we have 
\begin{eqnarray*}
\mathbf{\Omega}_{l}^{1/2}\hat{\mathbf{\Sigma}}_{l}\mathbf{\Omega}_{l}^{1/2} & = & \mathbf{F}_{l}\mathbf{D}_{l}\mathbf{F}_{l}^{\dagger}\\
 & = & \mathbf{F}_{l}\mathbf{G}_{l}^{\dagger}\mathbf{G}_{l}\mathbf{D}_{l}\mathbf{G}_{l}^{\dagger}\mathbf{G}_{l}\mathbf{F}_{l}^{\dagger}\\
 & = & \mathbf{F}_{l}\mathbf{G}_{l}^{\dagger}\hat{\mathbf{\Omega}}_{l}^{1/2}\mathbf{\Sigma}_{l}\hat{\mathbf{\Omega}}_{l}^{1/2}\mathbf{G}_{l}\mathbf{F}_{l}^{\dagger}.
\end{eqnarray*}
where the first and last equations follow from (\ref{eq:WFFar}) and
(\ref{eq:WFrev}) respectively.\end{IEEEproof}
\begin{rem}
The above two transformations are different at the inner point. The
transformation in (\ref{eq:CovTrans}) keeps the sum transmit power
unchanged while not decreasing the achievable rate, while the MAC-BC
transformation in \cite{Goldsmith_IT03_MIMO_broadcast_sum_cap} keeps
the achievable rate the same while not increasing the sum transmit
power. Furthermore, the extension of MAC-BC transformation in \cite{Goldsmith_IT03_MIMO_broadcast_sum_cap}
to B-MAC networks can only be sequentially calculated for iTree networks
defined in Section \ref{sub:itree}.
\end{rem}

\subsection{Characterization of the Pareto Optimal Input\label{sub:polite water-filling}}

We show that the Pareto optimal input covariance matrices have a \textit{polite
water-filling structure} defined below. It generalizes the well known
single-user water-filling solution to networks.
\begin{defn}
\label{def:DefGWF}Given input covariance matrices $\mathbf{\Sigma}_{1:L}$,
obtain its covariance transformation $\hat{\mathbf{\Sigma}}_{1:L}$
by (\ref{eq:CovTrans}) and calculate the interference-plus-noise
covariance matrices \textcolor{black}{$\mathbf{\Omega}_{1:L}$ and
}$\hat{\mathbf{\Omega}}_{1:L}$. For each link $l$, pre- and post-
whiten the channel $\mathbf{H}_{l,l}$ to produce an equivalent single-user
channel $\bar{\mathbf{H}}_{l}=\mathbf{\Omega}_{l}^{-1/2}\mathbf{H}_{l,l}\hat{\mathbf{\Omega}}_{l}^{-1/2}$.
Define $\mathbf{Q}_{l}\triangleq\hat{\mathbf{\Omega}}_{l}^{1/2}\mathbf{\Sigma}_{l}\hat{\mathbf{\Omega}}_{l}^{1/2}$
as the equivalent input covariance matrix of the link $l$. The input
covariance matrix $\mathbf{\Sigma}_{l}$ is said to possess a \emph{polite
water-filling structure} if $\mathbf{Q}_{l}$ is a water-filling over
$\bar{\mathbf{H}}_{l}$, i.e., 
\begin{eqnarray}
\mathbf{Q}_{l} & = & \mathbf{G}_{l}\mathbf{D}_{l}\mathbf{G}_{l}^{\dagger},\label{eq:DefWFfar}\\
\mathbf{D}_{l} & = & \left(\nu_{l}\mathbf{I}-\mathbf{\Delta}_{l}^{-2}\right)^{+},\nonumber 
\end{eqnarray}
where $\nu_{l}\geq0$ is the \textit{\textcolor{black}{polite water-filling
level}}; the equivalent channel \textcolor{black}{$\bar{\mathbf{H}}_{l}$'s
thin singular value decomposition (SVD) is} $\bar{\mathbf{H}}_{l}=\mathbf{F}_{l}\mathbf{\Delta}_{l}\mathbf{G}_{l}^{\dagger}$,
where $\mathbf{F}_{l}\in\mathbb{C}^{L_{R_{l}}\times N_{l}},\ \mathbf{G}_{l}\in\mathbb{C}^{L_{T_{l}}\times N_{l}},\ \mathbf{\Delta}_{l}\in\mathbb{R}_{++}^{N_{l}\times N_{l}}$;
and \textcolor{black}{$N_{l}=\textrm{Rank}\left(\mathbf{H}_{l,l}\right)$}.
$\mathbf{\Sigma}_{1:L}$ is said to possess the \emph{polite water-filling
structure} if all $\mathbf{\Sigma}_{l}$'s do.
\end{defn}

For B-MAC, the polite water-filling structure is proved by observing
that the Pareto optimal input covariance matrix for each link is the
solution of some single-user optimization problem. We use the notation
$\tilde{}$ to denote the Pareto optimal variables. Without loss of
generality, assume that $\mathbf{\tilde{\mathbf{\mathbf{\Sigma}}}}_{l}=\sum_{m=1}^{M_{l}}\tilde{p}_{l,m}\tilde{\mathbf{t}}_{l,m}\mathbf{\tilde{t}}_{l,m}^{\dagger},l=1,...,L$,
where $\tilde{p}_{l,m}>0,\:\forall l,m$, achieves a Pareto rate point
$\left[\mathcal{\tilde{I}}_{l}>0\right]_{l=1,...,L}$ and $\tilde{\mathbf{\hat{\mathbf{\Sigma}}}}_{l}=\sum_{m=1}^{M_{l}}\tilde{q}_{l,m}\mathbf{\tilde{r}}_{l,m}\mathbf{\tilde{r}}_{l,m}^{\dagger},l=1,...,L$
are the corresponding covariance transformation. The corresponding
interference-plus-noise covariance matrices are denoted as $\tilde{\mathbf{\Omega}}_{l}$
and $\tilde{\hat{\mathbf{\Omega}}}_{l}$. Then it can be proved by
contradiction that $\mathbf{\tilde{\mathbf{\Sigma}}}_{l}$ is the
solution of the following single-user optimization problem for link
$l$, where the transmission and reception schemes of other links
are fixed as the Pareto optimal scheme, $\{\tilde{p}_{k,m},\tilde{\mathbf{t}}_{k,m},\mathbf{\tilde{r}}_{k,m},k\neq l\}$.
\begin{eqnarray}
 & \underset{\mathbf{\Sigma}_{l}\succeq0}{\textrm{min}}\textrm{Tr}\left(\mathbf{\Sigma}_{l}\right)\label{eq:minpowl}\\
\textrm{s.t.} & \textrm{log}\left|\mathbf{I}+\mathbf{H}_{l,l}\mathbf{\Sigma}_{l}\mathbf{H}_{l,l}^{\dagger}\tilde{\mathbf{\Omega}}_{l}^{-1}\right|\geq\mathcal{\tilde{I}}_{l}\nonumber \\
 & \textrm{Tr}\left(\mathbf{\Sigma}_{l}\mathbf{A}_{k,m}^{(l)}\right)\leq\textrm{Tr}\left(\mathbf{\tilde{\mathbf{\Sigma}}}_{l}\mathbf{A}_{k,m}^{(l)}\right)\label{eq:INCinmip}\\
 & m=1,...,M_{k},k=1,...,L,k\neq l,\nonumber 
\end{eqnarray}
where $\mathbf{A}_{k,m}^{(l)}=\mathbf{\Phi}_{k,l}\mathbf{H}_{k,l}^{\dagger}\tilde{\mathbf{r}}_{k,m}\mathbf{\tilde{r}}_{k,m}^{\dagger}\mathbf{H}_{k,l}$;
and $\textrm{Tr}\left(\mathbf{\Sigma}_{l}\mathbf{A}_{k,m}^{(l)}\right)$
is the interference from link $l$ to the $m^{th}$ stream of link
$k$ and is constrained not to exceed the optimal value $\textrm{Tr}\left(\mathbf{\tilde{\mathbf{\Sigma}}}_{l}\mathbf{A}_{k,m}^{(l)}\right)$.
The constraints force the rates to be the Pareto point while the power
of link $l$ is minimized. The Lagrangian of problem (\ref{eq:minpowl})
is
\begin{align}
 & L_{l}\left(\mathbf{\lambda},\nu_{l},\mathbf{\Theta},\mathbf{\Sigma}_{l}\right)\nonumber \\
= & \textrm{Tr}\left(\mathbf{\Sigma}_{l}\left(\mathbf{A}_{l}\left(\mathbf{\lambda}\right)-\mathbf{\Theta}\right)\right)-\sum_{k\neq l}\sum_{m=1}^{M_{k}}\lambda_{k,m}\textrm{Tr}\left(\mathbf{\tilde{\mathbf{\Sigma}}}_{l}\mathbf{A}_{k,m}^{(l)}\right)\nonumber \\
 & +\nu_{l}\mathcal{\tilde{I}}_{l}-\nu_{l}\textrm{log}\left|\mathbf{I}+\mathbf{H}_{l,l}\mathbf{\Sigma}_{l}\mathbf{H}_{l,l}^{\dagger}\tilde{\mathbf{\Omega}}_{l}^{-1}\right|,\label{eq:Lag_minp}
\end{align}
where the dual variables $\nu_{l}\in\mathbb{R}_{+}$ and $\mathbf{\lambda}=\left[\lambda_{k,m}\right]_{m=1,...,M_{k},k\ne l}\in\mathbb{R}_{+}^{\sum_{k\neq l}M_{k}\times1}$
are associated with the rate constraint and the interference constraints
in (\ref{eq:INCinmip}) respectively; $\mathbf{A}_{l}\left(\mathbf{\lambda}\right)\triangleq\mathbf{I}+\sum_{k\neq l}\sum_{m=1}^{M_{k}}\lambda_{k,m}\mathbf{A}_{k,m}^{(l)}$
is a function of $\mathbf{\lambda}$; $\mathbf{\Theta}$ is the matrix
dual variables associated with the positive semidefiniteness constraint
on $\mathbf{\Sigma}_{l}$. Because problem (\ref{eq:minpowl}) is
convex, the duality gap is zero \cite{Boyd_04Book_Convex_optimization}
and $\mathbf{\tilde{\mathbf{\Sigma}}}_{l}$ minimizes the Lagrangian
(\ref{eq:Lag_minp}) with optimal dual variables $\tilde{\nu}_{l}$
and $\mathbf{\tilde{\lambda}}=\left[\mathbf{\tilde{\lambda}}_{k,m}\right]_{m=1,...,M_{k},k\ne l}$,
i.e., it is the solution of 
\begin{eqnarray}
\underset{\mathbf{\Sigma}_{l}\succeq0}{\textrm{min}} & \textrm{Tr}\left(\mathbf{\Sigma}_{l}\mathbf{A}_{l}\left(\mathbf{\tilde{\lambda}}\right)\right)-\tilde{\nu}_{l}\textrm{log}\left|\mathbf{I}+\mathbf{H}_{l,l}\mathbf{\Sigma}_{l}\mathbf{H}_{l,l}^{\dagger}\tilde{\mathbf{\Omega}}_{l}^{-1}\right|.\label{eq:PRTP1}
\end{eqnarray}

Note that in problem (\ref{eq:PRTP1}), the constant terms in the
Lagrangian (\ref{eq:Lag_minp}) have been deleted, and the term $\textrm{Tr}\left(\mathbf{\Sigma}_{l}\mathbf{\Theta}\right)$
is explicitly handled by adding the constraint $\mathbf{\Sigma}_{l}\succeq0$.

The following theorem proved in Appendix \ref{sub:ProofWFST} states
that the physical meaning of the optimal dual variables $\left[\mathbf{\tilde{\lambda}}_{k,m}\right]_{k\ne l}$
is exactly the optimal power allocation in the reverse links. 
\begin{thm}
\label{thm:optdual}The optimal dual variables of problem (\ref{eq:minpowl})
are given by 
\begin{eqnarray}
\tilde{\lambda}_{k,m} & = & \tilde{q}_{k,m},m=1,...,M_{k},\forall k\neq l\label{eq:optidual}\\
\tilde{\nu}_{l} & = & \frac{\tilde{q}_{l,m}\tilde{p}_{l,m}\left(1+\tilde{\mathbf{\gamma}}_{l,m}\right)\left|\tilde{\mathbf{r}}_{l,m}^{\dagger}\mathbf{H}_{l,l}\tilde{\mathbf{t}}_{l,m}\right|^{2}}{\tilde{\mathbf{\gamma}}_{l,m}^{2}},\forall m\label{eq:optidual2}
\end{eqnarray}
where $\tilde{\mathbf{\gamma}}_{l,m}$ is the SINR of the $m^{th}$
stream of link $l$ achieved by $\{\tilde{p}_{k,m},\tilde{\mathbf{t}}_{k,m},\mathbf{\tilde{r}}_{k,m}\}$.
Therefore, 
\begin{eqnarray*}
\mathbf{A}_{l}\left(\mathbf{\tilde{\lambda}}\right)=\mathbf{A}_{l}\left(\left[\tilde{q}_{k,m}\right]_{k\ne l}\right) & = & \tilde{\hat{\mathbf{\Omega}}}_{l}.
\end{eqnarray*}

\end{thm}

The polite water-filling structure is shown by a single-user-channel
view using the above results. Let $\bar{\mathbf{H}}_{l}=\tilde{\mathbf{\Omega}}_{l}^{-1/2}\mathbf{H}_{l,l}\tilde{\hat{\mathbf{\Omega}}}_{l}^{-1/2}$
and $\mathbf{Q}_{l}=\tilde{\mathbf{\hat{\mathbf{\Omega}}}}_{l}^{1/2}\mathbf{\Sigma}_{l}\tilde{\hat{\mathbf{\Omega}}}_{l}^{1/2}$.
Since $\tilde{\hat{\mathbf{\Omega}}}_{l}$ is non-singular, problem
(\ref{eq:PRTP1}) is equivalent to a single-user optimization problem
\begin{eqnarray}
\underset{\mathbf{Q}_{l}\succeq0}{\textrm{min}} & \textrm{Tr}\left(\mathbf{Q}_{l}\right)-\tilde{\nu}_{l}\textrm{log}\left|\mathbf{I}+\bar{\mathbf{H}}_{l}\mathbf{Q}_{l}\bar{\mathbf{H}}_{l}^{\dagger}\right|,\label{eq:PRTOP2}
\end{eqnarray}
of which $\tilde{\hat{\mathbf{\Omega}}}_{l}^{1/2}\mathbf{\tilde{\mathbf{\Sigma}}}_{l}\tilde{\hat{\mathbf{\Omega}}}_{l}^{1/2}$
is an optimal solution. Since the optimal solution to problem (\ref{eq:PRTOP2})
is unique and is given by the water-filling over $\bar{\mathbf{H}}_{l}$
\cite{Telatar_EuroTrans_1999_MIMOCapacity,Jindal_IT05_IFBC}, the
following theorem is proved.
\begin{thm}
\label{thm:WFST}For each $l$, perform the thin SVD as $\bar{\mathbf{H}}_{l}=\mathbf{F}_{l}\mathbf{\Delta}_{l}\mathbf{G}_{l}^{\dagger}$.
At a Pareto rate point, the input covariance matrix $\mathbf{\tilde{\mathbf{\Sigma}}}_{l}$
must have a \emph{polite water-filling structure}, i.e., the equivalent
input covariance matrix $\mathbf{\tilde{Q}}_{l}\triangleq\tilde{\hat{\mathbf{\Omega}}}_{l}^{1/2}\mathbf{\tilde{\mathbf{\Sigma}}}_{l}\tilde{\hat{\mathbf{\Omega}}}_{l}^{1/2}$
satisfies 
\begin{eqnarray}
\mathbf{\tilde{Q}}_{l} & = & \mathbf{G}_{l}\mathbf{D}_{l}\mathbf{G}_{l}^{\dagger},\label{eq:WFFar}\\
\mathbf{D}_{l} & = & \left(\tilde{\nu}_{l}\mathbf{I}-\mathbf{\Delta}_{l}^{-2}\right)^{+}.\nonumber 
\end{eqnarray}
Similarly, the corresponding $\tilde{\mathbf{\hat{\mathbf{\Sigma}}}}_{l}$
produces $\tilde{\hat{\mathbf{Q}}}_{l}\triangleq\tilde{\mathbf{\Omega}}_{l}^{1/2}\tilde{\mathbf{\hat{\mathbf{\Sigma}}}}_{l}\tilde{\mathbf{\Omega}}_{l}^{1/2}$,
which satisfies
\begin{eqnarray}
\tilde{\hat{\mathbf{Q}}}_{l} & = & \mathbf{F}_{l}\mathbf{D}_{l}\mathbf{F}_{l}^{\dagger}.\label{eq:WFrev}
\end{eqnarray}
\end{thm}
\begin{rem}
The insight given by the above proof is that restricting interference
to other links can be achieved by pre-whitening the channel with reverse
link interference-plus-noise covariance matrix. And thus, the B-MAC
can be converted to virtually independent equivalent channels $\bar{\mathbf{H}}_{l},\ l=1,...,L$.
The restriction of interference is achieved in two steps. First, in
$\bar{\mathbf{H}}_{l}=\tilde{\mathbf{\Omega}}_{l}^{-1/2}\mathbf{H}_{l,l}\tilde{\hat{\mathbf{\Omega}}}_{l}^{-1/2}$,
the multiplication of $\tilde{\hat{\mathbf{\Omega}}}_{l}^{-1/2}$
reduces the channel gain in the interfering directions so that in
$\mathbf{\tilde{Q}}_{l}$, less power will be filled in these directions.
Second, in $\mathbf{\tilde{\mathbf{\Sigma}}}_{l}=\tilde{\hat{\mathbf{\Omega}}}_{l}^{-1/2}\mathbf{\tilde{Q}}_{l}\tilde{\hat{\mathbf{\Omega}}}_{l}^{-1/2}$,
the power to the interfering directions is further reduced by the
multiplication of $\tilde{\hat{\mathbf{\Omega}}}_{l}^{-1/2}$.
\end{rem}

\begin{rem}
\label{rem:ext-Han-Kobayashi}The Lagrangian interpretation of $\tilde{\hat{\mathbf{\Omega}}}_{l}$
makes it possible to extend the duality and polite water-filling to
Han-Kobayashi transmission scheme. Cancelling the interference requires
the interference power to be greater than a threshold rather than
less than it. Therefore, some Lagrange multipliers are negative in
$\mathbf{A}_{l}\left(\mathbf{\lambda}\right)$. If we still interpret
the Lagrange multiplier as reverse link power, we must introduce the
concept of negative power for the duality and polite water-filling
for Han-Kobayashi scheme. The matrix $\tilde{\hat{\mathbf{\Omega}}}_{l}$
likely remains positive definite. Otherwise, the solution to problem
(\ref{eq:PRTP1}) has infinite power, which suggests there is no feasible
power allocation to satisfy all the constraints.
\end{rem}

Theorem \ref{thm:WFST} says that at the Pareto rate point, it is
necessary that $\mathbf{\Sigma}_{1:L}$ and $\hat{\mathbf{\Sigma}}_{1:L}$
have the polite water-filling structure. The following theorem states
that $\mathbf{\Sigma}_{l}$ having the polite water-filling structure
suffices for $\hat{\mathbf{\Sigma}}_{l}$ to have the polite water-filling
structure even at a non-Pareto rate point. This enables the optimization
of the network part by part. A lemma is needed for the proof and reveals
more insight to the duality. Although the covariance transformation
preserves total power such that $\sum_{l=1}^{L}\textrm{Tr}\left(\mathbf{\mathbf{\Sigma}}_{l}\right)=\sum_{l=1}^{L}\textrm{Tr}\left(\hat{\mathbf{\mathbf{\Sigma}}}_{l}\right)$,
in general, $\textrm{Tr}\left(\mathbf{\mathbf{\Sigma}}_{l}\right)=\textrm{Tr}\left(\hat{\mathbf{\mathbf{\Sigma}}}_{l}\right)$
is not true. Surprisingly, $\textrm{Tr}\left(\mathbf{\mathbf{Q}}_{l}\right)=\textrm{Tr}\left(\hat{\mathbf{\mathbf{Q}}}_{l}\right),\ \forall l$
is true as stated in the following lemma proved in Appendix \ref{sub:Proof-of-TrQ}.
\begin{lem}
\label{lem:TrQ} Let $\hat{\mathbf{\Sigma}}_{1:L}$ be the covariance
transformation of $\mathbf{\Sigma}_{1:L}$. Define the equivalent
covariance matrices $\mathbf{Q}_{l}\triangleq\hat{\mathbf{\Omega}}_{l}^{1/2}\mathbf{\Sigma}_{l}\hat{\mathbf{\Omega}}_{l}^{1/2}$
and $\mathbf{\hat{Q}}_{l}\triangleq\mathbf{\Omega}_{l}^{1/2}\hat{\mathbf{\Sigma}}_{l}\mathbf{\Omega}_{l}^{1/2}$.
The power of the forward and reverse link equivalent covariance matrices
of each link is equal, i.e., $\textrm{Tr}\left(\mathbf{\mathbf{Q}}_{l}\right)=\textrm{Tr}\left(\hat{\mathbf{\mathbf{Q}}}_{l}\right),\ \forall l$.
\end{lem}

\begin{thm}
\label{thm:FequRGWF}For a given $\mathbf{\Sigma}_{1:L}$ and its
covariance transformation $\mathbf{\hat{\Sigma}}_{1:L}$, if any $\mathbf{\Sigma}_{l}$
has the polite water-filling structure, so does $\hat{\mathbf{\Sigma}}_{l}$,
i.e., $\mathbf{\hat{Q}}_{l}\triangleq\mathbf{\Omega}_{l}^{1/2}\hat{\mathbf{\Sigma}}_{l}\mathbf{\Omega}_{l}^{1/2}$
is given by water-filling over the reverse equivalent channel $\bar{\mathbf{H}}_{l}^{\dagger}\triangleq\hat{\mathbf{\Omega}}_{l}^{-1/2}\mathbf{H}_{l,l}^{\dagger}\mathbf{\Omega}_{l}^{-1/2}$
as in (\ref{eq:WFrev}). 
\end{thm}

\begin{IEEEproof}
Because water-filling uniquely achieves the single-user MIMO channel
capacity \cite{Telatar_EuroTrans_1999_MIMOCapacity}, $\mathbf{Q}_{l}$
achieves the capacity of $\bar{\mathbf{H}}_{l}$. Since the capacities
of $\bar{\mathbf{H}}_{l}$ and $\bar{\mathbf{H}}_{l}^{\dagger}$ are
the same under the same power constraint \cite{Telatar_EuroTrans_1999_MIMOCapacity},
$\mathbf{\hat{Q}}_{l}$ achieves the capacity of $\bar{\mathbf{H}}_{l}^{\dagger}$
with the same power by Lemma \ref{lem:TrQ} and Lemma \ref{lem:mainlem}.
Therefore, $\mathbf{\hat{Q}}_{l}$ is a water-filling over $\bar{\mathbf{H}}_{l}^{\dagger}$.
\end{IEEEproof}

Decompose a MIMO problem to SISO problems can often reduce the complexity.%
{} For different decompositions, $\left\{ \mathbf{r}_{l,m}\right\} $
and $\left\{ q_{l,m}\right\} $ are different, and thus the covariance
transformation may also be different. But if the input covariance
matrices have the polite water-filling structure, its covariance transformation
is unique and has an explicit matrix expression. 
\begin{thm}
\label{thm:At-the-boundary}For any input covariance matrices $\mathbf{\Sigma}_{1:L}$
satisfying the polite water-filling structure, its covariance transformation
(\ref{eq:CovTrans}) is unique, i.e., for all decompositions of $\mathbf{\Sigma}_{1:L}$,
it will be transformed to the same dual input $\hat{\mathbf{\Sigma}}_{1:L}$,
and vice versa. Furthermore, the dual input $\hat{\mathbf{\mathbf{\Sigma}}}_{1:L}$
satisfies the following matrix equation%
{} 
\begin{eqnarray}
\hat{\mathbf{\Omega}}_{l}^{-1}\mathbf{H}_{l,l}^{\dagger}\hat{\mathbf{\Sigma}}_{l}\mathbf{H}_{l,l} & = & \mathbf{\Sigma}_{l}\mathbf{H}_{l,l}^{\dagger}\mathbf{\Omega}_{l}^{-1}\mathbf{H}_{l,l},\ l=1,...,L,\label{eq:MtformTF}
\end{eqnarray}
and can be explicitly expressed as 
\begin{align}
\hat{\mathbf{\Sigma}}_{l} & =\nu_{l}\left(\mathbf{\Omega}_{l}^{-1}-\left(\mathbf{H}_{l,l}\mathbf{\Sigma}_{l}\mathbf{H}_{l,l}^{\dagger}+\mathbf{\Omega}_{l}\right)^{-1}\right),\ l=1,...,L\label{eq:SigmhDirect}
\end{align}
where $\nu_{l}$ is the polite water-filling level in (\ref{eq:DefWFfar}).
\end{thm}

The proof is given in Appendix \ref{sub:Proof-of-conj}. The matrix
equation in (\ref{eq:MtformTF}) is natural. At the Pareto rate point,
the covariance transformation will give the same rates for the forward
and reverse links. Hence we have
\begin{align}
\textrm{log}\left|\mathbf{I}+\mathbf{H}_{l,l}^{\dagger}\hat{\mathbf{\Sigma}}_{l}\mathbf{H}_{l,l}\hat{\mathbf{\Omega}}_{l}^{-1}\right| & =\textrm{log}\left|\mathbf{I}+\mathbf{H}_{l,l}\mathbf{\Sigma}_{l}\mathbf{H}_{l,l}^{\dagger}\mathbf{\Omega}_{l}^{-1}\right|\nonumber \\
\Rightarrow\textrm{log}\left|\mathbf{I}+\hat{\mathbf{\Omega}}_{l}^{-1}\mathbf{H}_{l,l}^{\dagger}\hat{\mathbf{\Sigma}}_{l}\mathbf{H}_{l,l}\right| & =\textrm{log}\left|\mathbf{I}+\mathbf{\Sigma}_{l}\mathbf{H}_{l,l}^{\dagger}\mathbf{\Omega}_{l}^{-1}\mathbf{H}_{l,l}\right|.\label{eq:thmrateeq}
\end{align}
Theorem \ref{thm:At-the-boundary} shows that not only the determinant
but also the corresponding matrices are equal, i.e., $\hat{\mathbf{\Omega}}_{l}^{-1}\mathbf{H}_{l,l}^{\dagger}\hat{\mathbf{\Sigma}}_{l}\mathbf{H}_{l,l}=\mathbf{\Sigma}_{l}\mathbf{H}_{l,l}^{\dagger}\mathbf{\Omega}_{l}^{-1}\mathbf{H}_{l,l}$.

\subsection{\label{sub:extension-linear-constraints}Extension to a Single Linear
Constraint and Colored Noise}

So far we have assumed sum power constraint $\sum_{l=1}^{L}\textrm{Tr}\left(\mathbf{\Sigma}_{l}\right)\leq P_{T}$
and white noise for simplicity. But individual power constraints is
more common in a network, which can be handled by the following observation.
All results in this paper can be directly applied in a larger class
of problems with a single linear constraint $\sum_{l=1}^{L}\textrm{Tr}\left(\mathbf{\Sigma}_{l}\hat{\mathbf{W}}_{l}\right)\leq P_{T}$
and/or colored noise with covariance $\text{E}\left[\mathbf{w}_{l}\mathbf{\mathbf{w}}_{l}^{\dagger}\right]=\mathbf{W}_{l}$,
where $\hat{\mathbf{W}}_{l}$ and $\mathbf{W}_{l}$ are assumed to
be non-singular%
\footnote{A singular constraint or noise covariance matrix may result in infinite
power or infinite capacity.%
} and Hermitian%
. The single linear constraint appears in Lagrangian functions for
multiple linear constraints, which arise in cases of individual power
constraints, per-antenna power constraints, interference reduction
in cognitive radios, etc. \cite{Yu_IT06_Minimax_duality,Zhang_IT08_MACBC_LC}.%
{} Combined with a Lagrange multiplier update, the algorithms in this
paper can be generalized to solve the cases of multiple linear constraints.

For a single linear constraint and colored noise, we denote 
\begin{equation}
\left(\left[\mathbf{H}_{l,k}\right],\sum_{l=1}^{L}\textrm{Tr}\left(\mathbf{\Sigma}_{l}\hat{\mathbf{W}}_{l}\right)\leq P_{T},\left[\mathbf{W}_{l}\right]\right),\label{eq:net-color-linear-constraint}
\end{equation}
as a network where the channel matrices is given by $\left[\mathbf{H}_{l,k}\right]$;
the input covariance matrices must satisfy the linear constraint $\sum_{l=1}^{L}\textrm{Tr}\left(\mathbf{\Sigma}_{l}\hat{\mathbf{W}}_{l}\right)\leq P_{T}$;
and the covariance matrix of the noise at the receiver of link $l$
is given by $\mathbf{W}_{l}$. The extension is facilitated by the
following lemma which can be proved by variable substitutions.
\begin{lem}
\label{lem:whiten-net}The achievable rate region of the network (\ref{eq:net-color-linear-constraint})
is the same as the achievable rate region of the network with sum
power constraint and white noise 
\begin{equation}
\left(\left[\mathbf{W}_{l}^{-1/2}\mathbf{H}_{l,k}\hat{\mathbf{W}}_{k}^{-1/2}\right],\sum_{l=1}^{L}\textrm{Tr}\left(\mathbf{\Sigma}_{l}^{'}\right)\leq P_{T},\left[\mathbf{I}\right]\right).\label{eq:net-whitened}
\end{equation}
If $\mathbf{\Sigma}_{1:L}^{'}$ achieves certain rates and satisfies
the sum power constraint in network (\ref{eq:net-whitened}), $\mathbf{\Sigma}_{1:L}$
obtained by $\mathbf{\Sigma}_{l}=\hat{\mathbf{W}}_{l}^{-1/2}\mathbf{\Sigma}_{l}^{'}\hat{\mathbf{W}}_{l}^{-1/2},\:\forall l$
achieves the same rates and satisfies the linear constraint in network
(\ref{eq:net-color-linear-constraint}) and vice versa. 
\end{lem}

The above implies that the dual of colored noise in the forward link
is a linear constraint in the reverse link and the dual of the linear
constraint in the forward link is colored noise in the reverse link
as stated in the following theorem.
\begin{thm}
\label{thm:linear-color-dual}The dual of the network (\ref{eq:net-color-linear-constraint})
is the network
\begin{equation}
\left(\left[\mathbf{H}_{k,l}^{\dagger}\right],\sum_{l=1}^{L}\textrm{Tr}\left(\hat{\mathbf{\Sigma}}_{l}\mathbf{W}_{l}\right)\leq P_{T},\left[\hat{\mathbf{W}}_{l}\right]\right),\label{eq:net-forward-color-dual}
\end{equation}
in the sense that 1) both of them have the same achievable rate region;
2) If $\mathbf{\Sigma}_{1:L}$ achieves certain rates and satisfies
the linear constraint in network (\ref{eq:net-color-linear-constraint}),
its covariance transformation $\hat{\mathbf{\Sigma}}_{1:L}$ achieves
equal or larger rates, satisfies the dual linear constraint in network
(\ref{eq:net-forward-color-dual}), and $\sum_{l=1}^{L}\textrm{Tr}\left(\hat{\mathbf{\Sigma}}_{l}\mathbf{W}_{l}\right)=\sum_{l=1}^{L}\textrm{Tr}\left(\mathbf{\Sigma}_{l}\hat{\mathbf{W}}_{l}\right)\leq P_{T}$. \end{thm}
\begin{IEEEproof}
Apply Lemma \ref{lem:whiten-net} to networks (\ref{eq:net-color-linear-constraint})
and (\ref{eq:net-forward-color-dual}) to produce a network and its
dual with sum power constraint and white noise. Then the result follows
from Lemma \ref{lem:mainlem}. \end{IEEEproof}
\begin{rem}
For MIMO BC, the duality result here reduces to that in \cite{Zhang_IT08_MACBC_LC}. 
\end{rem}

\section{Algorithms\label{sec:Near optimal scheme}}

In this section, we use the weighted sum-rate maximization to illustrate
the benefit of polite water-filling. We first present the simpler
case of a sub-class of B-MAC networks, \textit{\emph{the interference}}\emph{
}\textit{\emph{tree (iTree) network}}s, with a concave objective function.
Then the algorithm is modified for the general B-MAC networks. Readers
who are interested in algorithm implementation only may directly go
to Section \ref{sub:alg-B-MAC} and read Table \ref{tab:table2} for
Algorithm PT and Table \ref{tab:table3} for Algorithm PP. Algorithm
P and P1 are more of theoretical value.

\subsection{The Optimization Problem and the Choice of the Coupling Matrices\label{sub:WSRMP}}

We consider the following weighted sum-rate maximization problem (\textbf{WSRMP})
with a fixed coupling matrix $\mathbf{\Phi}$, \textcolor{black}{
\begin{eqnarray}
\textrm{\textbf{WSRMP}: }g(\mathbf{\Phi})= & \underset{\mathbf{\Sigma}_{1:L}}{\textrm{max}} & f\left(\mathbf{\Sigma}_{1},...,\mathbf{\Sigma}_{L},\mathbf{\Phi}\right)\label{eq:WSRMP}\\
 & \text{s.t.} & f\left(\mathbf{\Sigma}_{1},...,\mathbf{\Sigma}_{L},\mathbf{\Phi}\right)\nonumber \\
 &  & \triangleq\sum_{l=1}^{L}w_{l}\mathcal{I}_{l}\left(\mathbf{\Sigma}_{1:L},\mathbf{\Phi}\right),\nonumber \\
 &  & \mathbf{\Sigma}_{l}\succeq0,\ \forall l,\nonumber \\
 &  & \sum_{l=1}^{L}\textrm{Tr}\left(\mathbf{\Sigma}_{l}\right)\leq P_{T},\nonumber 
\end{eqnarray}
}where $w_{l}\ge0$ is the weight for link $l$. 

We first take a detour to give a partial characterization of the optimal
choice of $\mathbf{\Phi}$, or equivalently, the optimal encoding
order in DPC and decoding order in SIC for B-MAC networks. It is in
general a difficult problem because the encoding and decoding orders
need to be optimized jointly. However, for \textcolor{black}{each
}\textit{\textcolor{black}{Pseudo }}\textit{BC}\textit{\textcolor{black}{{}
transmitter}}/\textit{\textcolor{black}{Pseudo }}\textit{MAC}\textit{\textcolor{black}{{}
receiver}} in any B-MAC network, the optimal encoding/decoding order
is easily determined by the weights and is consistent with the optimal
order of an individual BC or MAC, as proved in Theorem \ref{thm:optorder}
below.
\begin{defn}
In a B-MAC network%
, a transmitter with a set of associated links, whose indices forms
a set $\mathcal{L}_{\textrm{B}}$, is said to be a \textit{\textcolor{black}{Pseudo
}}\textit{BC}\textit{\textcolor{black}{{} transmitter}} if the links
in $\mathcal{L}_{\textrm{B}}$ either%
{} all interfere with a link $k$ or none of them interfere with the
link $k$, $\forall k\in\mathcal{L}_{\textrm{B}}^{C}$, i.e., the
columns of the coupling matrix $\mathbf{\Phi}$ indexed by $\mathcal{L}_{\textrm{B}}$,
excluding rows indexed by $\mathcal{L}_{\textrm{B}}$, are the same.
A receiver with a set of associated links, whose indices forms a set
$\mathcal{L}_{\textrm{M}}$, is said to be a \textit{\textcolor{black}{Pseudo
}}\textit{MAC}\textit{\textcolor{black}{{} receiver}} if the links in
$\mathcal{L}_{\textrm{M}}$ are either all interfered by a link $k$
or none of them is interfered by the link $k$, $\forall k\in\mathcal{L}_{\textrm{M}}^{C}$,
i.e., the rows of the coupling matrix $\mathbf{\Phi}$ indexed by
$\mathcal{L}_{\textrm{M}}$, excluding columns indexed by $\mathcal{L}_{\textrm{M}}$,
are the same. \textcolor{black}{}%

\end{defn}

For example, if%
{} $\forall l\in\mathcal{L}_{\textrm{B}}$, link $l$ is the last one
to be decoded at its receiver%
, then the corresponding transmitter is a \textcolor{black}{pseudo
}BC\textcolor{black}{{} transmitter}. If $\forall l\in\mathcal{L}_{\textrm{M}}$,
link $l$ is the first one to be encoded at its transmitter, then
the corresponding receiver is a \textcolor{black}{pseudo }MAC\textcolor{black}{{}
receiver}.
\begin{thm}
\label{thm:optorder}In a B-MAC network employing DPC and SIC with
the optimal encoding and decoding order $\pi^{*}$ of the following
problem 
\begin{eqnarray}
 & \max_{\pi} & g\left(\mathbf{\Phi}(\pi)\right),\label{eq:max-pi}
\end{eqnarray}
if there exists a pseudo BC transmitter (pseudo MAC receiver), its
link with the $n^{\text{th}}$ largest (smallest) weight is the $n^{\text{th}}$
one to be encoded (decoded). \end{thm}
\begin{IEEEproof}
It is proved by isolating the \textcolor{black}{links of a Pseudo
}BC\textcolor{black}{{} transmitter or a Pseudo }MAC\textcolor{black}{{}
receiver} from the network to form an individual BC or MAC. In the
optimal solution of (\ref{eq:max-pi}), $\left\{ \mathbf{\Sigma}_{l}:\ l\in\mathcal{L}_{\textrm{B}}\right\} $
and $\pi^{*}$ are also the optimal input and encoding order that
maximizes the weighted sum-rate of a BC with fixed interference from
links in $\mathcal{L}_{\textrm{B}}^{C}$ and under multiple linear
constraints that the interference to links in $\mathcal{L}_{\textrm{B}}^{C}$
must not exceed the optimal values. The result on BC with multiple
linear constraints in \cite{Zhang_08sIT_BC_MAC_duality_multiple_constraints}
implies that the optimal encoding order is as stated in the theorem.
For \textcolor{black}{a pseudo }MAC\textcolor{black}{{} receiver}, using
similar method and generalizing the result in Section \ref{sub:extension-linear-constraints}
to multiple linear constraints gives the desired decoding order.
\end{IEEEproof}

We briefly discuss the impact of the choice of the coupling matrix
$\mathbf{\Phi}$. The achievable rate region with any valid coupling
matrix is outer and inner bounded by the capacity region and the achievable
rate region without any interference cancellation respectively. In
some cases, such as MIMO MAC, MIMO BC, two-user MIMO interference
channel \cite{Jafar_IT07_DOFMIMOIFC}, and some MIMO X channels\cite{Jafar_IT08_DOFMIMOX},
the achievable rate without interference cancellation is optimal in
terms of degree of freedoms (DOF) at high SNR, and thus the choice
of the coupling matrix $\mathbf{\Phi}$ has no impact on the DOF.
For general cases, it is difficult to analyze the performance of different
coupling matrices. But the optimal encoding/decoding order for \textcolor{black}{each
}\textit{\textcolor{black}{Pseudo }}\textit{BC}\textit{\textcolor{black}{{}
transmitter}}/\textit{\textcolor{black}{Pseudo }}\textit{MAC}\textit{\textcolor{black}{{}
receiver}} in Theorem \ref{thm:optorder} can be used to improve any
given orders in a B-MAC network.

All the rest of the paper is for a fixed coupling matrix $\mathbf{\Phi}$
and the argument $\mathbf{\Phi}$ in $f\left(\mathbf{\Sigma}_{1},...,\mathbf{\Sigma}_{L},\mathbf{\Phi}\right)$
is omitted. We consider centralized algorithms with global channel
knowledge. Distributed algorithms can be developed based on them as
discussed in Section \ref{sec:Conclusion}.

\subsection{iTree Networks\label{sub:itree}}

iTree networks appears to be a natural extension of MAC and BC. We
define it below.
\begin{defn}
A B-MAC network with a fixed coupling matrix is called an \emph{Interference
Tree (iTree) Network} if after interference cancellation, the links
can be indexed such that any link is not interfered by the links with
smaller indices. 
\end{defn}

\begin{defn}
In an \emph{Interference Graph}, each node represents a link. A directional
edge from node $i$ to node $j$ means that link $i$ causes interference
to link $j$.
\end{defn}

\begin{rem}
An iTree network is related to but different from a network whose
channel gains has a tree topology.%
{} A network with tree topology is an iTree network only if the interference
cancellation order is proper. For example, a MAC, which has a tree
topology, is not an iTree network if the successive decoding is not
employed at the receiver. On the other hand, a network with loops
may still be an iTree network. We give such an example in Fig. \ref{fig:iTree}
where DPC and SIC are employed. With encoding/decoding order A, where
the signal $\mathbf{x}_{2}$ is decoded after $\mathbf{x}_{1}$ and
the signal $\mathbf{x}_{3}$ is encoded after $\mathbf{x}_{2}$, each
link $l\in\left\{ 2,3,4\right\} $ is not interfered by the first
$l-1$ links. Therefore, the network in Fig. \ref{fig:iTree} is still
an iTree network even though it has a loop of nonzero channel gains.
However, for encoding/decoding order B, SIC is not employed at $R_{1}/R_{2}$,
and $\mathbf{x}_{2}$ is encoded after $\mathbf{x}_{3}$ at $T_{2}/T_{3}$.
The network in Fig. \ref{fig:iTree} is no longer an iTree network
because the interference graph has directional loops. 
\end{rem}

\begin{figure}
\begin{centering}
\textsf{\includegraphics[clip,angle=270,scale=0.3]{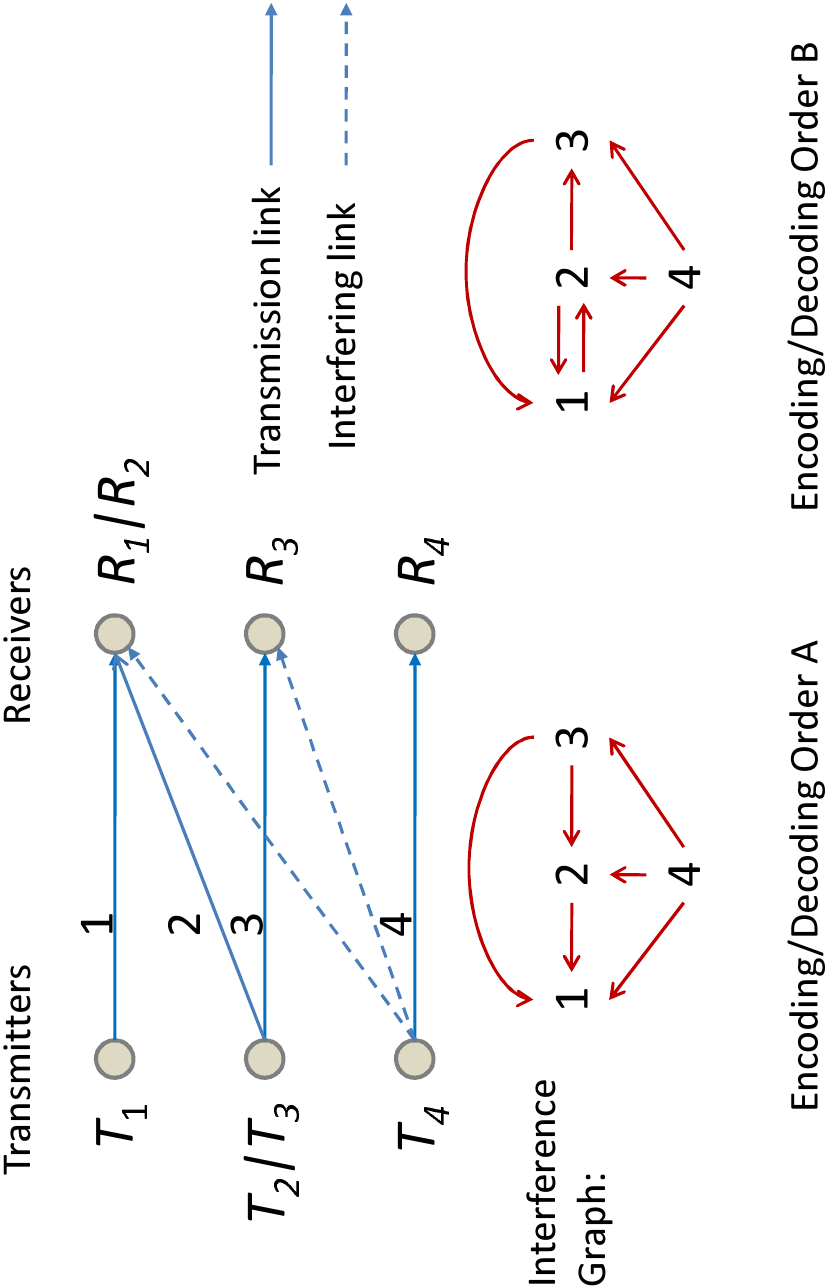}}
\par\end{centering}

\caption{\label{fig:iTree}A network with topology loop in terms of nonzero
channel gain and encoding/decoding order A is an iTree network, whose
interference graph does not have any directional loop. With encoding/decoding
order B, it is not an iTree network because the interference graph
has directional loops.}
\end{figure}

iTree networks can be equivalently defined using their interference
graphs.
\begin{lem}
A B-MAC network with a fixed coupling matrix is an iTree network if
and only if after interference cancellation, its interference graph
does not have any directional loop.\end{lem}
\begin{IEEEproof}
If the interference graph has no loops, the following algorithm can
find indices satisfying the definition of the iTree network: 1) $l=1$;
2) index a node whose edges are all incoming by $l$; 3) delete the
node and all the edges connected to it; 4) let $l=l+1$ and repeat
2) until all nodes are indexed. On the other hand, the interference
graph of an iTree network only has edges from a node to the nodes
with lower indices, making them impossible to form a directional loop.
\end{IEEEproof}

Since the coupling matrix of the reverse links is the transpose of
the coupling matrix of the forward links, the interference relation
is reversed as stated in the following lemmas. Without loss of generality,
we consider iTree networks where the $l^{th}$ link is not interfered
by the first $l-1$ links in this paper, i.e., the coupling matrix
is given by $\mathbf{\Phi}_{l,k}=0,\forall k\leq l$ and $\mathbf{\Phi}_{l,k}=1,\forall k>l$.
\begin{lem}
\label{lem:RevInf}If in an iTree network, the $l^{th}$ link is not
interfered by the links with lower indices, in the reverse links,
the $l^{th}$ link is not interfered by the links with higher indices.

The interference graph of the reverse links can be obtained from the
interference graph of the forward links by reversing the directions
of all the edges.
\end{lem}

Two examples of iTree networks with concave weighted sum-rate functions
are 1) a MAC with the decoding order equal to the ascending order
of the weights $\left\{ w_{l}\right\} $ and 2) a \textit{\emph{two-user
Z channel}} ($\mathbf{H}_{2,1}=\mathbf{0}$) specified in the following
theorem.
\begin{thm}
\label{thm:Zconcave}The weighted sum-rate of a \emph{two-user Z channel
is concave }function of the input covariance matrices\emph{ if the
following are satisfied}.
\begin{enumerate}
\item The channel matrices $\mathbf{H}_{1,2}$ and $\mathbf{H}_{2,2}$ are
invertible and satisfy $\mathbf{H}_{2,2}^{\dagger}\mathbf{H}_{2,2}\succeq\mathbf{H}_{1,2}^{\dagger}\mathbf{H}_{1,2}$. 
\item The weights satisfy $w_{1}\leq w_{2}$.
\end{enumerate}
\end{thm}

The proof is given in Appendix \ref{sub:Proof-of-TheoremZconcave}.
The important question of a complete characterization of the subset
of iTree networks with concave weighted sum-rate function is left
for future work.

\subsection{\label{sub:Optimality-Analysis}Algorithms for iTree Networks with
Concave Objective Functions}

The algorithm in this section illustrates the usage of polite water-filling
and is a nontrivial generalization of the algorithm in \cite{Jindal_IT05_IFBC}.
First we show that \textbf{WSRMP} (\ref{eq:WSRMP}) with concave objective
function $f(\cdot)$ can be equivalently solved by the following convex
optimization problem.

\begin{align}
\underset{\mathbf{\Sigma}(1:L)}{\textrm{max}} & f_{\textrm{mod}}\left(\mathbf{\Sigma}\left(1\right),\mathbf{\Sigma}\left(2\right),...,\mathbf{\Sigma}\left(L\right)\right)\label{eq:MWSRM}\\
\textrm{s.t.} & \forall k,l,\ \mathbf{\Sigma}_{l}\left(k\right)\succeq0,\:\textrm{and}\:\forall k\ \sum_{l=1}^{L}\textrm{Tr}\left(\mathbf{\Sigma}_{l}\left(k\right)\right)\leq P_{T},\nonumber 
\end{align}
where $\mathbf{\Sigma}\left(k\right)\triangleq\left(\mathbf{\Sigma}_{1}\left(k\right),...,\mathbf{\Sigma}_{L}\left(k\right)\right)$
for $k=1,...,L$ with $\mathbf{\Sigma}_{l}\left(k\right)\in\mathbb{C}^{L_{T_{l}}\times L_{T_{l}}}$.
The objective function is defined as
\begin{align*}
 & f_{\textrm{mod}}(\mathbf{\Sigma}\left(1\right),\mathbf{\Sigma}\left(2\right),...,\mathbf{\Sigma}\left(L\right))\triangleq\\
 & \frac{1}{L}{\displaystyle \sum_{i=1}^{L}}f\left(\mathbf{\Sigma}_{1}\left(\left[1-i\right]_{L}\right),...,\mathbf{\Sigma}_{l}\left(\left[l-i\right]_{L}\right),...,\mathbf{\Sigma}_{L}\left(\left[L-i\right]_{L}\right)\right),
\end{align*}
where $\left[n\right]_{L}=\left(n\ \textrm{mod}\ L\right)+1$. 

Problem (\ref{eq:MWSRM}) can be viewed as a weighted sum-rate maximization
problem of $L$ networks which are identical to the network in the
original problem in (\ref{eq:WSRMP}). The purpose of expanding the
single-network optimization problem in (\ref{eq:WSRMP}) to the multiple-network
optimization problem in (\ref{eq:MWSRM}) is to decouple the power
constraints so that the input covariance matrices of a power constraint
belong to different networks. In problem (\ref{eq:MWSRM}), the update
of the input covariance matrices of a power constraint is easier because
there is no interference among the networks, and other input covariance
matrices not belonging to this power constraint can be treated as
constants, resulting in monotonically convergent iterative algorithm.
A similar method is also used in \cite{Jindal_IT05_IFBC,Kobayashi_JSAC06_ITWMISOBC}.
We use the subscript $i$ to denote the terms for the $i^{th}$ network,
e.g., the input covariance matrices for the $i^{th}$ network is denoted
as $\mathbf{\Sigma}_{i,1:L}=\left(\mathbf{\Sigma}_{i,1},...,\mathbf{\Sigma}_{i,L}\right)$.
The one-to-one mapping between $\left(\mathbf{\Sigma}_{1,1:L},...,\mathbf{\Sigma}_{L,1:L}\right)$
and $\left(\mathbf{\Sigma}\left(1\right),...,\mathbf{\Sigma}\left(L\right)\right)$
is given by
\begin{eqnarray}
\mathbf{\Sigma}_{i,l} & = & \mathbf{\Sigma}_{l}\left(\left[l-i\right]_{L}\right),\label{eq:MApsigma}\\
\mathbf{\Sigma}_{l}\left(k\right) & = & \mathbf{\Sigma}_{\left[l-k\right]_{L},l},\forall i,k,l\in\left\{ 1,...,L\right\} .\nonumber 
\end{eqnarray}
Then we have 
\[
f_{\textrm{mod}}(\mathbf{\Sigma}\left(1\right),\mathbf{\Sigma}\left(2\right),...,\mathbf{\Sigma}\left(L\right))=\frac{1}{L}\sum_{i=1}^{L}\sum_{l=1}^{L}w_{l}\mathcal{I}_{i,l}\left(\mathbf{\Sigma}_{i,1:L},\mathbf{\Phi}\right).
\]
The following example illustrates the mapping in (\ref{eq:MApsigma})
for $L=4$:
\[
\begin{array}{ccccc}
 & 1^{\textrm{th}}\:\textrm{link} & 2^{\textrm{th}}\:\textrm{link} & 3^{\textrm{th}}\:\textrm{link} & 4^{\textrm{th}}\:\textrm{link}\\
1^{\textrm{th}}\:\textrm{network} & \mathbf{\Sigma}_{1}\left(1\right) & \mathbf{\Sigma}_{2}\left(2\right) & \mathbf{\Sigma}_{3}\left(3\right) & \mathbf{\Sigma}_{4}\left(4\right)\\
2^{\textrm{th}}\:\textrm{network} & \mathbf{\Sigma}_{1}\left(4\right) & \mathbf{\Sigma}_{2}\left(1\right) & \mathbf{\Sigma}_{3}\left(2\right) & \mathbf{\Sigma}_{4}\left(3\right)\\
3^{\textrm{th}}\:\textrm{network} & \mathbf{\Sigma}_{1}\left(3\right) & \mathbf{\Sigma}_{2}\left(4\right) & \mathbf{\Sigma}_{3}\left(1\right) & \mathbf{\Sigma}_{4}\left(2\right)\\
4^{\textrm{th}}\:\textrm{network} & \mathbf{\Sigma}_{1}\left(2\right) & \mathbf{\Sigma}_{2}\left(3\right) & \mathbf{\Sigma}_{3}\left(4\right) & \mathbf{\Sigma}_{4}\left(1\right)
\end{array},
\]
where the element at the $l^{\textrm{th}}$ column and $i^{\textrm{th}}$
row is $\mathbf{\Sigma}_{i,l}$, the input covariance matrix for the
$l^{\textrm{th}}$ link of the $i^{\textrm{th}}$ network. The power
constraint $\sum_{l=1}^{L}\textrm{Tr}\left(\mathbf{\Sigma}_{l}\left(k\right)\right)\leq P_{T}$
couples the networks together. 

The following lemma holds for $f_{\textrm{mod}}(\cdot)$.
\begin{lem}
\label{lem:Propfmod}The function $f_{\textrm{mod}}(\cdot)$ satisfies
the following properties: 
\begin{enumerate}
\item Let $\mathbf{\Sigma}_{l}\succeq0,\ \forall l$, satisfy $\sum_{l=1}^{L}\textrm{Tr}\left(\mathbf{\Sigma}_{l}\right)\leq P_{T}$.
The mapping $\mathbf{\Sigma}\left(k\right)=\mathbf{\Sigma}_{1:L},\ \forall k$,
results in $\sum_{l=1}^{L}\textrm{Tr}\left(\mathbf{\Sigma}_{l}\left(k\right)\right)\leq P_{T},\ \forall k$,
and $f_{\textrm{mod}}(\mathbf{\Sigma}(1:L))$ satisfying $f_{\textrm{mod}}(\mathbf{\Sigma}_{1:L},...,\mathbf{\Sigma}_{1:L})=f\left(\mathbf{\Sigma}_{1:L}\right)$.
\item Let $\mathbf{\Sigma}_{l}\left(k\right)\succeq0,\ \forall l,k$, satisfy
$\sum_{l=1}^{L}\textrm{Tr}\left(\mathbf{\Sigma}_{l}\left(k\right)\right)\leq P_{T},\ \forall k$.
The mapping $\mathbf{\Sigma}_{1:L}=\left(1/L\right)$ $\sum_{k=1}^{L}\mathbf{\Sigma}\left(k\right)$
results in $\sum_{l=1}^{L}\textrm{Tr}\left(\mathbf{\Sigma}_{l}\right)\leq P_{T}$
and $f\left(\mathbf{\Sigma}_{1:L}\right)$ satisfying%
\begin{eqnarray*}
 &  & f_{\textrm{mod}}\left(\mathbf{\Sigma}\left(1\right),\mathbf{\Sigma}\left(2\right),...,\mathbf{\Sigma}\left(L\right)\right)\\
 & \le & f\left(\mathbf{\Sigma}_{1:L}\right).
\end{eqnarray*}

\end{enumerate}
\end{lem}

\begin{IEEEproof}
The first property is obvious and the second property holds because
\begin{eqnarray}
 &  & f_{\textrm{mod}}(\mathbf{\Sigma}\left(1\right),\mathbf{\Sigma}\left(2\right),...,\mathbf{\Sigma}\left(L\right))\nonumber \\
 & \le & f\left(\frac{1}{L}\sum_{k=1}^{L}\mathbf{\Sigma}_{1}\left(k\right),...,\frac{1}{L}\sum_{k=1}^{L}\mathbf{\Sigma}_{L}\left(k\right)\right)\label{eq:fmodeqf}
\end{eqnarray}
where (\ref{eq:fmodeqf}) follows from the concavity of $f(\cdot)$.

\end{IEEEproof}

The lemma says that the problems (\ref{eq:WSRMP}) and (\ref{eq:MWSRM})
are equivalent and every optimal solution $\mathbf{\tilde{\Sigma}}\left(1\right),\mathbf{\tilde{\Sigma}}\left(2\right),...,\mathbf{\tilde{\Sigma}}\left(L\right)$
of problem (\ref{eq:MWSRM}) maps directly to the optimal solution
$\mathbf{\tilde{\Sigma}}_{1:L}=\left(1/L\right)\sum_{k=1}^{L}\tilde{\mathbf{\Sigma}}\left(k\right)$
of the \textbf{WSRMP} (\ref{eq:WSRMP}).

We first obtain insight of the problem by finding the necessary and
sufficient conditions satisfied by the optimum. With the insight,
we design an algorithm which monotonically increases the objective
function until convergence to the optimum.
\begin{thm}
\label{thm:optimality}Necessity: If $\mathbf{\tilde{\Sigma}}\left(1\right),...,\mathbf{\tilde{\Sigma}}\left(L\right)$
is an optimum of problem (\ref{eq:MWSRM}), then $\forall k$, $\mathbf{\tilde{\Sigma}}\left(k\right)$
must satisfy the following optimality conditions: 
\begin{enumerate}
\item For any $1\leq l\leq L$, $\mathbf{\tilde{\Sigma}}_{l}\left(k\right)$
possesses the polite water-filling structure, i.e., $\tilde{\mathbf{\hat{\mathbf{\Omega}}}}_{\left[l-k\right]_{L},l}^{1/2}\mathbf{\tilde{\Sigma}}_{l}\left(k\right)$
$\tilde{\mathbf{\hat{\mathbf{\Omega}}}}_{\left[l-k\right]_{L},l}^{1/2}$
is given by water-filling over the equivalent channel $\tilde{\mathbf{\Omega}}_{\left[l-k\right]_{L},l}^{-1/2}\mathbf{H}_{l,l}\tilde{\mathbf{\hat{\mathbf{\Omega}}}}_{\left[l-k\right]_{L},l}^{-1/2}$. 
\item The polite water-filling level for $\mathbf{\tilde{\Sigma}}_{l}\left(k\right)$
is given by $\tilde{\nu}_{l}=w_{l}/\tilde{\mu}$, where $\tilde{\nu}_{l}$
does not depend on $k$ and $\tilde{\mu}>0$ is chosen such that $\sum_{l=1}^{L}\textrm{Tr}\left(\mathbf{\tilde{\Sigma}}_{l}\left(k\right)\right)=P_{T}$.
\end{enumerate}
Sufficiency: If certain $\mathbf{\tilde{\Sigma}}\left(1\right),...,\mathbf{\tilde{\Sigma}}\left(L\right)$
satisfies the above optimality conditions, then it must satisfy the
KKT conditions of problem (\ref{eq:MWSRM}), and thus is the global
optimum of problem (\ref{eq:MWSRM}).
\end{thm}

We go over the proof in detail as it helps design the algorithm. It
can be proved by contradiction that the optimum $\mathbf{\tilde{\Sigma}}\left(1\right),...,\mathbf{\tilde{\Sigma}}\left(L\right)$
must satisfy $\sum_{l=1}^{L}\textrm{Tr}\left(\mathbf{\tilde{\Sigma}}_{l}\left(k\right)\right)=P_{T},\ \forall k$.
Otherwise, if $\sum_{l=1}^{L}\textrm{Tr}\left(\mathbf{\tilde{\Sigma}}_{l}\left(k\right)\right)<P_{T}$
for some $k$, we can strictly improve the rate of the first link
in the $(\left[1-k\right]_{L})^{\text{th}}$ network by increasing
the power of $\mathbf{\tilde{\Sigma}}_{1}\left(k\right)$ until $\sum_{l=1}^{L}\textrm{Tr}\left(\mathbf{\tilde{\Sigma}}_{l}\left(k\right)\right)=P_{T}$
without causing interference to other links, which contradicts that
$\mathbf{\tilde{\Sigma}}\left(1\right),...,\mathbf{\tilde{\Sigma}}\left(L\right)$
is an optimum. 

The remaining necessity part is proved by showing that if for some
$k$, $\mathbf{\Sigma}\left(k\right)$ does not satisfy the polite
water-filling structure in Theorem \ref{thm:optimality}, the objective
function can be strictly increased by enforcing this structure on
$\mathbf{\Sigma}\left(k\right)$. Without loss of generality, we only
need to prove this for $k=1$ due to the circular structure of $f_{\textrm{mod}}(\cdot)$,
i.e., 
\begin{eqnarray*}
 & f_{\textrm{mod}}(\mathbf{\Sigma}\left(1\right),\mathbf{\Sigma}\left(2\right),...,\mathbf{\Sigma}\left(L\right))\\
= & f_{\textrm{mod}}(\mathbf{\Sigma}\left(k\right),...,\mathbf{\Sigma}\left(L\right),\mathbf{\Sigma}\left(1\right),...,\mathbf{\Sigma}\left(k-1\right)).
\end{eqnarray*}

We define some notations and give two lemmas. For the $i^{\text{th}}$
network, fixing the input covariance matrices $\mathbf{\Sigma}_{i,j},\ j=i+1,...,L$
for the last $L-i$ links, the first $i$ links form a sub-network
which is also an iTree network
\begin{equation}
\left(\left[\mathbf{H}_{l,k}\right]_{k,l=1,...,i},\sum_{l=1}^{i}\textrm{Tr}\left(\mathbf{\Sigma}_{i,l}\right)=P_{T}^{i},\left[\mathbf{W}_{i,l}\right]_{l=1,...,i}\right),\label{eq:Sub1toi}
\end{equation}
where $\mathbf{W}_{i,l}=\mathbf{I}+\sum_{j=i+1}^{L}\mathbf{\Phi}_{l,j}\mathbf{H}_{l,j}\mathbf{\Sigma}_{i,j}\mathbf{H}_{l,j}^{\dagger}$
is the covariance matrix of the equivalent colored noise of link $l$;
$P_{T}^{i}=\sum_{j=1}^{i}\textrm{Tr}\left(\mathbf{\Sigma}_{i,j}\right)$.
By Theorem \ref{thm:linear-color-dual}, the corresponding dual sub-network
is 
\begin{equation}
\left(\left[\mathbf{H}_{k,l}^{\dagger}\right]_{k,l=1,...,i},\sum_{l=1}^{i}\textrm{Tr}\left(\hat{\mathbf{\Sigma}}_{i,l}\mathbf{W}_{i,l}\right)=P_{T}^{i},\left[\mathbf{I}\right]\right).\label{eq:Sub1toi-dual}
\end{equation}
Denote $\mathcal{I}_{i,l}\left(\mathbf{\Sigma}_{i,1:i},\mathbf{\Phi}\right)$
and $\mathcal{\hat{I}}_{i,l}\left(\hat{\mathbf{\Sigma}}_{i,1:i},\mathbf{\Phi}^{T}\right)$
the forward and reverse link rates of the $l^{\text{th}}$ link of
the $i^{\text{th}}$ sub-network (\ref{eq:Sub1toi}) achieved by $\mathbf{\Sigma}_{i,1:i}$
and $\hat{\mathbf{\Sigma}}_{i,1:i}$ respectively. In contrast, $\mathcal{I}_{i,l}\left(\mathbf{\Sigma}_{i,1:L},\mathbf{\Phi}\right)$
and $\mathcal{\hat{I}}_{i,l}\left(\hat{\mathbf{\Sigma}}_{i,1:L},\mathbf{\Phi}^{T}\right)$
denote the forward and reverse link rates of the $l^{\text{th}}$
link of the $i^{\text{th}}$ network.
\begin{lem}
\label{lem:RsigmaSub}Let $\mathbf{\hat{\Sigma}}_{i,1:L}=\left(\mathbf{\hat{\Sigma}}_{i,1},\mathbf{\hat{\Sigma}}_{i,2},...,\mathbf{\hat{\Sigma}}_{i,L}\right)$
be the covariance transformation of $\mathbf{\Sigma}_{i,1:L}=\left(\mathbf{\Sigma}_{i,1},\mathbf{\Sigma}_{i,2},...,\mathbf{\Sigma}_{i,L}\right)$,
applied to the $i^{th}$ network. Then $\hat{\mathbf{\Sigma}}_{i,1:i}=\left(\mathbf{\hat{\Sigma}}_{i,1},\mathbf{\hat{\Sigma}}_{i,2},...,\mathbf{\hat{\Sigma}}_{i,i}\right)$
is also the covariance transformation of $\mathbf{\Sigma}_{i,1:i}=\left(\mathbf{\Sigma}_{i,1},\mathbf{\Sigma}_{i,2},...,\mathbf{\Sigma}_{i,i}\right)$,
applied to the $i^{\text{th}}$ sub-network (\ref{eq:Sub1toi}).\end{lem}
\begin{IEEEproof}
The interference from link $i+1,...,L$ is counted in the colored
noise in the $i^{\text{th}}$ sub-network (\ref{eq:Sub1toi}). Therefore,
the MMSE-SIC receive vectors $\left\{ \mathbf{r}_{i,l,m},l=1,...,i\right\} $
are the same in both $i^{\text{th}}$ network and its sub-network.
Because the $i^{\text{th}}$ network is an iTree network, in its dual
network, there is no interference from link $i+1,...,L$ to link $1,...,i$
by Lemma \ref{lem:RevInf}. Therefore, for link $1,...,i$, there
is no difference between the $i^{\text{th}}$ dual sub-network (\ref{eq:Sub1toi-dual})
and the $i^{\text{th}}$ dual network. In both dual networks, the
transmit powers of the first $i$ links given by (\ref{eq:qpower})
in the covariance transformation achieve the same SINRs as in the
$i^{\text{th}}$ network. But for fixed $\left\{ \mathbf{r}_{i,l,m},l=1,...,i\right\} $
and $\left\{ \mathbf{t}_{i,l,m},l=1,...,i\right\} $, the transmit
powers producing the same SINRs are unique \cite[Lemma 1]{Rao_TOC07_netduality}.
Therefore the transmit powers of the first $i$ links must be the
same for the two dual networks, which implies that Lemma \ref{lem:RsigmaSub}
must hold.
\end{IEEEproof}

\begin{lem}
\label{lem:KKTiTree}The necessary and sufficient conditions for $\mathbf{\tilde{\Sigma}}_{1:L}$
to be the optimal solution of problem (\ref{eq:WSRMP}) for an iTree
network with concave objective function are listed below.
\begin{enumerate}
\item It possesses the polite water-filling structure. 
\item The polite water-filling level $\tilde{\nu}_{l}$ for $\mathbf{\tilde{\Sigma}}_{l}$
is given by $\tilde{\nu}_{l}=w_{l}/\tilde{\mu}$, where $\tilde{\mu}>0$
is chosen such that $\sum_{l=1}^{L}\textrm{Tr}\left(\mathbf{\tilde{\Sigma}}_{l}\right)=P_{T}$.
\end{enumerate}
\end{lem}

The proof is given in Appendix \ref{sub:Proof-of-Lemma-kktitree}.

Then we give a three-steps algorithm to improve the objective function.
Note that $\mathbf{\Sigma}\left(1\right)$ contains the input covariance
matrices of the $i^{\text{th}}$ link of the $i^{\text{th}}$ network
for $i=1,...,L$, i.e., $\mathbf{\Sigma}\left(1\right)=\left(\mathbf{\Sigma}_{1,1},\mathbf{\Sigma}_{2,2},...,\mathbf{\Sigma}_{L,L}\right)$.

\textbf{Step }1: For $i=1,...,L$, calculate $\hat{\mathbf{\Sigma}}_{i,1:i-1}$
by the covariance transformation applied to the $i^{\text{th}}$ sub-network.
Due to the special interference structure of the iTree network, the
calculation of the reverse transmit powers of the covariance transformation
can be simplified to be calculated one by one as follows. When calculating
$q_{i,l,m}$, the transmit powers $\left\{ q_{i,k,m}:\ m=1,...,M_{i,k},\ k=1,...,\right.$
$\left.l-1\right\} $ and $\left\{ q_{i,l,n}:\ n=1,...,m-1\right\} $
have been calculated. Therefore, we can calculate $\hat{\mathbf{\Sigma}}_{i,k}=\sum_{m=1}^{M_{i,k}}q_{i,k,m}\mathbf{r}_{i,k,m}\mathbf{r}_{i,k,m}^{\dagger},\ k=1,...,l-1$
and obtain the interference-plus-noise covariance matrix of the reverse
link $l$ as $\hat{\mathbf{\Omega}}_{i,l}=\mathbf{I}+\sum_{k=1}^{l-1}\mathbf{\Phi}_{k,l}\mathbf{H}_{k,l}^{\dagger}\hat{\mathbf{\Sigma}}_{i,k}\mathbf{H}_{k,l}$.
Then $q_{i,l,m}$ is given by
\begin{align}
 & q_{i,l,m}\label{eq:RLpowtree}\\
= & \frac{\gamma_{i,l,m}\left(\mathbf{t}_{i,l,m}^{\dagger}\hat{\mathbf{\Omega}}_{i,l}\mathbf{t}_{i,l,m}+\sum_{n=1}^{m-1}q_{i,l,n}\left|\mathbf{t}_{i,l,m}^{\dagger}\mathbf{H}_{l,l}^{\dagger}\mathbf{r}_{i,l,n}\right|^{2}\right)}{\left|\mathbf{t}_{i,l,m}^{\dagger}\mathbf{H}_{l,l}^{\dagger}\mathbf{r}_{i,l,m}\right|^{2}},\nonumber 
\end{align}
and correspondingly 
\begin{equation}
\hat{\mathbf{\Sigma}}_{i,l}=\sum_{m=1}^{M_{i,l}}q_{i,l,m}\mathbf{r}_{i,l,m}\mathbf{r}_{i,l,m}^{\dagger}.\label{eq:signhedi}
\end{equation}
By Theorem \ref{thm:linear-color-dual}, we have $\mathcal{\hat{I}}_{i,l}\left(\hat{\mathbf{\Sigma}}_{i,1:i},\mathbf{\Phi}^{T}\right)\geq\mathcal{I}_{i,l}\left(\mathbf{\Sigma}_{i,1:i},\mathbf{\Phi}\right)=\mathcal{I}_{i,l}\left(\mathbf{\Sigma}_{i,1:L},\mathbf{\Phi}\right),\ l=1,...,i$
and $\sum_{l=1}^{i}\textrm{Tr}\left(\hat{\mathbf{\Sigma}}_{i,l}\mathbf{W}_{i,l}\right)=\sum_{l=1}^{i}\textrm{Tr}\left(\mathbf{\Sigma}_{i,l}\right)=P_{T}^{i}$.

\textbf{Step} 2: Improve $\sum_{i=1}^{L}w_{i}\mathcal{\hat{I}}_{i,i}\left(\hat{\mathbf{\Sigma}}_{i,1:i},\mathbf{\Phi}^{T}\right)$
by enforcing the polite water-filling structure on $\hat{\mathbf{\Sigma}}_{i,i},\forall i$.
By Lemma \ref{lem:RevInf}, in the $i^{th}$ sub network, the reverse
link $i$ causes no interference to the first $i-1$ reverse links.
Fixing $\hat{\mathbf{\Sigma}}_{i,l},l=1,...,i-1,\forall i$, we can
improve $\sum_{i=1}^{L}w_{i}\mathcal{\hat{I}}_{i,i}\left(\hat{\mathbf{\Sigma}}_{i,1:i},\mathbf{\Phi}^{T}\right)$
without reducing the rates of reverse link $1,...,i-1$ in the $i^{th}$
sub-network for all $i$ by solving the following weighted sum-rate
maximization problem for $L$ parallel channels with colored noise:

\textcolor{black}{
\begin{align}
 & \underset{\hat{\mathbf{\Sigma}}_{i,i}\succeq0,\forall i=1,...,L}{\textrm{max}}\:\sum_{i=1}^{L}w_{i}\textrm{log}\left|\mathbf{I}+\mathbf{H}_{i,i}^{\dagger}\hat{\mathbf{\Sigma}}_{i,i}\mathbf{H}_{i,i}\hat{\mathbf{\Omega}}_{i,i}^{-1}\right|\label{eq:WSRM-Parallel}\\
\textrm{s.t.} & \sum_{i=1}^{L}\textrm{Tr}\left(\hat{\mathbf{\Sigma}}_{i,i}\mathbf{W}_{i,i}\right)\leq\sum_{i=1}^{L}\left(P_{T}^{i}-\sum_{l=1}^{i-1}\textrm{Tr}\left(\hat{\mathbf{\Sigma}}_{i,l}\mathbf{W}_{i,l}\right)\right),\nonumber 
\end{align}
}where $\hat{\mathbf{\Omega}}_{i,i}=\mathbf{I}+\sum_{l=1}^{i-1}\mathbf{\Phi}_{l,i}\mathbf{H}_{l,i}^{\dagger}\hat{\mathbf{\Sigma}}_{i,l}\mathbf{H}_{l,i}$
and $\mathbf{W}_{i,i}=\mathbf{I}+\sum_{j=i+1}^{L}\mathbf{\Phi}_{i,j}\mathbf{H}_{i,j}\mathbf{\Sigma}_{i,j}\mathbf{H}_{i,j}^{\dagger}=\mathbf{\Omega}_{i,i}$.
This is a classic problem with unique water-filling solution. The
solution can be obtained from Theorem \ref{thm:WFST} and Lemma \ref{lem:whiten-net}
for a network of parallel channels. Perform the thin SVD $\mathbf{\Omega}_{i,i}^{-1/2}\mathbf{H}_{i,i}\hat{\mathbf{\Omega}}_{i,i}^{-1/2}=\mathbf{F}_{i}\mathbf{\Delta}_{i}\mathbf{G}_{i}^{\dagger}$.
Let $N_{i}=\textrm{Rank}\left(\mathbf{H}_{i,i}\right)$ and $\delta_{i,j}$
be the $j^{\text{th}}$ diagonal element of $\mathbf{\Delta}_{i}^{2}$.
Obtain $\mathbf{D}_{i}$ as 
\begin{eqnarray}
\mathbf{D}_{i} & = & \textrm{diag}\left(d_{i,1},...,d_{i,N_{i}}\right),\nonumber \\
d_{i,j} & = & \left(\frac{w_{i}}{\mu}-\frac{1}{\delta_{i,j}}\right)^{+},j=1,...,N_{i},\label{eq:WFLd}
\end{eqnarray}
where $\mu$ is chosen such that $\sum_{i=1}^{L}\sum_{j=1}^{N_{i}}d_{i,j}=\sum_{i=1}^{L}\left(P_{T}^{i}-\sum_{l=1}^{i-1}\textrm{Tr}\left(\hat{\mathbf{\Sigma}}_{i,l}\mathbf{W}_{i,l}\right)\right)$.
The calculation of $\mu$ is easy: 1) Let $(i,j)$ index the $j^{\text{th}}$
eigen-channel of the $i^{\text{th}}$ network's $i^{\text{th}}$ link.
Initialize the set of indices of channels with non-negative power
as $\Gamma=\left\{ \left(1,1\right),...,\left(1,N_{1}\right),...\left(L,1\right),...,\right.$
$\left.\left(L,N_{L}\right)\right\} $ to include all eigen-channels;
2) Calculate $\mu$ for the channels in $\Gamma$ without $(\cdot)^{+}$
operation; 3) For all $(n,m)\in\Gamma$, if $d_{n,m}<0$, fix it with
zero power $d_{n,m}=0$, delete $\left(n,m\right)$ from $\Gamma$.
Repeat step 2) until all $d_{n,m}\ge0$. Then the optimal solution
of problem (\ref{eq:WSRM-Parallel}) is 
\begin{equation}
\mathbf{\hat{\mathbf{\Sigma}}}_{i,i}^{'}=\mathbf{\Omega}_{i,i}^{-1/2}\mathbf{F}_{i}\mathbf{D}_{i}\mathbf{F}_{i}^{\dagger}\mathbf{\Omega}_{i,i}^{-1/2},\ i=1,...,L.\label{eq:optsighsub}
\end{equation}
By Theorem \ref{thm:FequRGWF}, if $\mathbf{\Sigma}\left(1\right)=\left(\mathbf{\Sigma}_{1,1},...,\mathbf{\Sigma}_{L,L}\right)$
does not satisfy the polite water-filling structure or the polite
water-filling levels are not proportional to the weights, nor%
{} does $\left(\mathbf{\hat{\mathbf{\Sigma}}}_{1,1},...,\mathbf{\hat{\Sigma}}_{L,L}\right)$,
which means that it is not the optimal solution of problem (\ref{eq:WSRM-Parallel}).
Therefore if we let $\mathbf{\hat{\Sigma}}_{i,1:i}^{'}=\left(\mathbf{\hat{\Sigma}}_{i,1},...,\mathbf{\hat{\Sigma}}_{i,i-1},\mathbf{\hat{\Sigma}}_{i,i}^{'}\right),\forall i$,
we must have $\sum_{i=1}^{L}w_{i}\mathcal{\hat{I}}_{i,i}\left(\hat{\mathbf{\Sigma}}_{i,1:i}^{'},\mathbf{\Phi}^{T}\right)>\sum_{i=1}^{L}w_{i}\mathcal{\hat{I}}_{i,i}\left(\hat{\mathbf{\Sigma}}_{i,1:i},\mathbf{\Phi}^{T}\right)$.

\textbf{Step} 3: Improve the objective function by the covariance
transformation from $\mathbf{\hat{\Sigma}}_{i,1:i}^{'}$ to $\mathbf{\Sigma}_{i,1:i}^{'}=\left(\mathbf{\Sigma}_{i,1}^{'},...,\mathbf{\Sigma}_{i,i}^{'}\right)$
for all the $L$ sub-networks. Let $\mathbf{\Sigma}_{i,1:L}^{'}=\left(\mathbf{\Sigma}_{i,1}^{'},...,\mathbf{\Sigma}_{i,i}^{'},\mathbf{\Sigma}_{i,i+1},...,\right.$
$\left.\mathbf{\Sigma}_{i,L}\right),\forall i$. By Theorem \ref{thm:linear-color-dual},
we have $\mathcal{I}_{i,l}\left(\mathbf{\Sigma}_{i,1:i}^{'},\mathbf{\Phi}\right)\geq\mathcal{\hat{I}}_{i,l}\left(\mathbf{\hat{\Sigma}}_{i,1:i}^{'},\mathbf{\Phi}^{T}\right),\ l=1,...,i$
and $\sum_{i=1}^{L}\sum_{l=1}^{i}$ $\textrm{Tr}\left(\mathbf{\Sigma}_{i,l}^{'}\right)=\sum_{i=1}^{L}\left(\sum_{l=1}^{i-1}\textrm{Tr}\left(\mathbf{\hat{\Sigma}}_{i,l}\mathbf{W}_{i,l}\right)+\textrm{Tr}\left(\mathbf{\hat{\Sigma}}_{i,i}^{'}\mathbf{W}_{i,i}\right)\right)$
$=\sum_{i=1}^{L}P_{T}^{i}$, which implies that the sum power of $\left(\mathbf{\Sigma}_{1,1:L}^{'},...,\mathbf{\Sigma}_{L,1:L}^{'}\right)$
is still $LP_{T}$. Note that the first $i$ links cause no interference
to other links and $\mathcal{I}_{i,l}\left(\mathbf{\Sigma}_{i,1:L}^{'},\mathbf{\Phi}\right)=\mathcal{I}_{i,l}\left(\mathbf{\Sigma}_{i,1:i}^{'},\mathbf{\Phi}\right),l=1,...,i$.
Combining the above facts and the second property in Lemma \ref{lem:Propfmod},
the objective function can be strictly increased by updating $\mathbf{\Sigma}\left(k\right)$'s
as 
\begin{eqnarray}
\mathbf{\Sigma}\left(k\right) & = & \frac{1}{L}\sum_{l=1}^{L}\mathbf{\Sigma}^{'}\left(l\right),1\leq k\leq L,\label{eq:UpdateSigma}
\end{eqnarray}
where $\left(\mathbf{\Sigma}^{'}\left(1\right),...,\mathbf{\Sigma}^{'}\left(L\right)\right)$
is obtained from $\left(\mathbf{\Sigma}_{1,1:L}^{'},...,\mathbf{\Sigma}_{L,1:L}^{'}\right)$
according to (\ref{eq:MApsigma}); and the updated $\mathbf{\Sigma}\left(k\right)$'s
satisfy $\sum_{l=1}^{L}\textrm{Tr}\left(\mathbf{\Sigma}_{l}\left(k\right)\right)=P_{T},\forall k$.

This contradiction proves that the optimal $\mathbf{\tilde{\Sigma}}_{l}\left(k\right),\forall k,l$
must satisfy the polite water-filling structure with the polite water-filling
level $w_{l}/\tilde{\mu}_{k}$. The condition that $\tilde{\mu}_{k}=\tilde{\mu},\forall k$
can also be proved by contradiction. If $\tilde{\mu}_{k}$'s are different,
the polite water-filling levels of the input covariance matrices $\mathbf{\tilde{\Sigma}}_{i,1:L}$
for the $i^{\textrm{th}},\forall i$ network are not proportional
to the weights. By Lemma \ref{lem:KKTiTree}, $\mathbf{\tilde{\Sigma}}_{i,1:L}$
is not the optimal solution of problem (\ref{eq:WSRMP}) with sum
power constraint $\sum_{l=1}^{L}\textrm{Tr}\left(\mathbf{\tilde{\Sigma}}_{i,l}\right)$,
i.e., there exists $\mathbf{\Sigma}_{i,1:L}^{'}$ achieving a higher
weighted sum-rate for the $i^{\textrm{th}},\forall i$ network with
the same sum power. Then after updating $\mathbf{\Sigma}\left(k\right)$'s
from $\mathbf{\Sigma}_{i,1:L}^{'}$'s using (\ref{eq:UpdateSigma}),
the objective function of problem (\ref{eq:MWSRM}) is strictly increased.
This contradiction completes the proof for the necessity part.

Now, we prove the sufficiency part. For convenience, we map $\mathbf{\tilde{\Sigma}}\left(1\right),...,\mathbf{\tilde{\Sigma}}\left(L\right)$
to $\left(\mathbf{\tilde{\Sigma}}_{1,1:L},...,\right.$ $\left.\mathbf{\tilde{\Sigma}}_{L,1:L}\right)$
using (\ref{eq:MApsigma}). Then the Lagrangian of problem (\ref{eq:MWSRM})
is
\begin{align*}
 & L\left(\mu_{1:L},\mathbf{\Theta}_{1:L,1:L},\mathbf{\Sigma}_{1:L,1:L}\right)\\
= & \sum_{i=1}^{L}\sum_{l=1}^{L}w_{l}\textrm{log}\left|\mathbf{I}+\mathbf{H}_{l,l}\mathbf{\Sigma}_{i,l}\mathbf{H}_{l,l}^{\dagger}\mathbf{\Omega}_{i,l}^{-1}\right|\\
 & +\sum_{k=1}^{L}\mu_{k}\left(P_{T}-\sum_{l=1}^{L}\textrm{Tr}\left(\mathbf{\Sigma}_{[l-k]_{L},l}\right)\right)\\
 & +\sum_{i=1}^{L}\sum_{l=1}^{L}\textrm{Tr}\left(\mathbf{\Sigma}_{i,l}\mathbf{\Theta}_{i,l}\right).
\end{align*}
The KKT conditions are
\begin{eqnarray}
\nabla_{\mathbf{\Sigma}_{i,l}}L=0; &  & \sum_{l=1}^{L}\textrm{Tr}\left(\mathbf{\Sigma}_{[l-k]_{L},l}\right)=P_{T};\nonumber \\
\textrm{Tr}\left(\mathbf{\Sigma}_{i,l}\mathbf{\Theta}_{i,l}\right)=0; &  & \mu_{k}\geq0,\:\mathbf{\Sigma}_{i,l},\mathbf{\Theta}_{i,l}\succeq0;\label{eq:KKT1}
\end{eqnarray}
for all $i,k,l=1,...,L$. The condition $\nabla_{\mathbf{\Sigma}_{i,l}}L=0$
can be expressed as 
\begin{align}
 & \sum_{k\neq l}\frac{w_{k}}{\mu_{[l-i]_{L}}}\mathbf{\Phi}_{k,l}\mathbf{H}_{k,l}^{\dagger}\nonumber \\
 & \times\left(\mathbf{\Omega}_{i,k}^{-1}-\left(\mathbf{\Omega}_{i,k}+\mathbf{H}_{k,k}\mathbf{\Sigma}_{i,k}\mathbf{H}_{k,k}^{\dagger}\right)^{-1}\right)\mathbf{H}_{k,l}+\mathbf{I}\nonumber \\
= & \frac{w_{l}}{\mu_{[l-i]_{L}}}\mathbf{H}_{l,l}^{\dagger}\left(\mathbf{\Omega}_{i,l}+\mathbf{H}_{l,l}\mathbf{\Sigma}_{i,l}\mathbf{H}_{l,l}^{\dagger}\right)^{-1}\mathbf{H}_{l,l}+\frac{1}{\mu_{[l-i]_{L}}}\mathbf{\Theta}_{i,l}.\label{eq:DLequ}
\end{align}
Since $\mathbf{\tilde{\Sigma}}_{i,1:L}$ satisfies the polite water-filling
structure with polite water-filling levels given by $\left(w_{l}/\tilde{\mu}\right)$'s,
by Theorem \ref{thm:At-the-boundary}, the dual input covariance matrices
$\mathbf{\tilde{\hat{\Sigma}}}_{i,1:L}$ can be expressed as
\begin{equation}
\mathbf{\tilde{\hat{\Sigma}}}_{i,k}=\frac{w_{k}}{\tilde{\mu}}\left(\tilde{\mathbf{\Omega}}_{i,k}^{-1}-\left(\tilde{\mathbf{\Omega}}_{i,k}+\mathbf{H}_{k,k}\tilde{\mathbf{\Sigma}}_{i,k}\mathbf{H}_{k,k}^{\dagger}\right)^{-1}\right),\forall k.\label{eq:SigmahM}
\end{equation}
Then we have
\begin{align*}
 & \sum_{k\neq l}\frac{w_{k}}{\tilde{\mu}}\mathbf{\Phi}_{k,l}\mathbf{H}_{k,l}^{\dagger}\\
 & \times\left(\tilde{\mathbf{\Omega}}_{i,k}^{-1}-\left(\tilde{\mathbf{\Omega}}_{i,k}+\mathbf{H}_{k,k}\tilde{\mathbf{\Sigma}}_{i,k}\mathbf{H}_{k,k}^{\dagger}\right)^{-1}\right)\mathbf{H}_{k,l}+\mathbf{I}\\
= & \sum_{k\neq l}\mathbf{\Phi}_{k,l}\mathbf{H}_{k,l}^{\dagger}\mathbf{\tilde{\hat{\Sigma}}}_{i,k}\mathbf{H}_{k,l}+\mathbf{I}=\tilde{\mathbf{\hat{\mathbf{\Omega}}}}_{i,l}.
\end{align*}
Choose the dual variables $\mu_{1:L}$ as $\mu_{k}=\tilde{\mu},\forall k$
and substitute $\mathbf{\tilde{\Sigma}}_{i,1:L}$ into condition (\ref{eq:DLequ})
to obtain 
\begin{equation}
\tilde{\mathbf{\hat{\mathbf{\Omega}}}}_{i,l}=\frac{w_{l}}{\tilde{\mu}}\mathbf{H}_{l,l}^{\dagger}\left(\tilde{\mathbf{\Omega}}_{i,l}+\mathbf{H}_{l,l}\tilde{\mathbf{\Sigma}}_{i,l}\mathbf{H}_{l,l}^{\dagger}\right)^{-1}\mathbf{H}_{l,l}+\frac{1}{\tilde{\mu}}\mathbf{\Theta}_{i,l}.\label{eq:KeyKKT}
\end{equation}
Because (\ref{eq:KeyKKT}) is also the KKT condition of the single-user
polite water-filling problem and $\mathbf{\tilde{\Sigma}}_{i,l}$
has polite water-filling structure with polite water-filling level
$w_{l}/\tilde{\mu}$ by the optimality condition, $\left(\mathbf{\tilde{\Sigma}}_{1,1:L},...,\mathbf{\tilde{\Sigma}}_{L,1:L}\right)$
satisfies condition (\ref{eq:KeyKKT}). It can be verified that $\left(\mathbf{\tilde{\Sigma}}_{1,1:L},...,\mathbf{\tilde{\Sigma}}_{L,1:L}\right)$
also satisfies all other KKT conditions in (\ref{eq:KKT1}). Since
problem (\ref{eq:MWSRM}) is convex, KKT conditions are sufficient
for global optimality. This completes the proof for Theorem \ref{thm:optimality}.

\subsubsection{Algorithm P}

The proof of Theorem \ref{thm:optimality} actually gives an algorithm
to find the optimal solution of problem (\ref{eq:MWSRM}) by monotonically
increasing the objective function. We refer to it as\emph{ }Algorithm
P and summarize it in Table \ref{tab:table1}, where P stands for
Polite. The convergence and optimality of Algorithm P is proved in
the theorem below.%

\begin{table}
\caption{\label{tab:table1}Algorithm P (for iTree Networks with Concave $f(\cdot)$)}

\centering{}%
\begin{tabular}{l}
\hline 
{\small{Choose a feasible initial point $\mathbf{\Sigma}\left(k\right)=\left(\mathbf{\Sigma}_{1},...,\mathbf{\Sigma}_{L}\right),k=1,...,L$ }}\tabularnewline
{\small{such that $\mathbf{\Sigma}_{l}\succeq0,\forall l$ and $\sum_{l=1}^{L}\textrm{Tr}\left(\mathbf{\Sigma}_{l}\right)\le P_{T}$.}}\tabularnewline
\textbf{\small{While}}{\small{ not converge }}\textbf{\small{do}}{\small{ }}\tabularnewline
{\small{$\;$1. Calculate $\hat{\mathbf{\Sigma}}_{L,1:L-1}$ by the
covariance transformation applied}}\tabularnewline
{\small{$\;$$\;$$\;$$\;$to the $L^{\text{th}}$ sub-network. Because
all $\mathbf{\Sigma}\left(k\right)$'s are the same, }}\tabularnewline
{\small{$\;$$\;$$\;$$\;$the $\hat{\mathbf{\Sigma}}_{i,1:i-1}$'s
for other sub-networks can be obtained by}}\tabularnewline
{\small{$\;$$\;$$\;$$\;$$\;$$\;$$\hat{\mathbf{\Sigma}}_{i,1:i-1}=\left(\hat{\mathbf{\Sigma}}_{L,1},...,\hat{\mathbf{\Sigma}}_{L,i-1}\right),i=1,...,L-1$}}\tabularnewline
{\small{$\;$2. Solve for the optimal $\mathbf{\hat{\mathbf{\Sigma}}}_{i,i}^{'}$'s
in the optimization problem (\ref{eq:WSRM-Parallel})}}\tabularnewline
{\small{$\;$$\;$$\;$$\;$by polite water-filling.}}\tabularnewline
{\small{$\;$3. For $\forall i$, calculate $\mathbf{\Sigma}_{i,1:i}^{'}$
by the covariance transformation of }}\tabularnewline
{\small{$\;$$\;$$\;$$\;$$\mathbf{\hat{\mathbf{\Sigma}}}_{i,1:i}^{'}=\left(\mathbf{\hat{\Sigma}}_{i,1},...,\mathbf{\hat{\Sigma}}_{i,i-1},\mathbf{\hat{\Sigma}}_{i,i}^{'}\right)$
applied to the $i^{\text{th}}$ sub-network. }}\tabularnewline
{\small{$\;$$\;$$\;$$\;$Obtain $\left(\mathbf{\Sigma}^{'}\left(1\right),...,\mathbf{\Sigma}^{'}\left(L\right)\right)$
from $\left(\mathbf{\Sigma}_{1,1:L}^{'},...,\mathbf{\Sigma}_{L,1:L}^{'}\right)$
according}}\tabularnewline
{\small{$\;$$\;$$\;$$\;$to (\ref{eq:MApsigma}), where $\mathbf{\Sigma}_{i,1:L}^{'}=\left(\mathbf{\Sigma}_{i,1}^{'},...,\mathbf{\Sigma}_{i,i}^{'},\mathbf{\Sigma}_{i,i+1},...,\mathbf{\Sigma}_{i,L}\right),\forall i$.}}\tabularnewline
{\small{$\;$4. Update $\mathbf{\Sigma}\left(k\right)$'s using (\ref{eq:UpdateSigma}).}}\tabularnewline
\textbf{\small{End}}\tabularnewline
\end{tabular}
\end{table}

\begin{thm}
\label{thm:opttreenet}For any iTree network with concave weighted
sum-rate function $f(\cdot)$, Algorithm P converges to the optimal
solution of problem (\ref{eq:MWSRM}) and problem (\ref{eq:WSRMP}).\end{thm}
\begin{IEEEproof}
In each iteration, Algorithm P monotonically increases the objective
function in problem (\ref{eq:MWSRM}). Since the objective is upper
bounded, Algorithm P must converge to a fixed point. At the fixed
point, $\mathbf{\Sigma}\left(1\right)$ must satisfy the optimality
conditions in Theorem \ref{thm:optimality} because otherwise, the
algorithm can strictly increase the objective function. Then at the
fixed point, all $\mathbf{\Sigma}\left(k\right)$'s satisfy the optimality
conditions in Theorem \ref{thm:optimality} because at the end of
each iteration, we make all $\mathbf{\Sigma}\left(k\right)$'s the
same in (\ref{eq:UpdateSigma}). By Theorem \ref{thm:optimality},
the fixed point is the optimal solution of problem (\ref{eq:MWSRM}).
And by Lemma \ref{lem:Propfmod}, it must also be the optimal solution
of problem (\ref{eq:WSRMP}).
\end{IEEEproof}

\subsubsection{Algorithm P1}

Another algorithm named P1 can be designed with improved convergence
speed using a similar method as in \cite{Kobayashi_JSAC06_ITWMISOBC}
by modifying the updating step in (\ref{eq:UpdateSigma}) at the cost
of a simple linear search as 

\begin{equation}
\mathbf{\Sigma}\left(k\right)=\frac{\tilde{\beta}P_{T}}{P_{1}}\mathbf{\Sigma}^{'}\left(1\right)+\frac{\left(1-\tilde{\beta}\right)P_{T}}{\sum_{l=2}^{L}P_{l}}\sum_{l=2}^{L}\mathbf{\Sigma}^{'}\left(l\right),\forall k,\label{eq:impupdatr}
\end{equation}
where $P_{k}=\sum_{l=1}^{L}\textrm{Tr}\left(\mathbf{\Sigma}_{l}^{'}\left(k\right)\right),\forall k$;
and $\tilde{\beta}$ is the optimal solution of the optimization problem

\[
\underset{\beta\in\left[P_{1}/\left(LP_{T}\right),1\right]}{\textrm{max}}f\left(\frac{\beta P_{T}}{P_{1}}\mathbf{\Sigma}^{'}\left(1\right)+\frac{\left(1-\beta\right)P_{T}}{\sum_{l=2}^{L}P_{l}}\sum_{l=2}^{L}\mathbf{\Sigma}^{'}\left(l\right)\right).
\]
It is clear that the update (\ref{eq:impupdatr}) yields at least
the same improvement of the objective function as the update (\ref{eq:UpdateSigma})
because if we let $\tilde{\beta}=P_{1}/\left(LP_{T}\right)$, (\ref{eq:impupdatr})
reduces to (\ref{eq:UpdateSigma}). Algorithm P1 is provably convergent
and is observed to have much faster convergence speed in the simulations.

\subsection{\label{sub:alg-B-MAC}Algorithms for General B-MAC Networks}

\subsubsection{Algorithm PT}

We obtain Algorithm PT, standing for Polite\nobreakdash-Transformation,
for general B-MAC networks by a modification of Algorithm P1. It is
observed that $\tilde{\beta}$ in (\ref{eq:impupdatr}) is always
close to 1 in simulation, which suggests a more aggressive iterative
algorithm by modifying Algorithm P1 so that each iteration is simply
a polite water-filling in the forward links and a covariance transformation
in the reverse links. Algorithm PT can be reorganized in a simpler
form as in Table \ref{tab:table2} in terms of one network without
expanding to $L$ networks. 

After convergence, the solution of Algorithm PT satisfies the KKT
conditions of \textbf{WSRMP}.

\begin{table}
\caption{\label{tab:table2}Algorithm PT (for B-MAC Networks)}

\centering{}%
\begin{tabular}{l}
\hline 
{\small{Initialize $\mathbf{\Sigma}_{1:L}$ and $\hat{\mathbf{\Sigma}}_{1:L}$
with zeros or other values such that $\mathbf{\Sigma}_{l},\mathbf{\hat{\Sigma}}_{l}\succeq0,\forall l$.}}\tabularnewline
\textbf{\small{While}}{\small{ not converge }}\textbf{\small{do}}{\small{ }}\tabularnewline
{\small{$\;$1. Polite water-filling in the forward links }}\tabularnewline
{\small{$\;$$\;$a. For $\forall l$, obtain $\mathbf{\Omega}_{l}$
and $\hat{\mathbf{\Omega}}_{l}$ from $\mathbf{\Sigma}_{1:L}$ and
$\hat{\mathbf{\Sigma}}_{1:L}$ using (\ref{eq:whiteMG}) and (\ref{eq:WhiteMRV})}}\tabularnewline
{\small{$\;$$\;$$\;$$\;$$\;$$\;$respectively. Perform thin SVD
$\mathbf{\Omega}_{l}^{-1/2}\mathbf{H}_{l,l}\hat{\mathbf{\Omega}}_{l}^{-1/2}=\mathbf{F}_{l}\mathbf{\Delta}_{l}\mathbf{G}_{l}^{\dagger}$. }}\tabularnewline
{\small{$\;$$\;$b. Let $\delta_{l,i}$ be the $i^{th}$ diagonal
element of $\mathbf{\Delta}_{l}^{2}$ and let $\rho_{l,i}$ be the}}\tabularnewline
{\small{$\;$$\;$$\;$$\;$$\;$$\;$norm square of the $i^{th}$
column of $\hat{\mathbf{\Omega}}_{l}^{-1/2}\mathbf{G}_{l}$. Obtain
$\mathbf{D}_{l}$ as}}\tabularnewline
{\small{$\;$$\;$$\;$$\;$$\;$$\;$$\;$$\;$$\mathbf{D}_{l}=\textrm{diag}\left(d_{l,1},...,d_{l,N_{l}}\right)$,}}\tabularnewline
{\small{$\;$$\;$$\;$$\;$$\;$$\;$$\;$$\;$$d_{l,i}=\left(\frac{w_{l}}{\mu}-\frac{1}{\delta_{l,i}}\right)^{+},i=1,...,N_{l}$,}}\tabularnewline
{\small{$\;$$\;$$\;$$\;$$\;$$\;$where $\mu$ is chosen such
that $\sum_{l=1}^{L}\sum_{i=1}^{N_{l}}\rho_{l,i}d_{l,i}=P_{T}.$}}\tabularnewline
{\small{$\;$$\;$c. Update $\mathbf{\Sigma}_{l}$'s as}}\tabularnewline
{\small{$\;$$\;$$\;$$\;$$\;$$\;$$\;$$\;$$\mathbf{\Sigma}_{l}=\hat{\mathbf{\Omega}}_{l}^{-1/2}\mathbf{G}_{l}\mathbf{D}_{l}\mathbf{G}_{l}^{\dagger}\hat{\mathbf{\Omega}}_{l}^{-1/2},\forall l$.}}\tabularnewline
{\small{$\;$2. Covariance transformation in the reverse links }}\tabularnewline
{\small{$\;\;$$\;$$\;$$\;$$\;$Obtain $\hat{\mathbf{\Sigma}}_{1:L}$
from the covariance transformation of $\mathbf{\Sigma}_{1:L}$, using
(\ref{eq:CovTrans}).}}\tabularnewline
\textbf{\small{End}}\tabularnewline
\end{tabular}
\end{table}

\begin{thm}
\label{thm:ALGP2KKT}Apply Algorithm PT to solve problem (\ref{eq:WSRMP}),
whose Lagrange function is
\begin{align*}
L\left(\mu,\mathbf{\Theta}_{1:L},\mathbf{\Sigma}_{1:L}\right) & =\sum_{l=1}^{L}w_{l}\textrm{log}\left|\mathbf{I}+\mathbf{H}_{l,l}\mathbf{\Sigma}_{l}\mathbf{H}_{l,l}^{\dagger}\mathbf{\Omega}_{l}^{-1}\right|\\
 & +\mu\left(P_{T}-\sum_{l=1}^{L}\textrm{Tr}\left(\mathbf{\Sigma}_{l}\right)\right)+\sum_{l=1}^{L}\textrm{Tr}\left(\mathbf{\Sigma}_{l}\mathbf{\Theta}_{l}\right),
\end{align*}
and the KKT conditions are
\begin{eqnarray}
\nabla_{\mathbf{\Sigma}_{l}}L & = & w_{l}\mathbf{H}_{l,l}^{\dagger}\left(\mathbf{\Omega}_{l}+\mathbf{H}_{l,l}\mathbf{\Sigma}_{l}\mathbf{H}_{l,l}^{\dagger}\right)^{-1}\mathbf{H}_{l,l}+\mathbf{\Theta}_{l}\nonumber \\
 &  & -\mu\mathbf{I}-\sum_{k\neq l}w_{k}\mathbf{\Phi}_{k,l}\mathbf{H}_{k,l}^{\dagger}\nonumber \\
 &  & \times\left(\mathbf{\Omega}_{k}^{-1}-\left(\mathbf{\Omega}_{k}+\mathbf{H}_{k,k}\mathbf{\Sigma}_{k}\mathbf{H}_{k,k}^{\dagger}\right)^{-1}\right)\mathbf{H}_{k,l}\nonumber \\
 & = & \mathbf{0};\nonumber \\
\sum_{l=1}^{L}\textrm{Tr}\left(\mathbf{\Sigma}_{l}\right) & = & P_{T};\nonumber \\
\textrm{Tr}\left(\mathbf{\mathbf{\Sigma}}_{l}\mathbf{\Theta}_{l}\right) & = & 0;\nonumber \\
\mu\geq0; &  & \mathbf{\Sigma}_{l},\mathbf{\Theta}_{l}\succeq\mathbf{0};\label{eq:kktwsr-1}
\end{eqnarray}
for all $l=1,...,L$. If Algorithm PT converges, the solution of the
algorithm satisfies the above KKT conditions, and thus is a stationary
point. \end{thm}
\begin{IEEEproof}
After convergence, the solution $\mathbf{\bar{\Sigma}}_{1:L}$ of
Algorithm PT satisfies the polite water-filling structure and the
polite water-filling levels $\bar{\nu}_{l}$'s are proportional to
the weights, i.e., $\bar{\nu}_{l}=w_{l}/\bar{\mu}$. The rest of the
proof is similar to the proof of the sufficiency part in Theorem \ref{thm:optimality}.
\end{IEEEproof}

If the weighted sum-rate function is not concave, Algorithm PT may
not converge to the global optimum since it may get stuck at some
other stationary point. To evaluate the performance of the algorithm
for non-convex cases, we define \textit{pseudo global optimum} as
the best solution of the algorithm with many random initial points.
If the performance of the algorithm obtained from a single initial
point is close to the\emph{ }\textit{pseudo global optimum}, the solution
of the algorithm is said to be \textit{good}. Simulations show that
Algorithm PT usually gives \textit{good} solution.

We prove for two simple examples that Algorithm PT converges monotonically
to a stationary point. More discussion on the convergence is given
in the Remark \ref{rem:A-nontrivial-conv}.
\begin{example}
Consider weighted sum-rate maximization for a 2-user SISO MAC with
$w_{2}>w_{1}$ and user 1 decoded first. Let $p_{i},i=1,2$ be the
transmit powers. The optimization problem is

\begin{eqnarray}
\underset{p_{1},p_{2}}{\textrm{max}}f & \triangleq & w_{1}\textrm{log}\left(\frac{1+p_{1}g_{1}+p_{2}g_{2}}{1+p_{2}g_{2}}\right)+w_{2}\textrm{log}\left(1+p_{2}g_{2}\right),\nonumber \\
 &  & \:\textrm{s.t.}\, p_{1}+p_{2}\leq P_{T},\label{eq:Ex1}
\end{eqnarray}
where $g_{i}$ is the channel gain. Algorithm PT updates the transmit
powers as
\begin{eqnarray*}
p_{1}^{'} & = & \left(\nu w_{1}-\frac{1+p_{2}g_{2}}{g_{1}}\right)^{+}\\
p_{2}^{'} & = & \left(\nu w_{2}\frac{1+p_{2}g_{2}}{1+p_{1}g_{2}+p_{2}g_{2}}-\frac{1}{g_{2}}\right)^{+},
\end{eqnarray*}
and $\nu$ is chosen such that $p_{1}^{'}+p_{2}^{'}=P_{T}$. Suppose
$\frac{\partial f}{\partial p_{1}}>\frac{\partial f}{\partial p_{2}}$.
Then $\frac{1+p_{1}g_{2}+p_{2}g_{2}}{w_{2}g_{2}}>\frac{1+p_{1}g_{1}+p_{2}g_{2}}{w_{1}g_{1}}$.
If $\nu=\frac{1+p_{1}g_{1}+p_{2}g_{2}}{w_{1}g_{1}}$, then $p_{1}^{'}=p_{1}$
and $p_{2}^{'}<p_{2}$, i.e., we should increase the water-filling
level. Similarly if $\nu=\frac{1+p_{1}g_{2}+p_{2}g_{2}}{w_{2}g_{2}}$,
then $p_{1}^{'}>p_{1}$ and $p_{2}^{'}=p_{2}$, i.e., the water-filling
level must be decreased. Then we have $\frac{1+p_{1}g_{1}+p_{2}g_{2}}{w_{1}g_{1}}<\nu<\frac{1+p_{1}g_{2}+p_{2}g_{2}}{w_{2}g_{2}}$.
Because of $\nu>\frac{1+p_{1}g_{1}+p_{2}g_{2}}{w_{1}g_{1}}$, $p_{1}$
is increased and thus the transmit powers are updated according to
the gradient direction. But if we increase $p_{1}$ too much, the
objective may still be decreased. However it follows from $\nu<\frac{1+p_{1}g_{2}+p_{2}g_{2}}{w_{2}g_{2}}$
that $\frac{\partial f}{\partial p_{1}}>\frac{\partial f}{\partial p_{2}}$
after the update. Hence the objective is strictly increased after
each update until $\frac{\partial f}{\partial p_{1}}=\frac{\partial f}{\partial p_{2}}$,
which is the stationary point of problem (\ref{eq:Ex1}).
\end{example}

Similarly, one can prove the result for sum-rate maximization of a
2-user SISO Z channel.

\subsubsection{Algorithm PP}

We also designed Algorithm PP, standing for Polite\nobreakdash-Polite,
that uses polite water-filling for both the forward and reverse links.
It is modified from Algorithm PT. The algorithm is shown in Table
\ref{tab:table3}.\textcolor{black}{{} It can be proved that Theorem
\ref{thm:ALGP2KKT} also holds for Algorithm PP. }An advantage of
Algorithm PP is that the polite water-filling procedure to obtain
\textcolor{black}{$\hat{\mathbf{\Sigma}}_{1:L}$} has lower complexity
than the covariance transformation in Algorithm PT, as discussed in
Section \ref{sub:Complexity-Analysis}.

\begin{table}
\caption{\label{tab:table3}Algorithm PP (for B-MAC Networks)}

\centering{}%
\begin{tabular}{l}
\hline 
{\small{Initialize $\mathbf{\hat{\Sigma}}_{1:L}$ and $\mathbf{\Omega}_{l}$'s
such that $\sum_{l=1}^{L}\textrm{Tr}\left(\mathbf{\hat{\Sigma}}_{l}\right)\le P_{T}$,
}}\tabularnewline
{\small{$\mathbf{\hat{\Sigma}}_{l}\succeq0,\forall l$ and $\mathbf{\Omega}_{l}=\mathbf{I},\forall l$.}}\tabularnewline
\textbf{\small{While}}{\small{ not converge }}\textbf{\small{do}}{\small{ }}\tabularnewline
{\small{$\;$1. Polite water-filling in the forward links }}\tabularnewline
{\small{$\;$$\;$a. For $\forall l$, obtain $\hat{\mathbf{\Omega}}_{l}$
from $\hat{\mathbf{\Sigma}}_{1:L}$ using (\ref{eq:WhiteMRV}).}}\tabularnewline
{\small{$\;$$\;$$\;$$\;$$\;$$\;$Perform thin SVD $\mathbf{\Omega}_{l}^{-1/2}\mathbf{H}_{l,l}\hat{\mathbf{\Omega}}_{l}^{-1/2}=\mathbf{F}_{l}\mathbf{\Delta}_{l}\mathbf{G}_{l}^{\dagger}$. }}\tabularnewline
{\small{$\;$$\;$b. Let $\delta_{l,i}$ be the $i^{th}$ diagonal
element of $\mathbf{\Delta}_{l}^{2}$ and let $\rho_{l,i}$ be the}}\tabularnewline
{\small{$\;$$\;$$\;$$\;$$\;$$\;$norm square of the $i^{th}$
column of $\hat{\mathbf{\Omega}}_{l}^{-1/2}\mathbf{G}_{l}$. Obtain
$\mathbf{D}_{l}$ as}}\tabularnewline
{\small{$\;$$\;$$\;$$\;$$\;$$\;$$\;$$\;$$\mathbf{D}_{l}=\textrm{diag}\left(d_{l,1},...,d_{l,N_{l}}\right)$,}}\tabularnewline
{\small{$\;$$\;$$\;$$\;$$\;$$\;$$\;$$\;$$d_{l,i}=\left(\frac{w_{l}}{\mu}-\frac{1}{\delta_{l,i}}\right)^{+},i=1,...,N_{l}$,}}\tabularnewline
{\small{$\;$$\;$$\;$$\;$$\;$$\;$where $\mu$ is chosen such
that }}\tabularnewline
{\small{$\;$$\;$$\;$$\;$$\;$$\;$$\;$$\;$$\sum_{l=1}^{L}\sum_{i=1}^{N_{l}}\rho_{l,i}d_{l,i}=P_{T}.$ }}\tabularnewline
{\small{$\;$$\;$c. Update $\mathbf{\Sigma}_{l}$'s as}}\tabularnewline
{\small{$\;$$\;$$\;$$\;$$\;$$\;$$\;$$\;$$\mathbf{\Sigma}_{l}=\hat{\mathbf{\Omega}}_{l}^{-1/2}\mathbf{G}_{l}\mathbf{D}_{l}\mathbf{G}_{l}^{\dagger}\hat{\mathbf{\Omega}}_{l}^{-1/2},\forall l$.}}\tabularnewline
{\small{$\;$2. Polite water-filling in the reverse links }}\tabularnewline
{\small{$\;\;$a. For $\forall l$, obtain $\mathbf{\Omega}_{l}$
from $\mathbf{\Sigma}_{1:L}$ using (\ref{eq:whiteMG}).}}\tabularnewline
{\small{$\;$$\;$$\;$$\;$$\;$$\;$Perform thin SVD $\mathbf{\Omega}_{l}^{-1/2}\mathbf{H}_{l,l}\hat{\mathbf{\Omega}}_{l}^{-1/2}=\mathbf{F}_{l}\mathbf{\Delta}_{l}\mathbf{G}_{l}^{\dagger}$. }}\tabularnewline
{\small{$\;$$\;$b. Let $\delta_{l,i}$ be the $i^{th}$ diagonal
element of $\mathbf{\Delta}_{l}^{2}$ and let $\hat{\rho}_{l,i}$
be the}}\tabularnewline
{\small{$\;$$\;$$\;$$\;$$\;$$\;$norm square of the $i^{th}$
column of $\mathbf{\Omega}_{l}^{-1/2}\mathbf{F}_{l}$. Obtain $\mathbf{D}_{l}$
as}}\tabularnewline
{\small{$\;$$\;$$\;$$\;$$\;$$\;$$\;$$\;$$\mathbf{D}_{l}=\textrm{diag}\left(d_{l,1},...,d_{l,N_{l}}\right)$,}}\tabularnewline
{\small{$\;$$\;$$\;$$\;$$\;$$\;$$\;$$\;$$d_{l,i}=\left(\frac{w_{l}}{\mu}-\frac{1}{\delta_{l,i}}\right)^{+},i=1,...,N_{l},$}}\tabularnewline
{\small{$\;$$\;$$\;$$\;$$\;$$\;$where $\mu$ is chosen such
that }}\tabularnewline
{\small{$\;$$\;$$\;$$\;$$\;$$\;$$\;$$\;$$\sum_{l=1}^{L}\sum_{i=1}^{N_{l}}\hat{\rho}_{l,i}d_{l,i}=P_{T}.$ }}\tabularnewline
{\small{$\;$$\;$c. Update $\mathbf{\hat{\Sigma}}_{l}$'s as}}\tabularnewline
{\small{$\;$$\;$$\;$$\;$$\;$$\;$$\;$$\;$$\mathbf{\hat{\Sigma}}_{l}=\mathbf{\Omega}_{l}^{-1/2}\mathbf{F}_{l}\mathbf{D}_{l}\mathbf{F}_{l}^{\dagger}\mathbf{\Omega}_{l}^{-1/2},\forall l$.}}\tabularnewline
\textbf{\small{End}}\tabularnewline
\end{tabular}
\end{table}

\begin{rem}
\label{rem:A-nontrivial-conv}A nontrivial future work is to find
the convergence conditions of Algorithm PT or PP. It has been observed
that they always converge monotonically in iTree networks. It is possible
that the structure of the iTree networks guarantees it, which suggests
that iTree networks has more useful properties to be found. For general
B-MAC networks, one possible solution may be extending the approach
used in \cite{Palomar_09TSP_MIMOIWFgame} for iterative selfish water-filling
to polite water-filling. The physical interpretation of the convergence
conditions in \cite{Palomar_09TSP_MIMOIWFgame} is that the interference
among the links is sufficiently small. Because in the polite water-filling,
the interference among the links is better managed, we conjecture
that the convergence conditions for Algorithm PT or PP will be much
looser, which is verified by the simulations in Section \ref{sec:Simulation-Results},
where Algorithm PP and PT with a single initial point are observed
to converge for most general B-MAC networks simulated. Moreover, they
always converge after few trials to select a good initial point, which
suggests that it is at least possible to prove the local convergence
of the algorithms. Another approach to convergence proof may be interpreting
the polite water-filling as message passing in a factor graph and
derive the conditions of convergence on the graph.%
{} 
\end{rem}

\begin{rem}
The following is the connection to another iterative water-filling
algorithm proposed in \cite[Section IV]{Wei_07ITW_MultiuserWF} for
parallel SISO interference networks. It is derived by directly solving
the KKT conditions of the weighted-sum rate maximization problem.
As expected, the result%
{} also has a term $t_{k}(n)$ related to the interference to other
nodes. In the SISO case, the matrix $\mathbf{\hat{\Omega}}_{l}$ equals
to $1+\frac{t_{k}(n)}{\lambda_{k}}$ in \cite{Wei_07ITW_MultiuserWF}
only for stationary points. Algorithm PT or PP can be used there to
reduce the complexity. On the other hand, we can extend the algorithm
in \cite{Wei_07ITW_MultiuserWF} to MIMO B-MAC networks by directly
solving the KKT conditions for the MIMO case and by using the concept
of polite water-filling. It results in an algorithm that has similar
performance to Algorithm PT but does not take the advantage of the
duality. In the algorithm, we replace the matrix $\mathbf{\hat{\Omega}}_{l}$
by the term 
\begin{align}
 & \mathbf{I}+\sum_{k\neq l}\frac{w_{k}}{\mu}\mathbf{\Phi}_{k,l}\mathbf{H}_{k,l}^{\dagger}\nonumber \\
 & \times\left(\mathbf{\Omega}_{k}^{-1}-\left(\mathbf{\Omega}_{k}+\mathbf{H}_{k,k}\mathbf{\Sigma}_{k}\mathbf{H}_{k,k}^{\dagger}\right)^{-1}\right)\mathbf{H}_{k,l}\label{eq:WYkkt}
\end{align}
in (\ref{eq:kktwsr-1}), avoiding the covariance transformation in
Algorithm PT and reverse link polite water-filling in Algorithm PP.
But as a result, the equivalent channel becomes a function of $\mu$
and in order to perform the polite water-filling, SVD has to be repeatedly
done when searching for $\mu$ to satisfy the power constraint, leading
to a higher complexity than Algorithm PP since the reverse link polite
water-filling in Algorithm PP only requires one SVD for each link.

\end{rem}

\subsection{Complexity Analysis\label{sub:Complexity-Analysis}}

We first give a brief analysis to show the complexity order per iteration
of the proposed algorithms. The main computation complexity lies in
the SVD and matrix inverse operations in the polite water-filling
procedure and in the calculation of the MMSE-SIC receivers. These
operations are performed over the matrices whose dimensions are equal
to the number of the antennas at each node and are not increased with
the number of transmission links $L$. This, however, is not true
for generic optimization methods. We use the order of the total number
of SVD and matrix inverse operations to measure the complexity. 
\begin{itemize}
\item For Algorithm P, in step 1, we need to calculate $\hat{\mathbf{\Sigma}}_{L,1:L-1}$
by the covariance transformation, which has a complexity order of
$\mathcal{O}\left(L\right)$. In step 2, the complexity order of the
polite water-filling procedure to obtain $\mathbf{\hat{\mathbf{\Sigma}}}_{i,i}^{'}$'s
is also $\mathcal{O}\left(L\right)$. And in step 3, the covariance
transformation to calculate $\mathbf{\Sigma}_{i,1:i}^{'},i=1,...,L$
has a complexity order of $\mathcal{O}\left(L^{2}\right)$. Therefore,
the total complexity order is $\mathcal{O}\left(L^{2}\right)$.
\item Algorithm P1 only adds a linear search compared to Algorithm P. Therefore,
it has the same complexity order as Algorithm P.
\item The complexity of Algorithm PT depends on whether the network is iTree
network or not. For iTree network, the covariance transformation to
obtain $\hat{\mathbf{\Sigma}}_{1:L}$ in step 1 and the polite water-filling
to obtain $\mathbf{\Sigma}_{l}$'s in step 2 have a complexity order
of $\mathcal{O}\left(L\right)$ because the reverse link power in
the covariance transformation can be calculated one by one. For other
networks, we need to solve a $\sum_{l=1}^{L}M_{l}$-dimensional linear
equation as in (\ref{eq:qpower}), whose complexity depends on the
density and structure of the cross-talk matrix $\mathbf{\Psi}\left(\mathbf{T},\mathbf{R}\right)$.
In the worst case, the complexity order is $\mathcal{O}\left(L^{3}\right)$.
In other cases such as with triangular or sparse $\mathbf{\Psi}\left(\mathbf{T},\mathbf{R}\right)$,
the complexity is much lower. Fortunately, in practice, $\mathbf{\Psi}\left(\mathbf{T},\mathbf{R}\right)$
is usually sparse for a large wireless network because of path loss
or interference suppression techniques. 
\item The complexity order of Algorithm PP is always $\mathcal{O}\left(L\right)$
since the reverse input covariance matrices are also obtained by polite
water-filling.
\end{itemize}
In practice, Algorithm PT is the choice for iTree networks because
its performance is similar to Algorithm P1 but the complexity is much
lower, while Algorithm PP, which has even lower complexity, is a better
choice for general B-MAC networks because its convergence behavior
is less sensitive to the interference loops and interference strength.
Therefore, Algorithm P and P1 are more of theoretic value. Although
Algorithm P and P1 have higher complexity than the $\mathcal{O}\left(L\right)$
algorithms in \cite{Viswanathan_JSAC03_BCGD,Jindal_IT05_IFBC,Kobayashi_JSAC06_ITWMISOBC,Weiyu_IT06_DualIWF},
they are still much simpler than the standard interior point methods
of complexity order of $\mathcal{O}\left(L^{3}\right)$ \cite{Boyd_04Book_Convex_optimization}
and can be used in more general cases compared to those for MIMO MAC/BC. 

Second, we discuss the number of iterations needed or in other words,
the convergence speed. The number of neighbors of a node in the network
is expected to have little influence on the convergence speed of the
designed algorithms because at each link, the forward and reverse
link interference-plus-noise covariance matrices summarize all interference
no matter how many neighbors of the link are there. The convergence
speed is expected to depend on the girth of the network because each
iteration propagates the influence of the transmission scheme of a
transmitter to its closest neighbors. But in a wireless network, large
girth with uniformly strong channel gains has small probability to
happen. Therefore, the number of iterations needed is small as will
be verified by the simulation results in Section \ref{sec:Simulation-Results}.
Define the accuracy as the difference between the performance of a
certain number of iterations and the maximum, if known. Accuracy versus
iteration numbers gives a closer look at the asymptotic convergence
behavior. As expected, the algorithms in this paper have superior
accuracy because of the optimality of the polite water-filling structure.

In summary, Algorithm PT and PP are expected and have been verified
by simulation to have superior performance, complexity per iteration,
convergence speed, and accuracy than other algorithms that do not
take the advantage of the optimal input structure.

\section{Simulation Results\label{sec:Simulation-Results}}

\begin{figure}
\begin{centering}
\textsf{\includegraphics[clip,width=95mm]{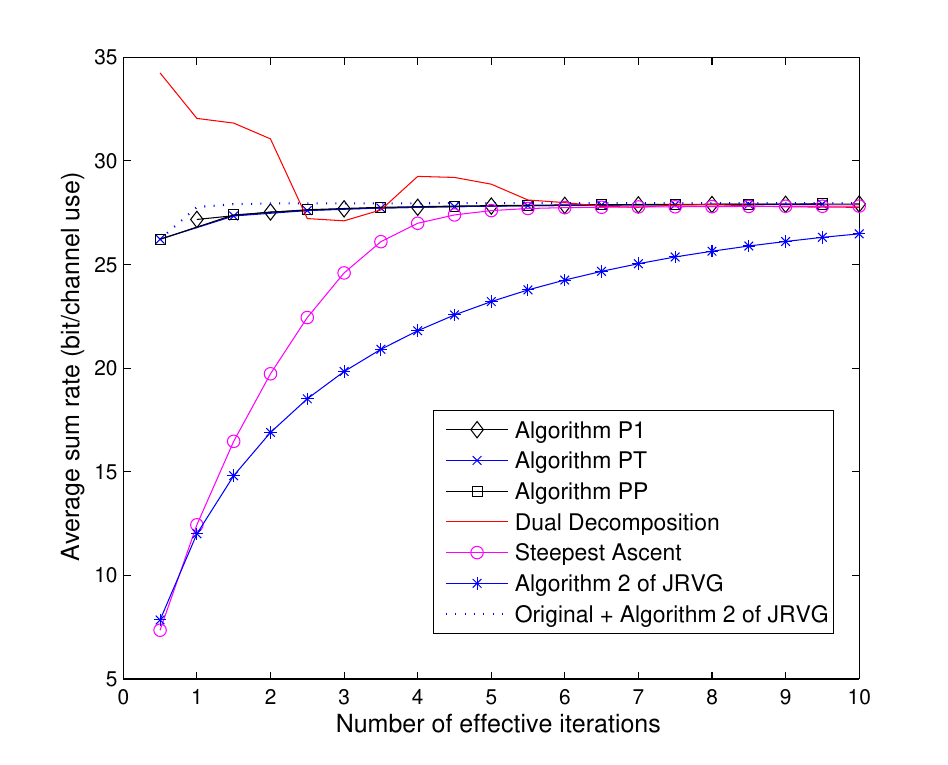}}
\par\end{centering}

\caption{\label{fig:Fig1}Convergence speed comparison of the sum-rate for
a 10-user MAC.}
\end{figure}

\begin{figure}
\begin{centering}
\textsf{\includegraphics[clip,width=95mm]{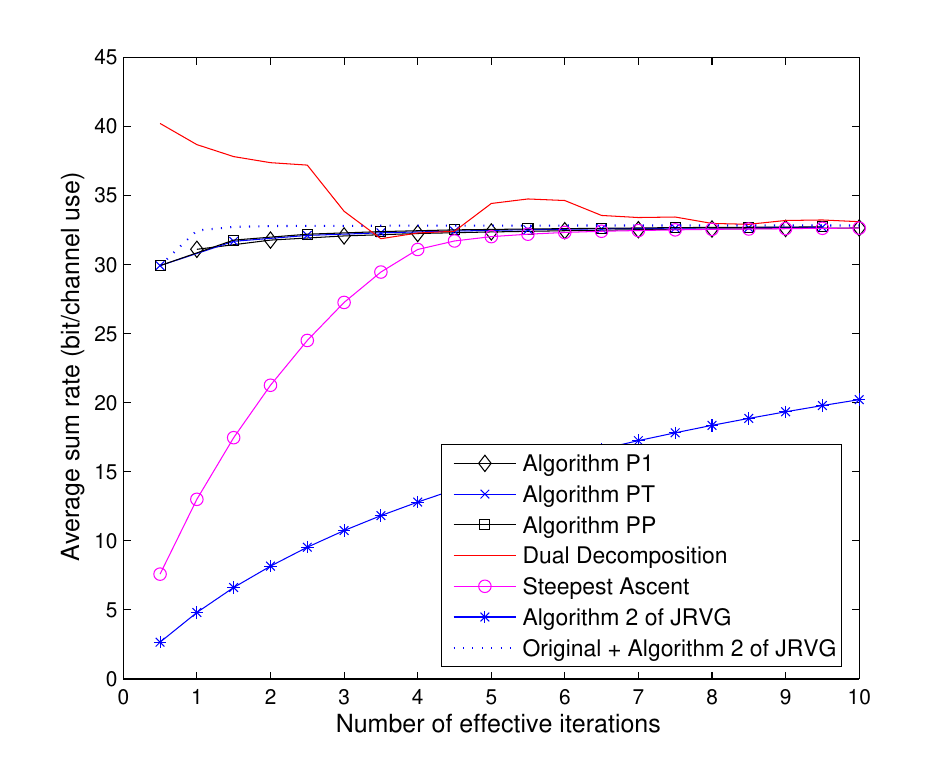}}
\par\end{centering}

\caption{\label{fig:Fig2}Convergence speed comparison of the sum-rate for
a 50-user MAC.}
\end{figure}

\begin{figure}
\begin{centering}
\textsf{\includegraphics[clip,width=95mm]{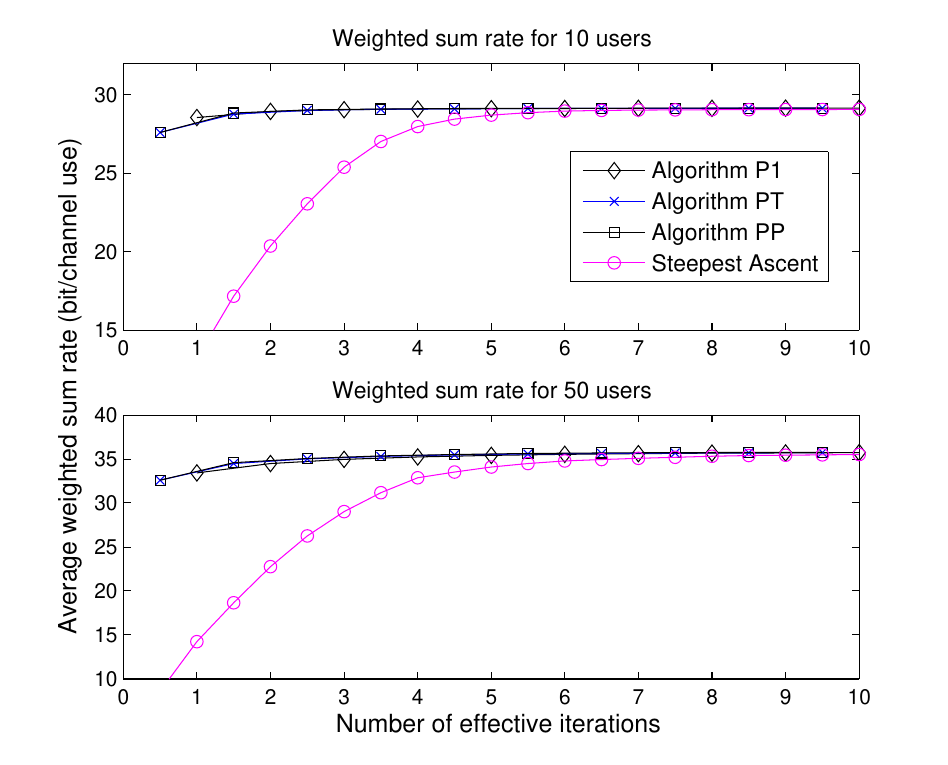}}
\par\end{centering}

\caption{\label{fig:Fig1-2}Convergence speed comparison of the weighted sum-rate
for a 10/50-user MAC.}
\end{figure}

\begin{figure}
\begin{centering}
\textsf{\includegraphics[clip,width=95mm]{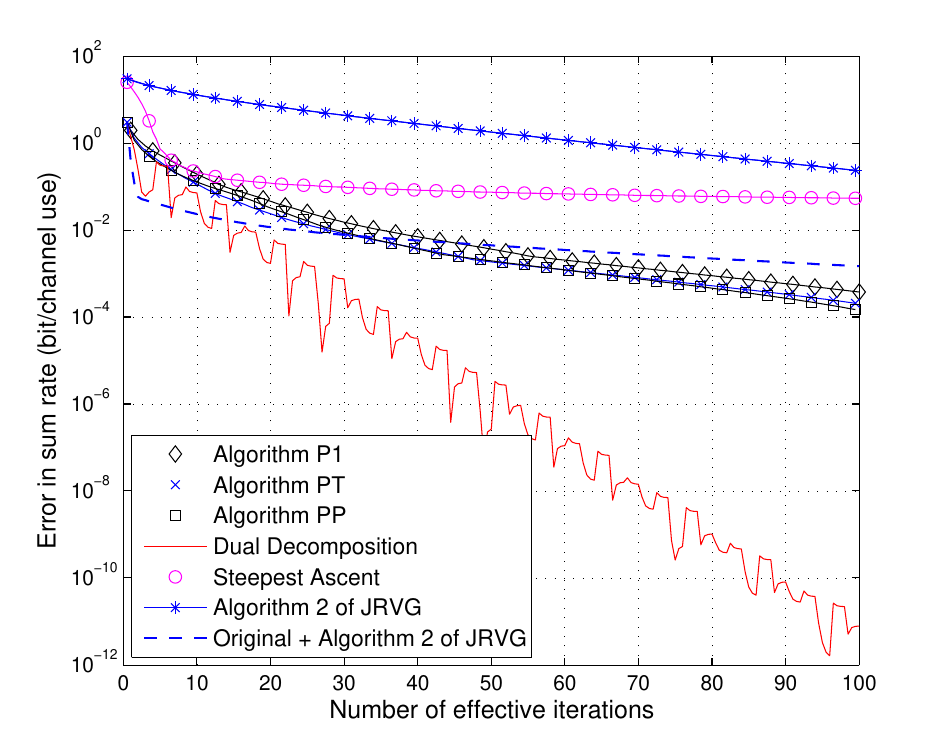}}
\par\end{centering}

\caption{\label{fig:FigSRError}Sum-rate accuracy of the algorithms for a 50-user
MAC.}
\end{figure}

\begin{figure}
\begin{centering}
\textsf{\includegraphics[clip,width=95mm]{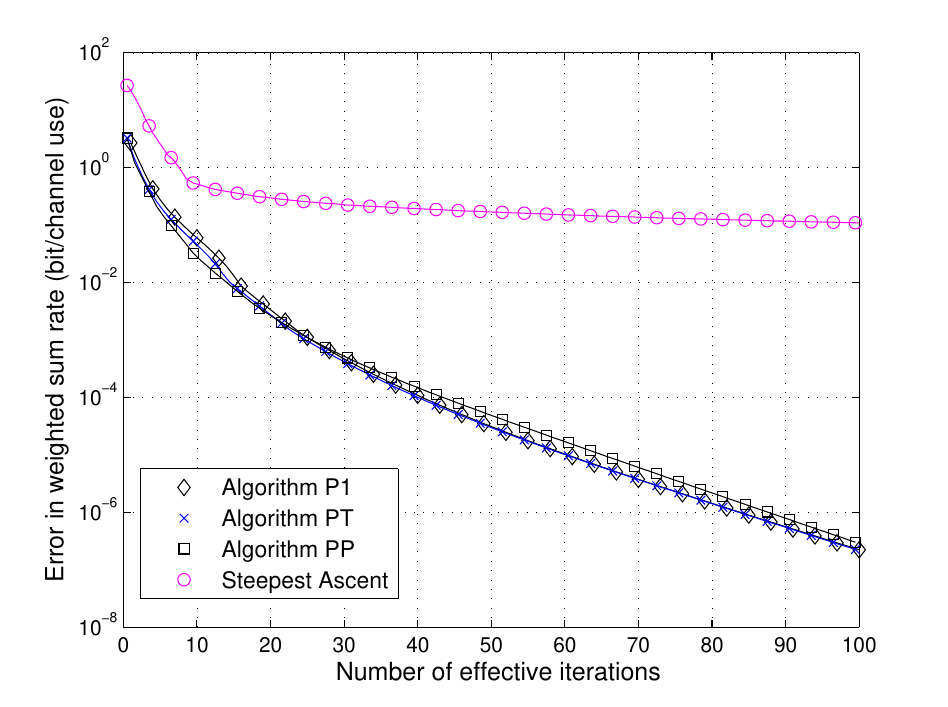}}
\par\end{centering}

\caption{\label{fig:FigWSRError}Weighted sum-rate accuracy of the algorithms
for a 50-user MAC.}
\end{figure}

\begin{figure}
\begin{centering}
\textsf{\includegraphics[clip,width=95mm]{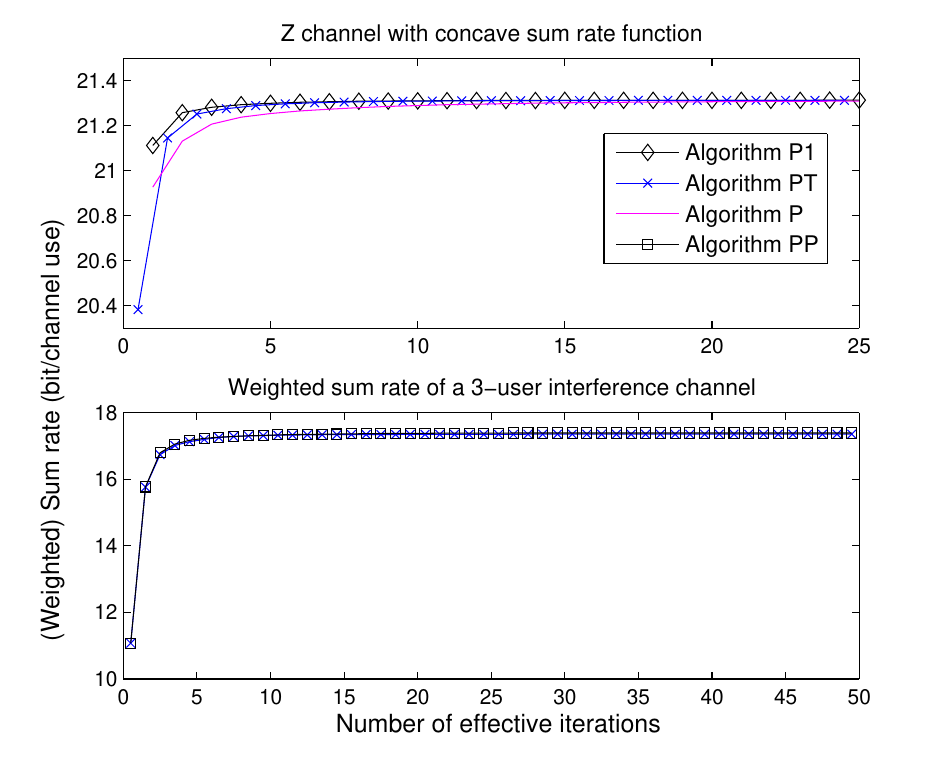}}
\par\end{centering}

\caption{\label{fig:Fig3}Convergence speed of the algorithms for a Z channel
and an interference channel.}
\end{figure}

\begin{figure}
\begin{centering}
\textsf{\includegraphics[clip,width=95mm]{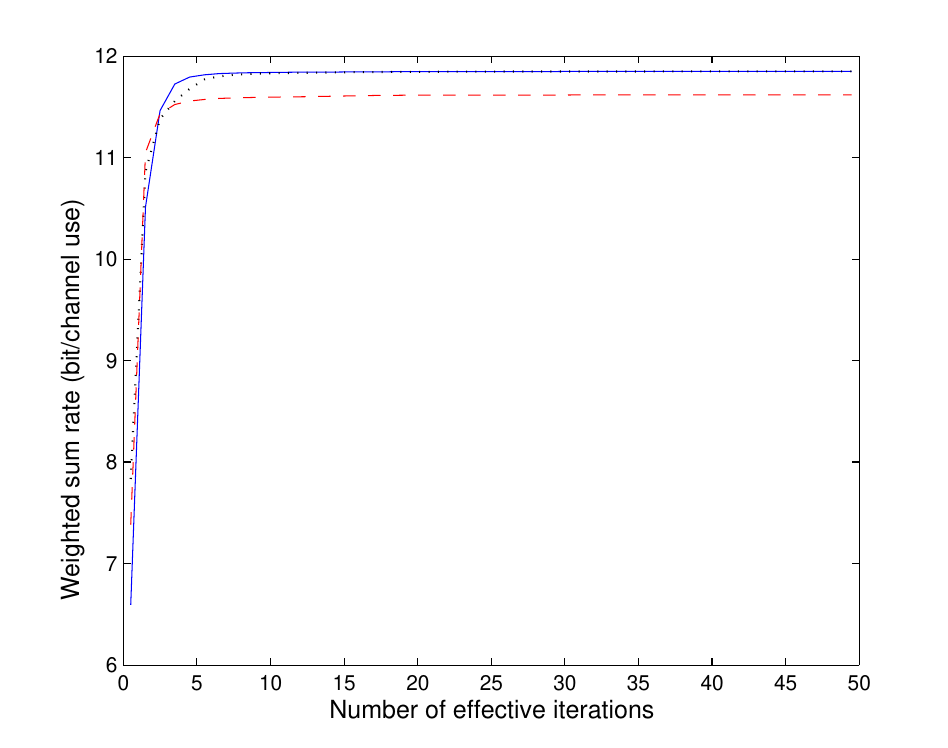}}
\par\end{centering}

\caption{\label{fig:FigXconvg}Convergence speed of Algorithm PP with different
initial points for a 2-user X channel.}
\end{figure}

\begin{figure}
\begin{centering}
\textsf{\includegraphics[clip,width=95mm]{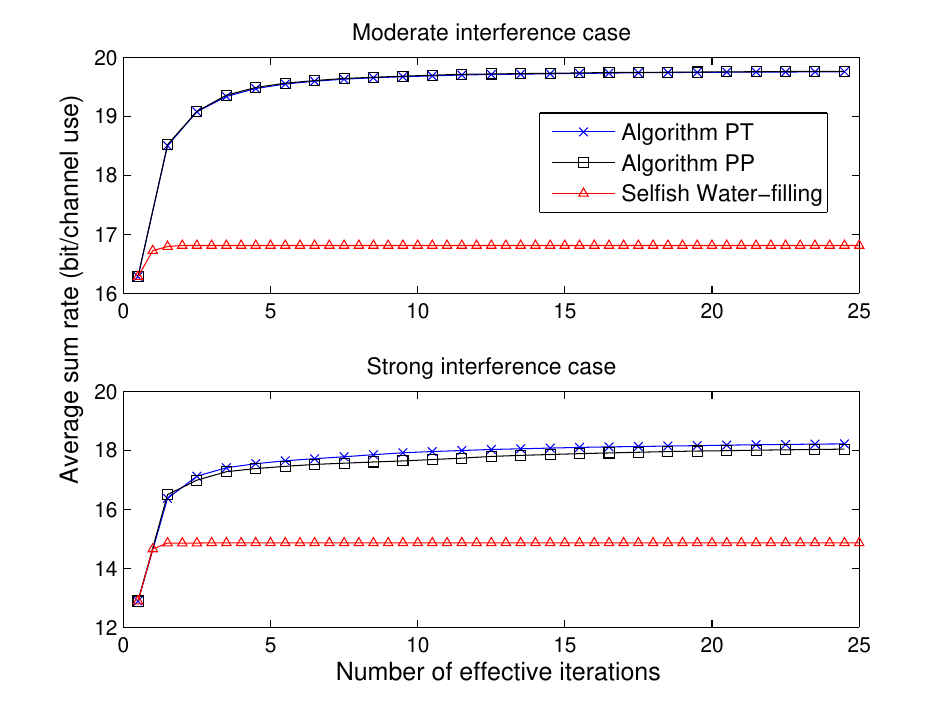}}
\par\end{centering}

\caption{\label{fig:FigiTreeGIWF}Convergence speed comparison with selfish
water-filling for the iTree network \textcolor{black}{in Fig. \ref{fig:iTree}.}}
\end{figure}

\begin{figure}
\begin{centering}
\textsf{\includegraphics[clip,width=95mm]{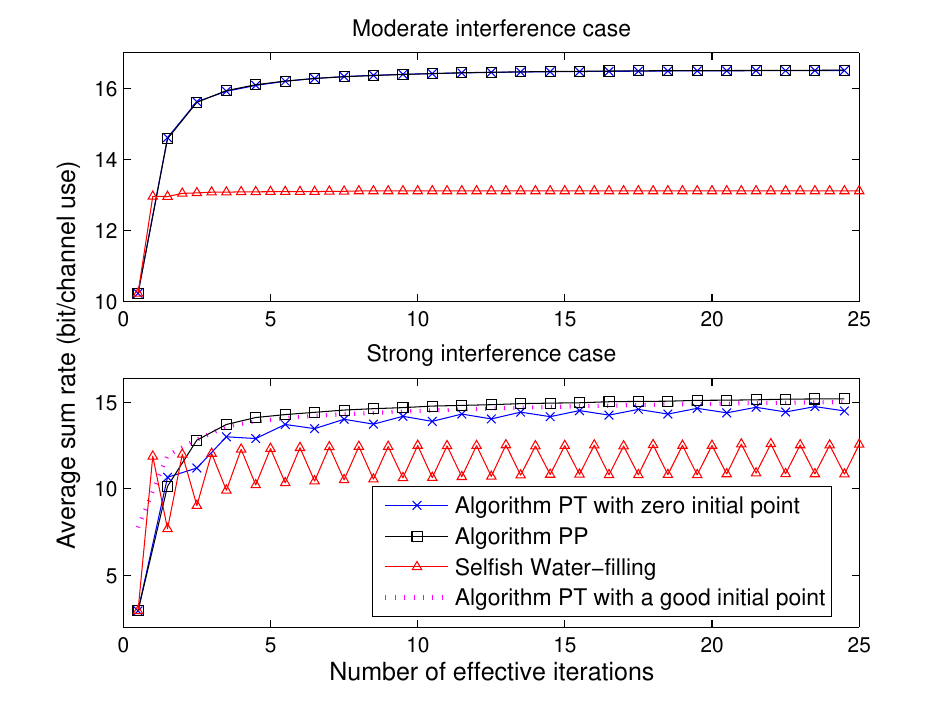}}
\par\end{centering}

\caption{\label{fig:FigIFCGIWF}Convergence behavior comparison with selfish
water-filling for a 3-user interference channel.}
\end{figure}

\begin{figure}
\begin{centering}
\textsf{\includegraphics[clip,width=95mm]{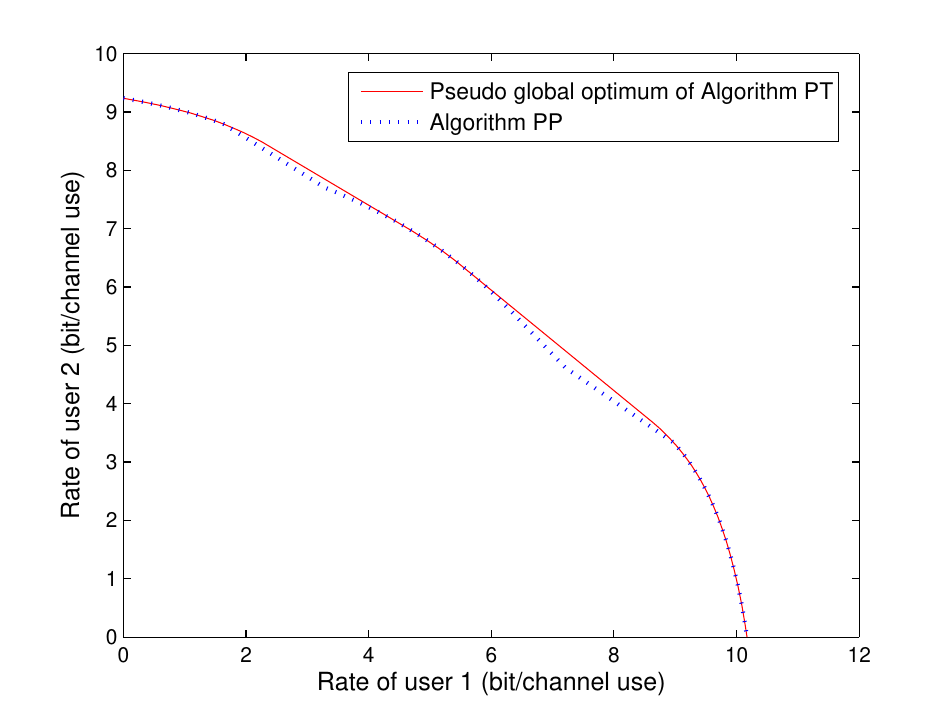}}
\par\end{centering}

\caption{\label{fig:FigIFCBP}Achieved rate regions of a two-user interference
channel.}
\end{figure}

Simulations are used to verify the performance of the proposed algorithms.
Block fading channel is assumed and the channel matrices are independently
generated by $\mathbf{H}_{l,k}=\sqrt{g_{l,k}}\mathbf{H}_{l,k}^{\text{(W)}},\forall k,l$,
where $\mathbf{H}_{l,k}^{\text{(W)}}$ has zero-mean i.i.d. Gaussian
entries with unit variance and the constant $g_{l,k}=0\textrm{dB},\ \forall k,l$
except for Fig. \ref{fig:FigiTreeGIWF} and Fig. \ref{fig:FigIFCGIWF}.
The code word length equal to the fading block length. For weighted
sum-rate, the weights are randomly chosen between 0.8 and 1.2. In
Fig. \ref{fig:Fig1}-\ref{fig:Fig1-2}, Fig. \ref{fig:FigiTreeGIWF}
and Fig. \ref{fig:FigIFCGIWF}, each simulation is averaged over 100
random channel realizations to show that the performance difference
is not an accident. In all other figures, the simulation is performed
over a single channel realization to show the details of the \textcolor{black}{asymptotic
accuracy} or the effect of different initial points. We use zero valued
initial points except for Fig. \ref{fig:FigXconvg}, Fig. \ref{fig:FigIFCGIWF},
and Fig. \ref{fig:FigIFCBP}, where both zero and random initial points
are simulated.

To make fair comparison of the convergence speed by taking into account
the complexity, we use \textit{effective iteration}\emph{s}. Each
iteration of the Algorithm PT or PP consists of an update in the forward
links and an update in the reverse links. The rates are output at
0.5, 1.5, 2.5, \emph{etc.}, effective iterations. For Algorithm P1,
the $i^{\textrm{th}}$ iteration is also the $i^{\textrm{th}}$ effective
iteration. For the other algorithms in the literature, the $i^{\textrm{th}}$
iteration is counted as the $\left(i/2\right)^{\textrm{th}}$ effective
iteration because they only have update in the forward links. This
is a conservative approach because some of these other algorithms'
complexity of each iteration is higher than that of one half iteration
of Algorithm PT or PP. Nevertheless, Algorithms PT and PP still show
clear advantages.

We first demonstrate the superior convergence speed of Algorithms
P1, PT and PP. Fig. \ref{fig:Fig1}-\ref{fig:Fig1-2} plot the (weighted)
sum-rate versus effective iteration number for a MIMO MAC with 2 transmit
antennas at each user and 8 receive antennas. The scenarios of sum-rate,
weighted sum-rate, 10 users, and 50 users are simulated. The results
are compared with those of the steepest ascent algorithm in \cite{Viswanathan_JSAC03_BCGD}
with Matlab code from \cite{Jindal_steepestGradient_matlab}, the
`Original Algorithm' and `Algorithm 2' of JRVG in \cite{Jindal_IT05_IFBC}
with Matlab code from \cite{Jindal_05online_Waterfilling_Aglroithm},
and the dual decomposition algorithm in \cite{Weiyu_IT06_DualIWF}.
The `Original Algorithm' in \cite{Jindal_IT05_IFBC} is obtained by
enforcing the sum-rate water-filling structure in \cite{Yu_IT04_MIMO_MAC_waterfilling_alg}
iteratively. It is observed in \cite{Jindal_IT05_IFBC} that the sum-rate
of the `Original Algorithm' grows fast initially and then may decrease
and/or oscillate. Therefore, a hybrid algorithm called `Original +
Algorithm 2' is considered the best in \cite{Jindal_IT05_IFBC}, where
the `Original Algorithm' is performed for the first five iterations,
and then `Algorithm 2' in \cite{Jindal_IT05_IFBC} is used for the
rest of the iterations. Algorithms P1, PT and PP have similar convergence
speed as that of `Original + Algorithm 2' and outperform all other
algorithms. As expected, the convergence speed of Algorithms P1, PT
and PP appears to be independent of the number of users, possibly
because the girth of the interference graph of the MAC is one, regardless
of the number of users. In Fig. \ref{fig:FigSRError} and Fig. \ref{fig:FigWSRError},
we compare the asymptotic convergence speed, or accuracy, defined
as the errors from the maximums of the (weighted) sum-rate. Algorithm
P1, PT and PP have the highest accuracy except for the dual decomposition
algorithm in \cite{Weiyu_IT06_DualIWF}. This is because the dual
decomposition algorithm is a water-filling directly on sum-rate, while
the polite water-filling is on each link. The errors accumulate when
summing all the links' rates together. In summary, for sum-rate, the
best algorithm will be `Original Algorithm + Algorithm PT or PP' for
a compromise of initial convergence speed and accuracy. For weighted
sum-rate, the best is Algorithm PT or PP.

In Fig. \ref{fig:Fig3}, we plot the sum-rate for a Z channel with
concave sum-rate function and the weighted sum-rate for a 3-user interference
channel. In both channels, each node has 4 antennas. For the Z channel,
Algorithms P1 and PT have faster convergence speed than Algorithm
P. For the 3-user interference channel where the problem may be non-convex,
Algorithm PT and PP have the same performance and converge fast and
monotonically to a stationary point.

In Fig. \ref{fig:FigXconvg}, we show the convergence speed of Algorithm
PP for a nontrivial example of B-MAC networks, the two-user X channel,
with 3 antennas at each node. We plot the weighted sum-rate versus
the effective iteration number with three different initial points
for the same channel realization. For all initial points, Algorithm
PP converges quickly to some stationary points.

As expected, the polite water-filling is superior to the selfish water-filling,
which can be obtained by fixing the $\hat{\mathbf{\Omega}}_{l}$'s
in Algorithm PT or PP as identity matrices. The performance of polite
and selfish water-filling is compared in Fig. \ref{fig:FigiTreeGIWF}
for the iTree network of Fig. \ref{fig:iTree} and in Fig. \ref{fig:FigIFCGIWF}
for a 3-user interference channel. %
Each node is assumed to have four antennas. In the upper sub-plot
of Fig. \ref{fig:FigiTreeGIWF}, we consider the moderate interference
case, where we set $g_{l,k}=0\textrm{dB},\ \forall k,l$. In the lower
sub-plot of Fig. \ref{fig:FigiTreeGIWF}, we consider strong interference
case, where we set $g_{l,3}=10\textrm{dB},\ l=1,2$ for the interfering
links, and $g_{l,k}=0\textrm{dB}$ for other $k,l$'s. It can be observed
that Algorithm PT and PP achieve a much higher sum-rate than the selfish
water-filling. In the upper sub-plot of Fig. \ref{fig:FigIFCGIWF},
we set $g_{l,k}=0\textrm{dB},\forall k,l$, and similar results as
in Fig. \ref{fig:FigiTreeGIWF} can be observed. In the lower sub-plot,
we set $g_{l,k}=10\textrm{dB},\forall k\neq l$, and $g_{l,k}=0\textrm{dB},\forall k=l$.
In this case, the selfish water-filling based algorithm no longer
converges. Algorithm PT converges with a good initial point found
by a typically three trials. Algorithm PP converges for almost all
channel realizations, which demonstrates its insensitivity to the
interference loops and strength.

Algorithms PT and PP can be used to find the approximate boundary
of the convex hull of the achievable rate region. In Fig. \ref{fig:FigIFCBP}%
, we plot the approximate boundary achieved by Algorithm PP for a
two-user interference channel with 3 antennas at each transmitter
and 4 antennas at each receiver. The sum power constraint is $P_{T}=10\textrm{dB}$.
The weights are $w_{1}=\mu$ and $w_{2}=1-\mu$, with $\mu$ varying
from 0.01 to 0.99. The boundary (dot line) obtained from a single
initial point for Algorithm PP is close to the\emph{ }\textit{pseudo
global optimum}\textit{\emph{ (solid line)}}. The result demonstrates
that Algorithm PP can find good solutions even with single initial
point.%
{} Not showing is that Algorithm PT with a single initial point also
achieves near pseudo optimum boundary.

\section{Conclusions and Future Work\label{sec:Conclusion}}

This paper extends BC-MAC rate duality to the MIMO one-hop interference
networks named B-MAC networks with Gaussian input and any valid coupling
matrices. The main contribution is the discovery that the Pareto optimal
input has a polite water-filling structure, which strikes an optimal
balance between reducing interference to others and maximizing a link's
own rate. It provides insight into interference management in networks
and a method to decompose a network to multiple equivalent single-user
links. It can be employed to design/improve most related optimization
algorithms. As an example, low complexity weighted sum-rate maximization
algorithms are designed and demonstrated to have superior performance,
convergence speed, and accuracy. A sub-class of the B-MAC networks,
the interference tree (iTree) networks, is identified, for which the
optimality of the algorithm is discussed. iTree networks appears to
be a natural extension of broadcast and multiaccess networks and possesses
some desirable properties. 

The results in this paper is a stepping stone for solving many interesting
problems. Some future works are listed below.
\begin{itemize}
\item \emph{Extension to Han-Kobayashi Transmission Scheme:} Han-Kobayashi
scheme cancels more interference at multiple receivers and is especially
beneficial when the interfering channel gain is large. But its optimization
is still an open problem. As discussed in Remark \ref{rem:ext-Han-Kobayashi}
of Section \ref{sub:polite water-filling}, the Lagrangian interpretation
of the reverse link interference-plus-noise covariance matrix for
the polite water-filling makes it possible to extend the duality and
polite water-filling to Han-Kobayashi transmission scheme so as to
help understand the optimal input structure and design low complexity
optimization algorithms. The approach in this paper may also be useful
in multi-hop networks, relay channels, and ad-hoc networks.
\item \emph{Multiple Linear Constraints, Cognitive Radio, and Other Optimization
Problems:} Polite water-filling and the extension to a single linear
constraint in Section \ref{sub:extension-linear-constraints} can
be employed to optimize B-MAC networks with multiple linear constraints,
such as per antenna power constraints and interference constraints
in cognitive radio. In these problems, the single linear constraint
in this paper becomes a weighted sum of multiple linear constraints
where the weights are Lagrange multipliers. The weights may be found
by making the constraints satisfied or by the duality of a larger
networks with some virtual nodes for the constraints. Other optimization
problems include ones with quality of service requirement, such as
those studied in our recent work \cite{Liu_10sTSP_Fairness_rate_polit_WF}. 
\item \emph{Distributed Optimization and Finite Rate Message Passing:} This
paper studies centralized optimization with perfect global channel
state information. In practice, distributed/game-theoretic optimization
algorithms with partial channel state information are desirable \cite{Liu_Allerton09_Duality_distributed,Liu_GC10_DistIWF,Liu_10sTSP_Fairness_rate_polit_WF}.
The insight from polite water-filling is well suited for this as it
turns the problem into single-user optimization problems under the
influence of interference from and to others, as summarized by the
covariance matrices $\mathbf{\Omega}_{l}$ and $\hat{\mathbf{\Omega}}_{l}$.
Partial or full knowledge of these covariance matrices can be obtained
from reverse link transmission or pilot training in time division
duplex (TDD) systems, or from message passing among the nodes in frequency
division duplex (FDD) systems \cite{Ashu_IT09_MGPS}. The simulation
results in this paper indicates that a few iterations suffices for
achieving most of the gain, making the message passing approach practical.
In real-world, the message passing is further limited to finite rate.
The single-user view of the polite water-filling structure makes it
convenient to extend the results on the single-user finite rate feedback
in \cite{Liu_IT05_grassmann_bounds,Liu_IT05_power_on_off_feedback}
to B-MAC networks.
\item \emph{Extension to Fading Channels, OFDM Systems, and Adaptive Transmission:}
Fixed channel matrices are considered in this paper. Assuming the
transmitters know the channel state information, it is straightforward
to extend the results of fixed channel matrices in this paper to time
or frequency varying channels like fading channels and OFDM systems
by considering parallel B-MAC networks. Then, the polite water-filling
will be across users, space (antennas), time, and frequency. For weighted
sum-rate maximization, the constant $\mu$ in the polite water-filling
level $w_{l}/\mu$ is chosen such that the average power constraint
across time/frequency is satisfied. When all the channel distributions
across time/frequency are known, such $\mu$ can be obtained by an
off-line algorithm. Otherwise, $\mu$ can be adaptively adjusted to
satisfy the average power constraints. Except for the partial characterization
in Section \ref{sub:WSRMP}, the optimization of encoding/decoding
order in general is an open and important future research. 
\item \emph{iTree Networks:} A fruitful approach to the open problem of
network information theory is to study special cases, such as deterministic
channels \cite{Tse_ITW07_Deterministic_model_for_relay,TSE_IT07_Infwith1bit}
and degree of freedom \cite{TSE_IT07_Infwith1bit,Jafar_IT08_DOF,Jafar_IT08_DOFMIMOX},
in order to gain insight into the problem. iTree networks appears
to be a useful case to study. Because there is no interference loops,
it is possible to say more about the optimality of designed transmission
schemes and capacity regions.. The tree structure also makes it easier
to analyze the distributed optimization by iterative message passing
\cite{Liu_Allerton09_Duality_distributed}, just as the tree structure
facilitates the analysis of LDPC decoding \cite{Richardson_IT01_design_LDPC}.
A factor graph interpretation of the polite water-filling may reveal
more insight. For wireless channels with interference loops, interference
cancellation techniques like the Han-Kobayashi scheme \cite{Han_IT81_hankobascheme}
can be used to break the strong loops and make the graph more like
a tree, making iTree networks relevant to practice.
\end{itemize}
\appendix

\subsection{Proof of Theorem \ref{thm:StongPareto}\label{sub:Proof-StrongPareto}}

Assume $\mathbf{r}$ is a boundary point and $\mathbf{\Sigma}_{1:L}^{'}$
achieves a rate point $\mathbf{r}'\ge\mathbf{r}$, and $\exists k,\ \text{s.t. }r_{k}^{'}>r_{k}$.
One can find an $0<\alpha<1$ such that $\mathcal{I}_{l}\left(\left(\mathbf{\Sigma}_{1}^{'},...,\alpha\mathbf{\Sigma}_{k}^{'},...,\mathbf{\Sigma}_{L}^{'}\right),\mathbf{\Phi}\right)\ge r_{l},\ \forall l$.
Then, the extra power $\textrm{Tr}\left(\mathbf{\mathbf{\Sigma}}_{k}^{'}\right)-\textrm{Tr}\left(\alpha\mathbf{\mathbf{\Sigma}}_{k}^{'}\right)$
can be used to improve \emph{all} non-zero rate users' rates over
$\mathbf{r}$ using $\beta\left(\mathbf{\Sigma}_{1}^{'},...,\alpha\mathbf{\Sigma}_{k}^{'},...,\mathbf{\Sigma}_{L}^{'}\right)$
as the input covariance matrices, where $\beta>1$ is chosen such
that $\sum_{l\ne k}\textrm{Tr}\left(\beta\mathbf{\mathbf{\Sigma}}_{l}^{'}\right)+\textrm{Tr}\left(\beta\alpha\mathbf{\mathbf{\Sigma}}_{k}^{'}\right)=P_{T}$.%
{} By the definition that the region is a union of sets $\left\{ \mathbf{x}:\ 0\le x_{l}\le\mathcal{I}_{l}\left(\mathbf{\Sigma}_{1:L},\mathbf{\Phi}\right),1\leq l\leq L\right\} $,
$\mathbf{r}$ must not be a boundary point, which contradicts with
the assumption. Therefore, the statement in Theorem \ref{thm:StongPareto}
must be true.

\subsection{Proof of Theorem \ref{thm:optdual}\label{sub:ProofWFST}}

Without loss of generality, we prove Theorem \ref{thm:optdual} for
link 1. Assume $M_{l}=M,l=1,...,L$ for simplicity of notations. The
proof is obtained by considering another sum power minimization problem
in (\ref{eq:Minpowl}). We first give Lemma \ref{lem:NesayKKT} which
says that there exist dual variables satisfying the KKT conditions
of problem (\ref{eq:Minpowl}). Then we show that these dual variables
are the optimal dual variables of problem (\ref{eq:minpowl}) and
are given by Theorem \ref{thm:optdual}.

First we define some notations. Define four sub matrices of the cross-talk
matrix $\tilde{\mathbf{\Psi}}=\mathbf{\Psi}\left(\tilde{\mathbf{T}},\tilde{\mathbf{R}}\right)$
in (\ref{eq:faiG}) as 
\[
\mathbf{\tilde{\mathbf{\Psi}}}=\left[\begin{array}{cc}
\mathbf{\tilde{\mathbf{\Psi}}}_{1}\in\mathbb{R}_{+}^{M\times M} & \mathbf{\tilde{\mathbf{\Psi}}}_{1,-1}\in\mathbb{R}_{+}^{M\times(L-1)M}\\
\mathbf{\tilde{\mathbf{\Psi}}}_{-1,1}\in\mathbb{R}_{+}^{(L-1)M\times M} & \mathbf{\tilde{\mathbf{\Psi}}}_{-1}\in\mathbb{R}_{+}^{(L-1)M\times(L-1)M}
\end{array}\right]
\]
where $\tilde{\mathbf{T}}=\left[\tilde{\mathbf{t}}_{l,m}\right]_{m=1,...,M,l=1,...,L}$
and $\tilde{\mathbf{R}}=\left[\tilde{\mathbf{r}}_{l,m}\right]_{m=1,...,M,l=1,...,L}$
are the Pareto optimal transmit and receive vectors. Let $\tilde{\mathbf{D}}_{1}\in\mathbb{R}_{+}^{M\times M}$
and $\tilde{\mathbf{D}}_{-1}\in\mathbb{R}_{+}^{(L-1)M\times(L-1)M}$
respectively be the sub matrices at the upper left and lower right
corner of $\mathbf{D}\left(\tilde{\mathbf{T}},\tilde{\mathbf{R}},\mathbf{\gamma}^{0}\right)$
in (\ref{eq:DG}), where $\mathbf{\gamma}^{0}$ is set as the SINRs
achieved by $\left\{ \tilde{\mathbf{T}},\tilde{\mathbf{R}},\mathbf{\tilde{p}}\right\} $.
Define 
\[
\mathbf{B}_{-1}=\tilde{\mathbf{D}}_{-1}^{-1}-\mathbf{\tilde{\mathbf{\Psi}}}_{-1}.
\]
For $k=2,...,L$, let $\mathbf{b}_{-1}^{k,m}\in\mathbb{C}^{1\times(L-1)M}$
be the $\left((k-2)M+m\right)^{th}$ row of $\mathbf{B}_{-1}$. 

Then we formulate the following sum power minimization problem for
the proof, where the minimization is over $\mathbf{\Sigma}_{1}$ and
the transmit powers of other links $\mathbf{p}_{-1}=\left[\mathbf{p}_{2}^{T},...,\mathbf{p}_{L}^{T}\right]^{T}\in\mathbb{R}_{+}^{\left(L-1\right)M\times1}$
with the transmit and receive vectors for other links fixed as $\{\tilde{\mathbf{t}}_{k,m},\mathbf{\tilde{r}}_{k,m},k\neq1\}$
\begin{eqnarray}
 & \underset{\mathbf{\Sigma}_{1}\succeq0,\mathbf{p}_{-1}\geq0}{\textrm{min}}\textrm{Tr}\left(\mathbf{\Sigma}_{1}\right)+\mathbf{1}^{T}\mathbf{p}_{-1}\label{eq:Minpowl}\\
\textrm{s.t.} & \textrm{log}\left|\mathbf{I}+\mathbf{H}_{1,1}\mathbf{\Sigma}_{1}\mathbf{H}_{1,1}^{\dagger}\mathbf{\Omega}_{1}^{-1}\left(\mathbf{p}_{-1}\right)\right|\geq\mathcal{\tilde{I}}_{1}\nonumber \\
 & \textrm{Tr}\left(\mathbf{\mathbf{\Sigma}}_{1}\mathbf{A}_{k,m}^{(1)}\right)-\mathbf{b}_{-1}^{k,m}\mathbf{p}_{-1}\leq-1,\begin{array}{c}
k=2,...,L\\
m=1,...,M
\end{array},\label{eq:Infcon}
\end{eqnarray}
where $\mathbf{\Omega}_{1}\left(\mathbf{p}_{-1}\right)=\mathbf{I}+\sum_{k=2}^{L}\mathbf{\Phi}_{1,k}\mathbf{H}_{1,k}\sum_{m=1}^{M}p_{k,m}\tilde{\mathbf{t}}_{k,m}\tilde{\mathbf{t}}_{k,m}^{\dagger}\mathbf{H}_{1,k}^{\dagger}$
is the interference-plus-noise covariance matrix of link 1. Note that
the constraints in (\ref{eq:Infcon}) imply the rate constraints:
$\sum_{m=1}^{M}\textrm{log}$ $\left(1+\gamma_{k,m}\right)\geq\mathcal{\tilde{I}}_{k}=\sum_{m=1}^{M}\textrm{log}\left(1+\tilde{\gamma}_{k,m}\right),k=2,...,L$,
where
\begin{equation}
\mathbf{\gamma}_{k,m}=\frac{p_{k,m}\left|\mathbf{\tilde{r}}_{k,m}^{\dagger}\mathbf{H}_{k,k}\tilde{\mathbf{t}}_{k,m}\right|^{2}}{1+\textrm{Tr}\left(\mathbf{\Sigma}_{1}\mathbf{A}_{k,m}^{(1)}\right)+{\displaystyle \sum_{l=2}^{L}}{\displaystyle \sum_{n=1}^{M}}p_{l,n}\tilde{\mathbf{\Psi}}_{k,m}^{l,n}},\label{eq:newSINR}
\end{equation}
is the SINR for the $m^{th}$ stream of link $k$ achieved by $\mathbf{\Sigma}_{1}$,
$\mathbf{p}_{-1}$, and $\{\tilde{\mathbf{t}}_{k,m},\mathbf{\tilde{r}}_{k,m},k\neq1\}$.
This is because by the definition of $\mathbf{b}_{-1}^{k,m}$, the
constraint $\textrm{Tr}\left(\mathbf{\Sigma}_{1}\mathbf{A}_{k,m}^{(1)}\right)-\mathbf{b}_{-1}^{k,m}\mathbf{p}_{-1}\leq-1$
is equivalent to $\mathbf{\gamma}_{k,m}\geq\tilde{\mathbf{\gamma}}_{k,m}$.
It can be proved by contradiction that the Pareto optimal input $\mathbf{\tilde{\mathbf{\Sigma}}}_{1}$
and $\mathbf{\tilde{p}}_{-1}$ is an optimal solution of (\ref{eq:Minpowl}).
The Lagrangian of problem (\ref{eq:Minpowl}) is
\begin{align}
L(\mathbf{\lambda},\nu_{1},\mathbf{\Theta},\mathbf{\Sigma}_{1},\mathbf{p}_{-1}) & =\textrm{Tr}\left(\mathbf{\Sigma}_{1}\left(\mathbf{A}_{1}\left(\mathbf{\lambda}\right)-\mathbf{\Theta}\right)\right)+\nu_{1}\mathcal{\tilde{I}}_{1}\nonumber \\
 & -\nu_{1}\textrm{log}\left|\mathbf{I}+\mathbf{H}_{1,1}\mathbf{\Sigma}_{1}\mathbf{H}_{1,1}^{\dagger}\mathbf{\Omega}_{1}^{-1}\left(\mathbf{p}_{-1}\right)\right|\nonumber \\
 & +\mathbf{1}^{T}\mathbf{p}_{-1}-\mathbf{\lambda}^{T}\mathbf{B}_{-1}\mathbf{p}_{-1}+\mathbf{\lambda}^{T}\mathbf{1},\label{eq:Lag-sigma1-p1}
\end{align}
where the dual variables $\nu_{1}\in\mathbb{R}_{+}$ and $\mathbf{\lambda}=\left[\lambda_{2,1},...,\lambda_{2,M},...,\lambda_{L,1},...,\lambda_{L,M}\right]^{T}\in\mathbb{R}_{+}^{\left(L-1\right)M\times1}$
are associated with the rate constraint and the constraints in (\ref{eq:Infcon})
respectively; $\mathbf{A}_{1}\left(\mathbf{\lambda}\right)$ is defined
in (\ref{eq:Lag_minp}); $\mathbf{\Theta}$ is the matrix dual variables
associated with the positive semidefiniteness constraint on $\mathbf{\Sigma}_{1}$.
Using the enhanced Fritz John necessary conditions in \cite[Sec. 5.2]{Dimitri_Cvxbook03},
it can be proved that there exist dual variables satisfying the KKT
conditions of problem (\ref{eq:Minpowl}) as stated in the following
lemma.
\begin{lem}
\label{lem:NesayKKT}The KKT conditions are necessary for $\mathbf{\tilde{\mathbf{\Sigma}}}_{1},\mathbf{\tilde{p}}_{-1}$
to be optimal for problem (\ref{eq:Minpowl}), i.e., there exist dual
variables $\mathbf{\tilde{\mathbf{\lambda}}}=\left[\mathbf{\tilde{\lambda}}_{k,m}\right]_{k\ne1},\:\tilde{\nu}_{1}\geq0$
and $\tilde{\mathbf{\Theta}}\succeq0,\ \textrm{Tr}\left(\mathbf{\tilde{\mathbf{\Sigma}}}_{1}\tilde{\mathbf{\Theta}}\right)=0$
such that
\begin{align}
 & \tilde{\nabla}_{\mathbf{\Sigma}_{1}}L\triangleq\left.\nabla_{\mathbf{\Sigma}_{1}}L(\mathbf{\tilde{\mathbf{\lambda}}},\tilde{\nu}_{1},\tilde{\mathbf{\Theta}},\mathbf{\Sigma}_{1},\mathbf{\tilde{p}}_{-1})\right|_{\mathbf{\Sigma}_{1}=\tilde{\mathbf{\Sigma}}_{1}}=\mathbf{0},\label{eq:PSD1}\\
 & \tilde{\nabla}_{\mathbf{p}_{-1}}L\triangleq\left.\nabla_{\mathbf{p}_{-1}}L(\mathbf{\tilde{\mathbf{\lambda}}},\tilde{\nu}_{1},\tilde{\mathbf{\Theta}},\mathbf{\tilde{\mathbf{\Sigma}}}_{1},\mathbf{p}_{-1})\right|_{\mathbf{p}_{-1}=\mathbf{\tilde{p}}_{-1}}=\mathbf{0}.\label{eq:PSD2}
\end{align}

\end{lem}

Note that $\mathbf{\Omega}_{1}\left(\mathbf{\tilde{p}}_{-1}\right)=\tilde{\mathbf{\Omega}}_{1}$.
Then (\ref{eq:PSD1}) can be expressed as
\begin{align}
\mathbf{A}_{1}\left(\tilde{\mathbf{\lambda}}\right)-\tilde{\nu}_{1}\mathbf{H}_{1,1}^{\dagger}\left(\tilde{\mathbf{\Omega}}_{1}+\mathbf{H}_{1,1}\tilde{\mathbf{\Sigma}}_{1}\mathbf{H}_{1,1}^{\dagger}\right)^{-1}\mathbf{H}_{1,1}-\tilde{\mathbf{\Theta}}=\mathbf{0}.\label{eq:PSDv1}
\end{align}
Because (\ref{eq:PSDv1}) is also the equation that the derivative
of Lagrangian of the problem (\ref{eq:minpowl}) equals to zero, and
$\mathbf{\tilde{\mathbf{\Sigma}}}_{1}$ is the optimal solution of
problem (\ref{eq:minpowl}), $\mathbf{\tilde{\mathbf{\lambda}}},\tilde{\nu}_{1},\tilde{\mathbf{\Theta}}$
satisfy the KKT conditions of the convex problem (\ref{eq:minpowl})
and thus are the optimal dual variables of problem (\ref{eq:minpowl}). 

We show below that $\mathbf{\tilde{\mathbf{\lambda}}}$ and $\tilde{\nu}_{1}$
in Lemma \ref{lem:NesayKKT} is uniquely given by (\ref{eq:optidual})
and (\ref{eq:optidual2}). Combining this fact and Lemma \ref{lem:NesayKKT},
we conclude that $\mathbf{\tilde{\mathbf{\lambda}}}$ and $\tilde{\nu}_{1}$
in (\ref{eq:optidual}) and (\ref{eq:optidual2}) are the optimal
dual variables of problem (\ref{eq:minpowl}). It is difficult to
directly solve $\mathbf{\tilde{\mathbf{\lambda}}}$ and $\tilde{\nu}_{1}$
from (\ref{eq:PSD1}) and (\ref{eq:PSD2}). To simplify the problem,
we restrict $\mathbf{\Sigma}_{1}$ to be $\mathbf{\Sigma}_{1}=\tilde{\mathbf{T}}_{1}\textrm{diag}\left(\mathbf{p}_{1}\right)\mathbf{\tilde{T}}_{1}^{\dagger}$
in the Lagrangian. For convenience, define 
\begin{equation}
r_{1}^{s}\left(\mathbf{p}\right)=\sum_{m=1}^{M}\textrm{log}\left(1+\gamma_{1,m}\left(\tilde{\mathbf{T}},\tilde{\mathbf{R}},\mathbf{p}\right)\right),\label{eq:Defr1p}
\end{equation}
where $\gamma_{1,m}\left(\tilde{\mathbf{T}},\tilde{\mathbf{R}},\mathbf{p}\right)$
is the SINR of the $m^{th}$ stream of link $1$ achieved by $\left\{ \tilde{\mathbf{T}},\tilde{\mathbf{R}},\mathbf{p}\right\} $.
Note that $\textrm{log}\left|\mathbf{I}+\mathbf{H}_{1,1}\mathbf{\Sigma}_{1}\mathbf{H}_{1,1}^{\dagger}\mathbf{\Omega}_{1}^{-1}\left(\mathbf{p}\right)\right|=r_{1}^{s}\left(\mathbf{p}\right)$,
and \textbf{$\tilde{\mathbf{t}}_{1,m}^{\dagger}\tilde{\mathbf{\Theta}}\tilde{\mathbf{t}}_{1,m}=0$}
because $\textrm{Tr}\left(\mathbf{\tilde{\mathbf{\Sigma}}}_{1}\tilde{\mathbf{\Theta}}\right)=0$
and $\tilde{\mathbf{\Theta}}$ is positive semidefinite. Then the
Lagrangian (\ref{eq:Lag-sigma1-p1})%
{} can be rewritten as follow.
\begin{eqnarray*}
 & \bar{L}(\tilde{\mathbf{\lambda}},\tilde{\nu}_{1},\mathbf{p})=\sum_{m=1}^{M}p_{1,m}\tilde{\mathbf{t}}_{1,m}^{\dagger}\mathbf{A}_{1}\left(\tilde{\mathbf{\lambda}}\right)\tilde{\mathbf{t}}_{1,m}+\mathbf{1}^{T}\mathbf{p}_{-1}\\
 & +\tilde{\nu}_{1}\left(\mathcal{\tilde{I}}_{1}-r_{1}^{s}\left(\mathbf{p}\right)\right)-\tilde{\mathbf{\lambda}}^{T}\mathbf{B}_{-1}\mathbf{p}_{-1}+\tilde{\mathbf{\lambda}}^{T}\mathbf{1}.
\end{eqnarray*}
Note that (\ref{eq:PSD1}) and (\ref{eq:PSD2}) imply
\begin{eqnarray}
\nabla_{\mathbf{p}_{1}}\bar{L}\left.\left(\tilde{\mathbf{\lambda}},\tilde{\nu}_{1},[\mathbf{p}_{1}^{T},\tilde{\mathbf{p}}_{-1}^{T}]^{T}\right)\right|_{\mathbf{p}_{1}=\tilde{\mathbf{p}}_{1}} & = & \mathbf{0},\label{eq:tidu0p}\\
\nabla_{\mathbf{p}_{-1}}\bar{L}\left.\left(\tilde{\mathbf{\lambda}},\tilde{\nu}_{1},[\tilde{\mathbf{p}}_{1}^{T},\mathbf{p}_{-1}^{T}]^{T}\right)\right|_{\mathbf{p}_{-1}=\tilde{\mathbf{p}}_{-1}} & = & \mathbf{0}.\label{eq:tidu0p1}
\end{eqnarray}
As will be shown below, the solution of the above equations is uniquely
given by (\ref{eq:optidual}) and (\ref{eq:optidual2}), which therefore
must also be the solution of the equations in (\ref{eq:PSD1}) and
(\ref{eq:PSD2}).

For convenience, define $\hat{\mathbf{q}}_{1}=\left[\hat{q}_{1,1},...,\hat{q}_{1,M}\right]$
and 
\[
\hat{q}_{1,m}=\frac{\tilde{\nu}_{1}\mathbf{\tilde{\gamma}}_{1,m}^{2}}{\tilde{p}_{1,m}\left(1+\mathbf{\tilde{\gamma}}_{1,m}\right)\tilde{G}_{1,m}},
\]
where $\tilde{G}_{1,m}=\left|\tilde{\mathbf{r}}_{1,m}^{\dagger}\mathbf{H}_{1,1}\tilde{\mathbf{t}}_{1,m}\right|^{2}$.
It can be derived from (\ref{eq:Defr1p}) and $\tilde{\mathbf{\gamma}}_{1,m}=\gamma_{1,m}\left(\tilde{\mathbf{T}},\tilde{\mathbf{R}},\mathbf{\tilde{p}}\right)$
that
\begin{align*}
\tilde{\nu}_{1}\frac{\partial r_{1}^{s}\left(\mathbf{p}\right)}{\partial p_{1,m}}|_{\mathbf{p}=\mathbf{\tilde{p}}} & =\frac{\tilde{G}_{1,m}\hat{q}_{1,m}}{\tilde{\mathbf{\gamma}}_{1,m}}-\sum_{n=1}^{M}\mathbf{\tilde{\mathbf{\Psi}}}_{1,n}^{1,m}\hat{q}_{1,n},m=1,...,M,\\
\tilde{\nu}_{1}\frac{\partial r_{1}^{s}\left(\mathbf{p}\right)}{\partial p_{k,m}}|_{\mathbf{p}=\mathbf{\tilde{p}}} & =-\sum_{n=1}^{M}\mathbf{\tilde{\mathbf{\Psi}}}_{1,n}^{k,m}\hat{q}_{1,n},m=1,...,M,k\neq1.
\end{align*}
Noting that $\tilde{\mathbf{t}}_{1,m}^{\dagger}\mathbf{A}_{1}\left(\tilde{\mathbf{\lambda}}\right)\tilde{\mathbf{t}}_{1,m}={\displaystyle \sum_{k=2}^{L}}{\displaystyle \sum_{n=1}^{M}}\tilde{\lambda}_{k,n}\mathbf{\tilde{\mathbf{\Psi}}}_{k,n}^{1,m}+1$,
equation (\ref{eq:tidu0p}) and (\ref{eq:tidu0p1}) can be expressed
as
\begin{align}
\mathbf{\tilde{D}}_{1}^{-1}\hat{\mathbf{q}}_{1}-\mathbf{\tilde{\mathbf{\Psi}}}_{-1,1}^{T}\tilde{\mathbf{\lambda}}-\mathbf{\tilde{\mathbf{\Psi}}}_{1}^{T}\hat{\mathbf{q}}_{1} & =\mathbf{1},\label{eq:TDJM1}\\
\mathbf{B}_{-1}^{T}\tilde{\mathbf{\lambda}}-\mathbf{\tilde{\mathbf{\Psi}}}_{1,-1}^{T}\hat{\mathbf{q}}_{1} & =\mathbf{1}.\label{eq:TDJM2}
\end{align}
Define $\mathbf{\hat{q}}=\left[\hat{\mathbf{q}}_{1}^{T},\tilde{\mathbf{\lambda}}{}^{T}\right]^{T}$.
Then (\ref{eq:TDJM1}) and (\ref{eq:TDJM2}) together form the following
linear equations
\[
\left(\mathbf{D}^{-1}\left(\tilde{\mathbf{T}},\tilde{\mathbf{R}},\mathbf{\gamma}^{0}\right)-\mathbf{\Psi}^{T}\left(\tilde{\mathbf{T}},\tilde{\mathbf{R}}\right)\right)\mathbf{\hat{q}}=\mathbf{1}.
\]
Since $\mathbf{D}^{-1}\left(\tilde{\mathbf{T}},\tilde{\mathbf{R}},\mathbf{\gamma}^{0}\right)-\mathbf{\Psi}^{T}\left(\tilde{\mathbf{T}},\tilde{\mathbf{R}}\right)$
is invertible \cite{Rao_TOC07_netduality}, $\mathbf{\hat{q}}$ is
uniquely given by 
\[
\mathbf{\hat{q}}=\left(\mathbf{D}^{-1}\left(\tilde{\mathbf{T}},\tilde{\mathbf{R}},\mathbf{\gamma}^{0}\right)-\mathbf{\Psi}^{T}\left(\tilde{\mathbf{T}},\tilde{\mathbf{R}}\right)\right)^{-1}\mathbf{1}=\mathbf{\tilde{q}},
\]
which means $\mathbf{\tilde{\mathbf{\lambda}}},\tilde{\nu}_{1}$ is
uniquely given by (\ref{eq:optidual}) and (\ref{eq:optidual2}).

\subsection{\label{sub:Proof-of-TrQ}Proof of Lemma \ref{lem:TrQ}}

It is a consequence of the SINR duality applied to a single-user channel
$\bar{\mathbf{H}}_{l}\triangleq\mathbf{\Omega}_{l}^{-1/2}\mathbf{H}_{l,l}\hat{\mathbf{\Omega}}_{l}^{-1/2}$.
Decompose $\mathbf{Q}_{l}$ and $\mathbf{\hat{Q}}_{l}$ to beams as
$\mathbf{Q}_{l}=\sum_{m=1}^{M_{l}}d_{l,m}\mathbf{u}_{l,m}\mathbf{u}_{l,m}^{\dagger}$,
where $d_{l,m}=p_{l,m}\left\Vert \hat{\mathbf{\Omega}}_{l}^{1/2}\mathbf{t}_{l,m}\right\Vert ^{2}$
is the equivalent transmit power and $\mathbf{u}_{l,m}=\hat{\mathbf{\Omega}}_{l}^{1/2}\sqrt{p_{l,m}}\mathbf{t}_{l,m}/\sqrt{d_{l,m}}$
is the equivalent transmit vector; and $\mathbf{\hat{Q}}_{l}=\sum_{m=1}^{M_{l}}\hat{d}_{l,m}\mathbf{v}_{l,m}\mathbf{v}_{l,m}^{\dagger}$,
where $\hat{d}_{l,m}=q_{l,m}\left\Vert \mathbf{\Omega}_{l}^{1/2}\mathbf{r}_{l,m}\right\Vert ^{2}$
and $\mathbf{v}_{l,m}=\mathbf{\Omega}_{l}^{1/2}\sqrt{q_{l,m}}\mathbf{r}_{l,m}/\sqrt{\hat{d}_{l,m}}$.
According to the covariance transformation, $\left\{ \mathbf{t}_{l,m}\right\} $,
$\left\{ \mathbf{r}_{l,m}\right\} $, $\left\{ p_{l,m}\right\} $
and $\left\{ q_{l,m}\right\} $ achieves the same set of SINRs in
the forward and reverse links. Correspondingly, $\left\{ \mathbf{u}_{l,m}\right\} $,
$\left\{ \mathbf{v}_{l,m}\right\} $, $\left\{ d_{l,m}\right\} $
and $\left\{ \hat{d}_{l,m}\right\} $ achieves the same set of SINRs
in $\bar{\mathbf{H}}_{l}$ and $\bar{\mathbf{H}}_{l}^{\dagger}$.
Apply Theorem 1 of \cite{Rao_TOC07_netduality} to this single-user
network, we obtain $\sum_{m=1}^{M_{l}}d_{l,m}=\sum_{m=1}^{M_{l}}\hat{d}_{l,m}$,
i.e., $\textrm{Tr}\left(\mathbf{Q}_{l}\right)=\textrm{Tr}\left(\mathbf{\hat{Q}}_{l}\right)$.

\subsection{Proof of Theorem \ref{thm:At-the-boundary}\label{sub:Proof-of-conj}}

The proof contains three parts.%
{} The uniqueness of the covariance transformation is shown by the observation
that it can be equivalently viewed as the covariance transformations
for parallel single-user channels. The other two parts, the matrix
equations (\ref{eq:MtformTF}) and (\ref{eq:SigmhDirect}), are proved
by utilizing the polite water-filling structure.

First we prove that different decompositions of $\mathbf{\Sigma}_{1:L}$
produce the same $\hat{\mathbf{\Sigma}}_{1:L}$ through the covariance
transformation. We will use the notations in the proof of Lemma \ref{lem:TrQ}.
Suppose $\mathbf{\Sigma}_{1:L}$ is transformed to $\hat{\mathbf{\Sigma}}_{1:L}$
using the decomposition $\mathbf{\Sigma}_{l}=\sum_{m=1}^{M_{l}}p_{l,m}\mathbf{t}_{l,m}\mathbf{t}_{l,m}^{\dagger},l=1,...,L$
with $\left\{ \mathbf{r}_{l,m}\right\} $ as the MMSE-SIC receive
vectors and $\left\{ q_{l,m}\right\} $ as the reverse link transmit
powers. The corresponding equivalent input covariance matrices are
$\mathbf{Q}_{l}=\sum_{m=1}^{M_{l}}d_{l,m}\mathbf{u}_{l,m}\mathbf{u}_{l,m}^{\dagger}$
and $\mathbf{\hat{Q}}_{l}=\sum_{m=1}^{M_{l}}\hat{d}_{l,m}\mathbf{v}_{l,m}\mathbf{v}_{l,m}^{\dagger}$.
For a new decomposition $\mathbf{\Sigma}_{l}=\sum_{m=1}^{M_{l}}p_{l,m}^{'}\mathbf{t}_{l,m}^{'}\left(\mathbf{t}_{l,m}^{'}\right)^{\dagger},\ l=1,...,L$,
it can be changed from the old decomposition one link at a time. Without
loss of generality, let $l$ be the first link to be changed and prove
$\hat{\mathbf{\Sigma}}_{1:L}$ remains the same. In this case, the
transmission schemes of all other links are still $\left\{ \mathbf{t}_{k,m},p_{k,m},\forall k\neq l\right\} $,
and the MMSE-SIC receive vectors are still $\left\{ \mathbf{r}_{k,m},\forall k\neq l\right\} $
because the interference at link $k(\ne l)$ does not depend on the
decomposition of $\mathbf{\Sigma}_{l}$. Furthermore, we artificially
fix the transmit powers of all other reverse links as $\left\{ q_{k,m},\forall k\neq l\right\} $.
Then both $\mathbf{\Omega}_{l}$ and $\mathbf{\hat{\Omega}}_{l}$
are not changed. Find $\hat{\mathbf{\Sigma}}_{l}^{'}$ by calculating
the MMSE-SIC receive vectors $\left\{ \mathbf{r}_{l,m}^{'}\right\} $
and the unique $q_{l,m}^{'}$ such that the SINR of the $m^{\text{th}}$
stream of the $l^{\text{th}}$ forward and reverse link are the same,
$\forall m$. Since $\mathbf{Q}_{l}$ is a water-filling solution
over the forward equivalent channel $\bar{\mathbf{H}}_{l}$ and thus
achieves the capacity of $\bar{\mathbf{H}}_{l}$ \cite{Telatar_EuroTrans_1999_MIMOCapacity},
the corresponding $\mathbf{\hat{Q}}_{l}^{'}=\mathbf{\Omega}_{l}^{1/2}\hat{\mathbf{\Sigma}}_{l}^{'}\mathbf{\Omega}_{l}^{1/2}$
achieves the capacity of $\bar{\mathbf{H}}_{l}^{\dagger}$ with power
$\text{Tr}\left(\mathbf{\hat{Q}}_{l}^{'}\right)=\text{Tr}\left(\mathbf{Q}_{l}\right)$
by Lemma \ref{lem:TrQ} applied to a single-user network of link $l$.
Therefore, $\mathbf{\hat{Q}}_{l}^{'}=\mathbf{\hat{Q}}_{l}$ because
$\mathbf{\hat{Q}}_{l}$ also achieves the capacity of the $\bar{\mathbf{H}}_{l}^{\dagger}$
with the same power $\text{Tr}\left(\mathbf{\hat{Q}}_{l}\right)=\text{Tr}\left(\mathbf{Q}_{l}\right)$
and the capacity achieving covariance matrix is unique \cite{Telatar_EuroTrans_1999_MIMOCapacity}.
It follows that $\hat{\mathbf{\Sigma}}_{l}^{'}=\hat{\mathbf{\Sigma}}_{l}$,
which implies the interference from reverse link $l$ to other reverse
links does not change. Then, $\left\{ \mathbf{r}_{l,m}^{'},q_{l,m}^{'}\right\} $
and $\left\{ \mathbf{r}_{k,m},q_{k,m},\forall k\neq l\right\} $ achieves
the same SINRs as the forward links and thus $\left\{ \hat{\mathbf{\Sigma}}_{l}^{'}=\hat{\mathbf{\Sigma}}_{l},\hat{\mathbf{\Sigma}}_{k\ne l}\right\} $
is the result of the covariance transformation and is invariant with
the decomposition of $\mathbf{\Sigma}_{l}$.

Second, we prove (\ref{eq:MtformTF}). In the rest of the proof, the
subscript $l$ will be omitted for simplicity. We will use the notations
in Theorem \ref{thm:WFST}. Noting that $\mathbf{G}\mathbf{\Delta}\mathbf{F}^{\dagger}=\hat{\mathbf{\Omega}}^{-\frac{1}{2}}\mathbf{H}^{\dagger}\mathbf{\Omega}^{-\frac{1}{2}}$,
$\mathbf{F}\mathbf{\Delta}\mathbf{G}^{\dagger}=\mathbf{\Omega}^{-\frac{1}{2}}\mathbf{H}\hat{\mathbf{\Omega}}^{-\frac{1}{2}}$,
$\mathbf{G}\mathbf{\Delta}^{2}\mathbf{G}^{\dagger}=\hat{\mathbf{\Omega}}^{-\frac{1}{2}}\mathbf{H}^{\dagger}\mathbf{\Omega}^{-1}\mathbf{H}\hat{\mathbf{\Omega}}^{-\frac{1}{2}}$,
$\mathbf{F}^{\dagger}\mathbf{F}=\mathbf{I}$ and $\mathbf{G}^{\dagger}\mathbf{G}=\mathbf{I}$,
we have 
\begin{eqnarray*}
 &  & \mathbf{\Delta}\mathbf{D}\mathbf{\Delta}=\mathbf{D}\mathbf{\Delta}^{2}\\
\Rightarrow &  & \hat{\mathbf{\Omega}}^{-\frac{1}{2}}\mathbf{G}\mathbf{\Delta}\mathbf{F}^{\dagger}\mathbf{F}\mathbf{D}\mathbf{F}^{\dagger}\mathbf{F}\mathbf{\Delta}\mathbf{G}^{\dagger}\\
 & = & \hat{\mathbf{\Omega}}^{-\frac{1}{2}}\mathbf{G}\mathbf{D}\mathbf{G}^{\dagger}\mathbf{G}\mathbf{\Delta}^{2}\mathbf{G}^{\dagger}\\
\Rightarrow &  & \hat{\mathbf{\Omega}}^{-\frac{1}{2}}\hat{\mathbf{\Omega}}^{-\frac{1}{2}}\mathbf{H}^{\dagger}\mathbf{\Omega}^{-\frac{1}{2}}\mathbf{F}\mathbf{D}\mathbf{F}^{\dagger}\mathbf{\Omega}^{-\frac{1}{2}}\mathbf{H}\hat{\mathbf{\Omega}}^{-\frac{1}{2}}\\
 & = & \hat{\mathbf{\Omega}}^{-\frac{1}{2}}\mathbf{G}\mathbf{D}\mathbf{G}^{\dagger}\hat{\mathbf{\Omega}}^{-\frac{1}{2}}\mathbf{H}^{\dagger}\mathbf{\Omega}^{-1}\mathbf{H}\hat{\mathbf{\Omega}}^{-\frac{1}{2}}\\
\Rightarrow &  & \hat{\mathbf{\Omega}}^{-1}\mathbf{H}^{\dagger}\mathbf{\hat{\mathbf{\Sigma}}}\mathbf{H}=\mathbf{\Sigma}\mathbf{H}^{\dagger}\mathbf{\Omega}^{-1}\mathbf{H}.
\end{eqnarray*}
where the last equation follows from (\ref{eq:WFFar}) and (\ref{eq:WFrev})
in Theorem \ref{thm:WFST}.

Finally, we prove (\ref{eq:SigmhDirect}). Expand $\mathbf{F}$ and
$\mathbf{G}$ to full unitary matrices $\tilde{\mathbf{F}}=\left[\mathbf{F}\bar{\mathbf{F}}\right]\in\mathbb{C}^{L_{R}\times L_{R}}$
and $\tilde{\mathbf{G}}=\left[\mathbf{G}\bar{\mathbf{G}}\right]\in\mathbb{C}^{L_{T}\times L_{T}}$.
Zero pad $\mathbf{\Delta}$ and $\mathbf{D}$ to $\mathbf{\tilde{\Delta}}\in\mathbb{C}^{L_{R}\times L_{T}}$
and $\mathbf{\tilde{D}}\in\mathbb{C}^{L_{T}\times L_{T}}$. Then we
have

\begin{align}
 & \nu\left(\mathbf{\Omega}^{-1}-\left(\mathbf{H}\Sigma\mathbf{H}^{\dagger}+\mathbf{\Omega}\right)^{-1}\right)\nonumber \\
= & \nu\mathbf{\Omega}^{-1/2}\left(\mathbf{I}-\left(\bar{\mathbf{H}}\mathbf{Q}\bar{\mathbf{H}}^{\dagger}+\mathbf{I}\right)^{-1}\right)\mathbf{\Omega}^{-1/2}\nonumber \\
= & \nu\mathbf{\Omega}^{-1/2}\left(\mathbf{I}-\left(\mathbf{\tilde{F}}\mathbf{\tilde{\Delta}}\mathbf{\tilde{D}}\mathbf{\tilde{\Delta}}^{T}\mathbf{\tilde{F}}^{\dagger}+\mathbf{I}\right)^{-1}\right)\mathbf{\Omega}^{-1/2}\label{eq:Q-water-filling}\\
= & \nu\mathbf{\Omega}^{-1/2}\mathbf{\tilde{F}}\left(\mathbf{I}-\left(\mathbf{\tilde{\Delta}}\mathbf{\tilde{D}}\mathbf{\tilde{\Delta}}^{T}+\mathbf{I}\right)^{-1}\right)\mathbf{\tilde{F}}^{\dagger}\mathbf{\Omega}^{-1/2}\nonumber \\
= & \mathbf{\Omega}^{-1/2}\mathbf{\tilde{F}}\mathbf{\tilde{D}}\mathbf{\tilde{F}}^{\dagger}\mathbf{\Omega}^{-1/2}=\hat{\mathbf{\Sigma}}\label{eq:Omghsim}
\end{align}
where the Equation (\ref{eq:Q-water-filling}) follows from $\mathbf{Q}=\mathbf{G}\mathbf{D}\mathbf{G}^{\dagger}$
and the Equation (\ref{eq:Omghsim}) follows from $\mathbf{D}=\left(\nu\mathbf{I}-\mathbf{\Delta}^{-2}\right)^{+}$
and $\mathbf{\hat{Q}}=\mathbf{F}\mathbf{D}\mathbf{F}^{\dagger}$.

\subsection{Proof of Theorem \ref{thm:Zconcave}\label{sub:Proof-of-TheoremZconcave}}

The weighted sum-rate function for a two-user Z channel can be expressed
as
\begin{align}
 & f\left(\mathbf{\Sigma}_{1},\mathbf{\Sigma}_{2}\right)\label{eq:ZWSRF}\\
= & w_{1}\textrm{log}\left|\mathbf{I}+\mathbf{H}_{1,1}\mathbf{\Sigma}_{1}\mathbf{H}_{1,1}^{\dagger}+\mathbf{H}_{1,2}\mathbf{\Sigma}_{2}\mathbf{H}_{1,2}^{\dagger}\right|\nonumber \\
 & +w_{1}\left(\textrm{log}\left|\mathbf{I}+\mathbf{H}_{2,2}\mathbf{\Sigma}_{2}\mathbf{H}_{2,2}^{\dagger}\right|-\textrm{log}\left|\mathbf{I}+\mathbf{H}_{1,2}\mathbf{\Sigma}_{2}\mathbf{H}_{1,2}^{\dagger}\right|\right)\nonumber \\
 & +\left(w_{2}-w_{1}\right)\textrm{log}\left|\mathbf{I}+\mathbf{H}_{2,2}\mathbf{\Sigma}_{2}\mathbf{H}_{2,2}^{\dagger}\right|.\nonumber 
\end{align}
Because the sum of concave functions is still concave, and the first
and third terms in (\ref{eq:ZWSRF}) are logdet functions and thus
are concave, we only need to show that the second term 
\[
g\left(\mathbf{\Sigma}_{2}\right)=\textrm{log}\left|\mathbf{I}+\mathbf{H}_{2,2}\mathbf{\Sigma}_{2}\mathbf{H}_{2,2}^{\dagger}\right|-\textrm{log}\left|\mathbf{I}+\mathbf{H}_{1,2}\mathbf{\Sigma}_{2}\mathbf{H}_{1,2}^{\dagger}\right|
\]
is a concave function.

\textcolor{black}{Consider the convex combination of two different
inputs, $\mathbf{X}_{2}\succeq0$ and $\mathbf{Z}_{2}\succeq0$, 
\begin{eqnarray*}
\mathbf{\Sigma}_{2} & = & t\mathbf{Z}_{2}+\left(1-t\right)\mathbf{X}_{2}\\
 & = & \mathbf{X}_{2}+t\mathbf{Y}_{2},
\end{eqnarray*}
where $\mathbf{Y}_{2}=\mathbf{Z}_{2}-\mathbf{X}_{2}$ is Hermitian
and $0\leq t\leq1$. Then $g\left(\mathbf{\Sigma}_{2}\right)$ is
a concave function if and only if $\frac{d^{2}}{dt^{2}}g\left(\mathbf{\Sigma}_{2}\right)\leq0$
for any $\mathbf{X}_{2}\succeq0$, $\mathbf{Z}_{2}\succeq0$ and $0\leq t\leq1$
\cite{Boyd_04Book_Convex_optimization}.}

It can be derived that \cite{Blum_ITS03_OPTsigMIMOIFC}
\begin{eqnarray*}
\frac{d^{2}}{dt^{2}}g\left(\mathbf{\Sigma}_{2}\right) & = & -\textrm{Tr}\left[\mathbf{A}\mathbf{Y}_{2}\mathbf{A}\mathbf{Y}_{2}\right]+\textrm{Tr}\left(\mathbf{B}\mathbf{Y}_{2}\mathbf{B}\mathbf{Y}_{2}\right),
\end{eqnarray*}
where $\mathbf{A}=\left(\left(\mathbf{H}_{2,2}^{\dagger}\mathbf{H}_{2,2}\right)^{-1}+\mathbf{\Sigma}_{2}\right)^{-1}$
and $\mathbf{B}=\left(\left(\mathbf{H}_{1,2}^{\dagger}\mathbf{H}_{1,2}\right)^{-1}+\mathbf{\Sigma}_{2}\right)^{-1}$
are positive definite matrices. If $\mathbf{H}_{2,2}^{\dagger}\mathbf{H}_{2,2}\succeq\mathbf{H}_{1,2}^{\dagger}\mathbf{H}_{1,2}$,
then $\mathbf{A}\succeq\mathbf{B}$ and 
\begin{eqnarray*}
 & \textrm{Tr}\left(\mathbf{A}\mathbf{Y}_{2}\mathbf{A}\mathbf{Y}_{2}\right)\\
= & \textrm{Tr}\left(\mathbf{A}^{-1/2}\mathbf{Y}_{2}\mathbf{A}\mathbf{Y}_{2}\mathbf{A}^{-1/2}\right)\\
\geq & \textrm{Tr}\left(\mathbf{A}^{-1/2}\mathbf{Y}_{2}\mathbf{B}\mathbf{Y}_{2}\mathbf{A}^{-1/2}\right)\\
= & \textrm{Tr}\left(\mathbf{B}^{-1/2}\mathbf{Y}_{2}\mathbf{A}\mathbf{Y}_{2}\mathbf{B}^{-1/2}\right)\\
\geq & \textrm{Tr}\left(\mathbf{B}^{-1/2}\mathbf{Y}_{2}\mathbf{B}\mathbf{Y}_{2}\mathbf{B}^{-1/2}\right)\\
= & \textrm{Tr}\left(\mathbf{B}\mathbf{Y}_{2}\mathbf{B}\mathbf{Y}_{2}\right),
\end{eqnarray*}
which implies that $\frac{d^{2}}{dt^{2}}g\left(\mathbf{\Sigma}_{2}\right)\leq0$
and $g\left(\mathbf{\Sigma}_{2}\right)$ is a concave function.

\subsection{Proof of Lemma \ref{lem:KKTiTree}\label{sub:Proof-of-Lemma-kktitree}}

Because problem (\ref{eq:WSRMP}) is convex, we only need to show
that the conditions in Lemma \ref{lem:KKTiTree} are necessary and
sufficient to satisfy the KKT conditions in (\ref{eq:kktwsr-1}).
The sufficient part is proved in Theorem \ref{thm:ALGP2KKT}. The
rest is to prove that if $\mathbf{\tilde{\Sigma}}_{1:L}$ satisfy
the KKT conditions with the optimal dual variable $\tilde{\mu}$,
it must satisfy the conditions in Lemma \ref{lem:KKTiTree}. For iTree
networks, we have $\mathbf{\Phi}_{k,l}=0$ for $k\ge l$, and thus
the condition $\nabla_{\mathbf{\Sigma}_{l}}L|_{\mathbf{\Sigma}_{1:L}=\mathbf{\tilde{\Sigma}}_{1:L}}=0$
can be solved one by one from $l=1$ to $l=L$. For $l=1$, $\nabla_{\mathbf{\Sigma}_{l}}L|_{\mathbf{\Sigma}_{1:L}=\mathbf{\tilde{\Sigma}}_{1:L}}=0$
can be expressed as
\begin{equation}
\mathbf{\tilde{\hat{\Omega}}}_{1}=\frac{w_{1}}{\tilde{\mu}}\mathbf{H}_{1,1}^{\dagger}\left(\mathbf{\tilde{\Omega}}_{1}+\mathbf{H}_{1,1}\mathbf{\tilde{\Sigma}}_{1}\mathbf{H}_{1,1}^{\dagger}\right)^{-1}\mathbf{H}_{1,1}+\frac{1}{\tilde{\mu}}\mathbf{\Theta}_{1},\label{eq:KeyKKT1}
\end{equation}
where $\mathbf{\tilde{\hat{\Omega}}}_{1}=\mathbf{I}$. Because (\ref{eq:KeyKKT1})
is also the KKT condition of the single-user polite water-filling
problem, the solution $\mathbf{\tilde{\Sigma}}_{1}$ is unique and
must satisfy the polite water-filling structure in Lemma \ref{lem:KKTiTree}.
By Theorem \ref{thm:At-the-boundary}, we have 
\[
\mathbf{\tilde{\hat{\Sigma}}}_{1}=\frac{w_{1}}{\tilde{\mu}}\left(\mathbf{\tilde{\Omega}}_{1}^{-1}-\left(\mathbf{\tilde{\Omega}}_{1}+\mathbf{H}_{1,1}\mathbf{\tilde{\Sigma}}_{1}\mathbf{H}_{1,1}^{\dagger}\right)^{-1}\right).
\]
Then for $l=2$, $\nabla_{\mathbf{\Sigma}_{l}}L|_{\mathbf{\Sigma}_{1:L}=\mathbf{\tilde{\Sigma}}_{1:L}}=0$
can be expressed as
\begin{equation}
\mathbf{\tilde{\hat{\Omega}}}_{2}=\frac{w_{2}}{\tilde{\mu}}\mathbf{H}_{2,2}^{\dagger}\left(\mathbf{\tilde{\Omega}}_{2}+\mathbf{H}_{2,2}\mathbf{\tilde{\Sigma}}_{2}\mathbf{H}_{2,2}^{\dagger}\right)^{-1}\mathbf{H}_{2,2}+\frac{1}{\tilde{\mu}}\mathbf{\Theta}_{2},\label{eq:KeyKKT2}
\end{equation}
where 
\begin{align*}
\mathbf{\tilde{\hat{\Omega}}}_{2} & =\mathbf{I}+\mathbf{H}_{1,2}^{\dagger}\mathbf{\tilde{\hat{\Sigma}}}_{1}\mathbf{H}_{1,2}\\
 & =\mathbf{I}+\mathbf{H}_{1,2}^{\dagger}\frac{w_{1}}{\tilde{\mu}}\left(\mathbf{\tilde{\Omega}}_{1}^{-1}-\left(\mathbf{\tilde{\Omega}}_{1}+\mathbf{H}_{1,1}\mathbf{\tilde{\Sigma}}_{1}\mathbf{H}_{1,1}^{\dagger}\right)^{-1}\right)\mathbf{H}_{1,2}.
\end{align*}
Similarly, the solution $\mathbf{\tilde{\Sigma}}_{2}$ must satisfy
the conditions in Lemma \ref{lem:KKTiTree}. It can be shown by similar
proof that the same is true for $\mathbf{\tilde{\Sigma}}_{l},\: l=3,...,L$.

\section*{Acknowledgment}

The authors would like to thank Ashutosh Sabharwal for encouragement
and helpful discussions, and Giuseppe Caire for discussions on the
state-of-the-art algorithms for multiple linear constraints and detailed
comments on the paper.



\begin{thebibliography}{10}
\providecommand{\url}[1]{#1}
\csname url@samestyle\endcsname
\providecommand{\newblock}{\relax}
\providecommand{\bibinfo}[2]{#2}
\providecommand{\BIBentrySTDinterwordspacing}{\spaceskip=0pt\relax}
\providecommand{\BIBentryALTinterwordstretchfactor}{4}
\providecommand{\BIBentryALTinterwordspacing}{\spaceskip=\fontdimen2\font plus
\BIBentryALTinterwordstretchfactor\fontdimen3\font minus
  \fontdimen4\font\relax}
\providecommand{\BIBforeignlanguage}[2]{{%
\expandafter\ifx\csname l@#1\endcsname\relax
\typeout{** WARNING: IEEEtran.bst: No hyphenation pattern has been}%
\typeout{** loaded for the language `#1'. Using the pattern for}%
\typeout{** the default language instead.}%
\else
\language=\csname l@#1\endcsname
\fi
#2}}
\providecommand{\BIBdecl}{\relax}
\BIBdecl

\bibitem{Carleial_IT75_StrongIFC}
A.~B. Carleial, ``A case where interference does not reduce capacity,''
  \emph{IEEE Trans. Inform. Th.}, vol. 21, no. 1, pp. 569--570, Sept. 1975.

\bibitem{Han_IT81_hankobascheme}
T.~S. Han and K.~Kobayashi, ``A new achievable rate region for the interference
  channel,'' \emph{IEEE Trans. Inform. Th.}, vol. 27, no. 1, pp. 49--60, Jan.
  1981.

\bibitem{Sato_IT81_StrongIFC}
H.~Sato, ``The capacity of the gaussian interference channel under strong
  interference,'' \emph{IEEE Trans. Inform. Th.}, vol. 27, no. 6, pp. 786--788,
  Nov. 1981.

\bibitem{Maddah-Al_IT_MMKforX}
M.~Maddah-Ali, A.~Motahari, and A.~Khandani, ``Communication over {MIMO} {X}
  channels: Interference alignment, decomposition, and performance analysis,''
  \emph{IEEE Transactions on Information Theory}, vol.~54, no.~8, pp.
  3457--3470, Aug. 2008.

\bibitem{Jafar_IT08_DOFMIMOX}
S.~A. Jafar and S.~Shamai, ``Degrees of freedom region for the {MIMO X}
  channel,'' \emph{IEEE Transactions on Information Theory}, vol. 54, No. 1,
  pp. 151--170, Jan. 2008.

\bibitem{Jafar_09IT_DOFXchannel}
V.~Cadambe and S.~Jafar, ``Interference alignment and the degrees of freedom of
  wireless {X} networks,'' \emph{IEEE Transactions on Information Theory},
  vol.~55, no.~9, pp. 3893 --3908, sept. 2009.

\bibitem{Yu_IT04_MIMO_MAC_waterfilling_alg}
W.~Yu, W.~Rhee, S.~Boyd, and J.~Cioffi, ``Iterative water-filling for
  {Gaussian} vector multiple-access channels,'' \emph{IEEE Trans. Info.
  Theory}, vol.~50, no.~1, pp. 145--152, 2004.

\bibitem{Jindal_IT05_IFBC}
N.~Jindal, W.~Rhee, S.~Vishwanath, S.~A. Jafar, and A.~Goldsmith, ``Sum power
  iterative water-filling for multi-antenna {Gaussian} broadcast channels,''
  \emph{IEEE Trans. Info. Theory}, vol.~51, no.~4, pp. 1570--1580, Apr. 2005.

\bibitem{Weiyu_IT06_DualIWF}
W.~Yu, ``Sum-capacity computation for the gaussian vector broadcast channel via
  dual decomposition,'' \emph{IEEE Trans. Inform. Theory}, vol.~52, no.~2, pp.
  754--759, Feb. 2006.

\bibitem{Cadambe_IT08_DOF}
V.~R. Cadambe and S.~A. Jafar, ``Interference alignment and degrees of freedom
  of the k-user interference channel,'' \emph{IEEE Transactions on Information
  Theory}, vol.~54, no.~8, pp. 3425--3441, 2008.

\bibitem{Cadambe_globecom08_DistributedIA}
K.~Gomadam, V.~Cadambe, and S.~Jafar, ``Approaching the capacity of wireless
  networks through distributed interference alignment,'' \emph{IEEE GLOBECOM
  '08}, pp. 1--6, 30 2008-Dec. 4 2008.

\bibitem{Liu_Allerton09_Duality_distributed}
A.~Liu, A.~Sabharwal, Y.~Liu, H.~Xiang, and W.~Luo, ``Distributed {MIMO}
  network optimization based on duality and local message passing,'' in
  \emph{Proc. Allerton Conf. on Commun., Control, and Computing}, Sep. 2009,
  pp. 1--8.

\bibitem{Costa_IT83_Dirty_paper}
M.~Costa, ``Writing on dirty paper (corresp.),'' \emph{IEEE Trans. Info.
  Theory}, vol.~29, no.~3, pp. 439--441, 1983.

\bibitem{Kramer_IT07sub_bound_Interference_channel}
\BIBentryALTinterwordspacing
X.~Shang, G.~Kramer, and B.~Chen, ``A new outer bound and the
  noisy-interference sum-rate capacity for gaussian interference channels,''
  \emph{submitted to IEEE Trans. Info. Theory}, 2007. [Online]. Available:
  \url{http://arxiv.org/abs/0712.1987}
\BIBentrySTDinterwordspacing

\bibitem{Motahari_IT09_IFCbound}
A.~S. Motahari and A.~K. Khandani, ``Capacity bounds for the gaussian
  interference channel,'' \emph{IEEE Trans. Inform. Th.}, vol. 55, no. 2, pp.
  620--643, Feb. 2009.

\bibitem{Annapureddy_IT09_IFCweek}
V.~S. Annapureddy and V.~V. Veeravalli, ``Gaussian interference networks: Sum
  capacity inthe low interference regime and new outerbounds on the capacity
  region,'' \emph{IEEE Trans. Inform. Th.}, vol. 55, no. 6, pp. 3032--3050,
  June 2009.

\bibitem{Kramer_Allerton08_MIMO_interference}
\BIBentryALTinterwordspacing
X.~Shang, B.~Chen, G.~Kramer, and H.~V. Poor, ``On the capacity of mimo
  interference channels,'' in \emph{submitted to Allerton 2008}, 2008.
  [Online]. Available: \url{http://arxiv.org/abs/0807.1543}
\BIBentrySTDinterwordspacing

\bibitem{Bandemer_Asilomar08_MISOweekIFC}
B.~Bandemer, A.~Sezgin, and A.~Paulraj, ``On the noisy interference regime of
  the {MISO} gaussian interference channel,'' \emph{in Proc. 42nd Asilomar
  Conference on Systems, Signals and Computers, Pacific Grove, CA}, pp.
  1098--1102, 2008.

\bibitem{Annapureddy_Asilomar08_MIMOweekIFC}
V.~S. Annapureddy, V.~V. Veeravalli, and S.~Vishwanath, ``On the sum capacity
  of mimo interference channel in the low interference regime,'' \emph{in Proc.
  42nd Asilomar Conference on Systems, Signals and Computers, Pacific Grove,
  CA}, Nov 2008.

\bibitem{Jafar_IT07_DOFMIMOIFC}
S.~Jafar and M.~Fakhereddin, ``Degrees of freedom for the mimo interference
  channel,'' \emph{IEEE Transactions on Information Theory,}, vol.~53, no.~7,
  pp. 2637 --2642, jul. 2007.

\bibitem{TSE_IT07_Infwith1bit}
R.~Etkin, D.~Tse, and H.~Wang, ``Gaussian interference channel capacity to
  within one bit,'' \emph{IEEE Transactions on Information Theory}, vol.~54,
  no.~12, pp. 5534--5562, Dec. 2008.

\bibitem{Telatar_ISIT07_MIMOIFC}
E.~Telatar and D.~Tse, ``Bounds on the capacity region of a class of
  interference channels,'' \emph{in Proc. IEEE International Symposium on
  Information Theory 2007}, vol. 2871-2874, Nice, France, Jun. 2007.

\bibitem{Tse:98}
D.~N.~C. Tse and S.~V. Hanly, ``Multiaccess fading channels--part {I}:
  polymatroid structure, optimal resource allocation and throughput
  capacities,'' \emph{IEEE Trans. Info. Theory}, vol.~44, no.~7, pp.
  2796--2815, Nov. 1998.

\bibitem{Viswanathan_JSAC03_BCGD}
H.~Viswanathan, S.~Venkatesan, and H.~Huang, ``Downlink capacity evaluation of
  cellular networks with known-interference cancellation,'' \emph{IEEE J.
  Select. Areas Commun.}, vol.~21, no.~5, pp. 802--811, June 2003.

\bibitem{Kobayashi_JSAC06_ITWMISOBC}
M.~Kobayashi and G.~Caire, ``An iterative water-filling algorithm for maximum
  weighted sum-rate of {Gaussian MIMO-BC},'' \emph{IEEE J. Select. Areas
  Commun.}, vol.~24, no.~8, pp. 1640--1646, Aug. 2006.

\bibitem{Yu_JSAC02_Distributed_power_control_DSL}
W.~Yu, G.~Ginis, and J.~Cioffi, ``{Distributed multiuser power control for
  digital subscriber lines},'' \emph{IEEE J. Select. Areas Commun.}, vol.~20,
  no.~5, pp. 1105--1115, 2002.

\bibitem{Popescu_Globecom03_Water_filling_not_good}
O.~Popescu and C.~Rose, ``{Water filling may not good neighbors make},'' in
  \emph{Proceedings of GLOBECOM 2003}, vol.~3, 2003, pp. 1766--1770.

\bibitem{Lai_IT08_water_filling_game_MAC}
L.~Lai and H.~El~Gamal, ``{The water-filling game in fading multiple-access
  channels},'' \emph{IEEE Transactions on Information Theory}, vol.~54, no.~5,
  pp. 2110--2122, 2008.

\bibitem{Cioffi_03ISIT_GaussianIFCgame}
S.~T. Chung, S.~J. Kim, J.~Lee, and J.~Cioffi, ``A game-theoretic approach to
  power allocation in frequency-selective gaussian interference channels,''
  \emph{in Proceedings 2003 IEEE International Symposium on Information
  Theory,}, pp. 316 -- 316, june-4 july Jun. 2003.

\bibitem{Tse_07JSAC_specturmgame}
R.~Etkin, A.~Parekh, and D.~Tse, ``Spectrum sharing for unlicensed bands,''
  \emph{IEEE Journal on Selected Areas in Communications}, vol.~25, no.~3, pp.
  517 --528, april 2007.

\bibitem{Shum_07JSAC_GaussianIFCwaterfilling}
K.~Shum, K.-K. Leung, and C.~W. Sung, ``Convergence of iterative waterfilling
  algorithm for gaussian interference channels,'' \emph{IEEE Journal on
  Selected Areas in Communications}, vol.~25, no.~6, pp. 1091 --1100, august
  2007.

\bibitem{Cendrillon_TSP07_DSL}
R.~Cendrillon, J.~Huang, M.~Chiang, and M.~Moonen, ``Autonomous spectrum
  balancing for digital subscriber lines,'' \emph{IEEE Transactions on Signal
  Processing}, vol.~55, no.~8, pp. 4241 --4257, aug. 2007.

\bibitem{Palomar_08IT_GaussianIFCwaterfiling}
G.~Scutari, D.~Palomar, and S.~Barbarossa, ``Asynchronous iterative
  water-filling for gaussian frequency-selective interference channels,''
  \emph{IEEE Transactions on Information Theory}, vol.~54, no.~7, pp. 2868
  --2878, july 2008.

\bibitem{Blum_ITS03_OPTsigMIMOIFC}
S.~Ye and R.~Blum, ``Optimized signaling for mimo interference systems with
  feedback,'' \emph{IEEE Transactions on Signal Processing}, vol.~51, no.~11,
  pp. 2839--2848, Nov 2003.

\bibitem{Arslan_TWC07_ImproGameIFC}
G.~Arslan, M.~Demirkol, and Y.~Song, ``Equilibrium efficiency improvement in
  mimo interference systems: A decentralized stream control approach,''
  \emph{IEEE Transactions on Wireless Communications}, vol.~6, no.~8, pp.
  2984--2993, August 2007.

\bibitem{Palomar_JSAC08_MIMOIFCgame}
G.~Scutari, D.~Palomar, and S.~Barbarossa, ``Competitive design of multiuser
  {MIMO} systems based on game theory: A unified view,'' \emph{IEEE Journal on
  Selected Areas in Communications}, vol.~26, no.~7, pp. 1089--1103, September
  2008.

\bibitem{Palomar_09TSP_MIMOIWFgame}
------, ``The {MIMO} iterative waterfilling algorithm,'' \emph{IEEE Trans. on
  Signal Processing}, vol.~57, pp. 1917--1935, May 2009.

\bibitem{Huang_06ISIT_SpectrumIFC}
J.~Huang, R.~Cendrillon, M.~Chiang, and M.~Moonen, ``Autonomous spectrum
  balancing (asb) for frequency selective interference channels,'' \emph{in
  Proceedings 2006 IEEE International Symposium on Information Theory}, pp. 610
  --614, july 2006.

\bibitem{Wei_07ITW_MultiuserWF}
W.~Yu, ``Multiuser water-filling in the presence of crosstalk,''
  \emph{Information Theory and Applications Workshop, San Diego, CA, U.S.A},
  pp. 414 --420, 29 2007-feb. 2 2007.

\bibitem{Yu_06Globecom_DistributedSensorNet}
J.~Yuan and W.~Yu, ``Distributed cross-layer optimization of wireless sensor
  networks: A game theoretic approach,'' \emph{in Global Telecommunications
  Conf. (GLOBECOM), San Francisco, U.S.A.}, 2006.

\bibitem{Berry_JSAC06_IfPriceSISO}
J.~Huang, R.~A. Berry, and M.~L. Honig, ``Distributed interference compensation
  for wireless networks,'' \emph{IEEE J. Sel. Areas Commun.}, vol. 24, no. 5,
  pp. 1074--1084, May 2006.

\bibitem{Berry_MonoIFCpricing_ISIT09}
C.~Shi, R.~A. Berry, and M.~L. Honig, ``Monotonic convergence of distributed
  interference pricing in wireless networks,'' \emph{in Proc. IEEE ISIT, Seoul,
  Korea}, June 2009.

\bibitem{Berry_MILCOM09_MIMOIFprice}
------, ``Local interference pricing for distributed beamforming in mimo
  networks,'' \emph{in Proc. IEEE MILCOM, Boston, MA}, Oct. 2009.

\bibitem{Rashid:98}
F.~Rashid-Farrokhi, K.~Liu, and L.~Tassiulas, ``Transmit beamforming and power
  control for cellular wireless systems,'' \emph{IEEE J. Select. Areas
  Commun.}, vol. Vol. 16, no.~8, pp. 1437--1450, Oct. 1998.

\bibitem{Visotski__VTC99_SIMODual}
E.~Visotski and U.~Madhow, ``Optimum beamforming using transmit antenna
  arrays,'' \emph{in Proc. IEEE VTC, Houston, TX}, vol.~1, pp. 851--856, May
  1999.

\bibitem{Madhow_VTC99_LOSsinrdual}
E.~Visotsky and U.~Madhow, ``Optimum beamforming using transmit antenna
  arrays,'' \emph{in Proc. IEEE VTC, Houston, TX}, vol.~1, pp. 851--856, May
  1999.

\bibitem{Chang_TWC02_BFalgduality}
J.-H. Chang, L.~Tassiulas, and F.~Rashid-Farrokhi, ``Joint transmitter receiver
  diversity for efficient space division multiaccess,'' \emph{IEEE Transactions
  on Wireless Communications}, vol.~1, no.~1, pp. 16--27, Jan 2002.

\bibitem{Martin_ITV_04_BFdual}
M.~Schubert and H.~Boche, ``Solution of the multiuser downlink beamforming
  problem with individual {SINR} constraints,'' \emph{IEEE Transactions on
  Vehicular Technology}, vol.~53, no.~1, pp. 18--28, Jan. 2004.

\bibitem{Gan_06_GLOBECOM_SPGP}
G.~Zheng, T.-S. Ng, and K.-K. Wong, ``Joint power control and beamforming for
  sum-rate maximization in multiuser {MIMO} downlink channels,'' \emph{IEEE
  GLOBECOM '06.}, pp. 1--5, 27 2006-Dec. 1 2006.

\bibitem{Boche_CISS07_Weighted_sum_rate}
S.~Shi, M.~Schubert, and H.~Boche, ``{Weighted sum-rate optimization for
  multiuser MIMO systems},'' in \emph{IEEE Conference on Information Sciences
  and Systems (CISS)}, 2007, pp. 425--430.

\bibitem{Martin_ITS_07_MMSEduality}
------, ``Downlink {MMSE} transceiver optimization for multiuser {MIMO}
  systems: Duality and sum-{MSE} minimization,'' \emph{IEEE Transactions on
  Signal Processing}, vol.~55, no.~11, pp. 5436--5446, Nov. 2007.

\bibitem{Rao_TOC07_netduality}
B.~Song, R.~Cruz, and B.~Rao, ``Network duality for multiuser {MIMO}
  beamforming networks and applications,'' \emph{IEEE Trans. Commun.}, vol.~55,
  no.~3, pp. 618--630, March 2007.

\bibitem{Goldsmith_IT03_SISO_broadcast_min_rate_power_control}
N.~Jindal and A.~Goldsmith, ``Capacity and optimal power allocation for fading
  broadcast channels with minimum rates,'' \emph{IEEE Trans. Info. Theory},
  vol.~49, no.~11, pp. 2895--2909, 2003.

\bibitem{Goldsmith_IT03_MIMO_broadcast_sum_cap}
S.~Vishwanath, N.~Jindal, and A.~Goldsmith, ``Duality, achievable rates, and
  sum-rate capacity of {Gaussian} {MIMO} broadcast channels,'' \emph{IEEE
  Trans. Info. Theory}, vol.~49, no.~10, pp. 2658--2668, 2003.

\bibitem{Tse_IT03_MIMO_broadcast}
P.~Viswanath and D.~Tse, ``Sum capacity of the vector {Gaussian} broadcast
  channel and uplink-downlink duality,'' \emph{IEEE Trans. Info. Theory},
  vol.~49, no.~8, pp. 1912--1921, 2003.

\bibitem{Shamai_ISIT04_Broadcast_capacity_region}
H.~Weingarten, Y.~Steinberg, and S.~S. (Shitz), ``The capacity region of the
  {Gaussian} {MIMO} broadcast channel,'' in \emph{IEEE International Symposium
  on Information Theory (ISIT 04)}, Chicago, IL, USA, June 2004, p. 174.

\bibitem{Zhang_IT08_MACBC_LC}
L.~Zhang, R.~Zhang, Y.~Liang, Y.~Xin, and H.~V. Poor, ``On gaussian {MIMO}
  {BC-MAC} duality with multiple transmit covariance constraints,''
  \emph{submitted to IEEE Trans. on Information Theory}, Sept. 2008.

\bibitem{Yu_IT06_Minimax_duality}
W.~Yu, ``{Uplink-downlink duality via minimax duality},'' \emph{IEEE Trans.
  Info. Theory}, vol.~52, no.~2, pp. 361--374, 2006.

\bibitem{Caire_09ISIT_BC_linear_constraints}
H.~Huh, H.~Papadopoulos, and G.~Caire, ``{MIMO Broadcast Channel Optimization
  under General Linear Constraints},'' in \emph{Proc. IEEE Int. Symp. on Info.
  Theory (ISIT)}, 2009.

\bibitem{Caire_09sTSP_BC_intercell_interference}
H.~Huh, H.~C. Papadopoulos, and G.~Caire, ``Multiuser {MISO} transmitter
  optimization for intercell interference mitigation,'' \emph{IEEE Transactions
  on Signal Processing}, vol.~58, no.~8, pp. 4272 --4285, Aug. 2010.

\bibitem{Xiaohu_IT09_MISOsingledet}
X.~Shang, B.~Chen, and H.~V. Poor, ``Multi-user miso interference channels with
  single-user detection: Optimality of beamforming and the achievable rate
  region,'' \emph{Submitted to IEEE Trans. Inform. Th.}, July 2009.

\bibitem{Tse_ITW07_Deterministic_model_for_relay}
A.~Avestimehr, S.~Diggavi, and D.~Tse, ``{A Deterministic Model for Wreless
  Relay Networks an its Capacity},'' \emph{2007 IEEE Information Theory
  Workshop on Information Theory for Wireless Networks}, pp. 1--6, 2007.

\bibitem{Jafar_IT08_DOF}
V.~R. Cadambe and S.~A. Jafar, ``Interference alignment and degrees of freedom
  of the k-user interference channel,'' \emph{IEEE Transactions on Information
  Theory}, vol.~54, no.~8, pp. 3425--3441, 2008.

\bibitem{Hjorungnes_TSP07_ComplexDiff}
A.~Hjorungnes and D.~Gesbert, ``Complex-valued matrix differentiation:
  Techniques and key results,'' \emph{IEEE Transactions on Signal Processing},
  vol.~55, no.~6, pp. 2740--2746, June 2007.

\bibitem{Maddah-Ali_ISIT06_MMKforX}
M.~Maddah-Ali, A.~Motahari, and A.~Khandani, ``Signaling over {MIMO} multi-base
  systems: Combination of multi-access and broadcast schemes,'' \emph{ISIT
  '06}, pp. 2104--2108, July 2006.

\bibitem{Telatar_EuroTrans_1999_MIMOCapacity}
E.~Telatar, ``Capacity of multi-antenna gaussian channels,'' \emph{Europ.
  Trans. Telecommu.}, vol.~10, pp. 585--595, Nov./Dec. 1999.

\bibitem{Cover_book91_information_theory}
T.~M. Cover and J.~A. Thomas, \emph{Elements of Information Theory}.\hskip 1em
  plus 0.5em minus 0.4em\relax New York: John Wiley and Sons, 1991.

\bibitem{Varanasi_Asilomar97_MMSE_is_optimal}
M.~Varanasi and T.~Guess, ``Optimum decision feedback multiuser equalization
  with successive decoding achieves the total capacity of the gaussian
  multiple-access channel,'' in \emph{Proc. Thirty-First Asilomar Conference on
  Signals, Systems and Computers}, vol.~2, 1997, pp. 1405--1409.

\bibitem{Joham_Globe08_Gduality}
R.~Hunger and M.~Joham, ``A general rate duality of the {MIMO} multiple access
  channel and the {MIMO} broadcast channel,'' \emph{IEEE GLOBECOM 2008}, pp.
  1--5, 30 2008-Dec. 4 2008.

\bibitem{Boyd_04Book_Convex_optimization}
S.~Boyd and L.~Vandenberghe, \emph{Convex Optimization}.\hskip 1em plus 0.5em
  minus 0.4em\relax Cambridge University Press, 2004.

\bibitem{Zhang_08sIT_BC_MAC_duality_multiple_constraints}
L.~Zhang, R.~Zhang, Y.~Liang, Y.~Xin, and H.~Poor, ``On {Gaussian MIMO BC-MAC}
  duality with multiple transmit covariance constraints,'' \emph{IEEE Trans.
  Info. Theory}, submitted, 2008.

\bibitem{Jindal_steepestGradient_matlab}
\BIBentryALTinterwordspacing
J.~Lee and N.~Jindal, ``Symmetric capacity code,'' May 2006. [Online].
  Available: \url{http://www.ece.umn.edu/~nihar/symmetric_cap_code.html}
\BIBentrySTDinterwordspacing

\bibitem{Jindal_05online_Waterfilling_Aglroithm}
\BIBentryALTinterwordspacing
N.~Jindal, W.~Rhee, S.~Vishwanath, S.~Jafar, and A.~Goldsmith, ``Sum power
  iterative waterfilling code,'' June 2005. [Online]. Available:
  \url{http://www.ece.umn.edu/~nihar/iterative_wf_code.html}
\BIBentrySTDinterwordspacing

\bibitem{Liu_10sTSP_Fairness_rate_polit_WF}
\BIBentryALTinterwordspacing
A.~Liu, Y.~Liu, H.~Xiang, and W.~Luo, ``{MIMO B-MAC} interference network
  optimization under rate constraints by polite water-filling and duality,''
  \emph{IEEE Trans. Signal Processing}, vol.~59, no.~1, pp. 263 --276, Jan.
  2011. [Online]. Available: \url{http://arxiv.org/abs/1007.0982}
\BIBentrySTDinterwordspacing

\bibitem{Liu_GC10_DistIWF}
------, ``Distributed and iterative polite water-filling for {MIMO B-MAC}
  networks,'' \emph{submitted to IEEE Globecom 2010}, 2010.

\bibitem{Ashu_IT09_MGPS}
V.~Aggarwal, Y.~Liu, and A.~Sabharwal, ``Sum-capacity of interference channels
  with a local view,'' \emph{submitted to IEEE Trans. Inf. Theory}, Oct., 2009.

\bibitem{Liu_IT05_grassmann_bounds}
\BIBentryALTinterwordspacing
W.~Dai, Y.~Liu, and B.~Rider, ``Quantization bounds on {G}rassmann manifolds
  and applications to {MIMO} communications,'' \emph{IEEE Trans. Info. Theory},
  vol.~54, no.~3, pp. 1108--1123, March 2008. [Online]. Available:
  \url{http://arxiv.org/abs/cs.IT/0603039}
\BIBentrySTDinterwordspacing

\bibitem{Liu_IT05_power_on_off_feedback}
\BIBentryALTinterwordspacing
W.~Dai, Y.~Liu, V.~K.~N. Lau, and B.~Rider, ``On the information rate of {MIMO}
  systems with finite rate channel state feedback using power on/off
  strategy,'' \emph{IEEE Trans. Info. Theory}, vol.~55, no.~11, pp. 5032 --
  5047, Nov. 2009. [Online]. Available:
  \url{http://arxiv.org/abs/cs.IT/0603040}
\BIBentrySTDinterwordspacing

\bibitem{Richardson_IT01_design_LDPC}
T.~Richardson, M.~Shokrollahi, and R.~Urbanke, ``Design of capacity-approaching
  irregular low-density parity-check codes,'' \emph{IEEE Trans. Info. Theory},
  vol.~47, no.~2, pp. 619 -- 637, Feb. 2001.

\bibitem{Dimitri_Cvxbook03}
D.~P. Bertsekas, A.~Nedic, and A.~E. Ozdaglar, \emph{Convex Analysis and
  Optimization}.\hskip 1em plus 0.5em minus 0.4em\relax Athena Scientific,
  April, 2003.

\end{thebibliography}
\end{document}